\documentclass[12pt]{article}

\usepackage[margin=1.25in]{geometry}
\usepackage[labelfont=bf,labelsep=space, margin = 1cm]{caption}
\usepackage{graphicx}
\usepackage{authblk}
\usepackage{subcaption}
\usepackage{amssymb}
\usepackage{amsmath}
\usepackage{mathrsfs}
\usepackage{mathtools}
\usepackage{amsthm}
\usepackage{amsmath}
\usepackage{booktabs}
\usepackage{accents}
\usepackage[shortlabels]{enumitem}
\usepackage[hyphens]{url}
\usepackage{float}
\usepackage{setspace}
\usepackage{multirow}
\usepackage{comment}
\usepackage{tikz}
\usepackage{tikz-3dplot}
\usepackage{pgfplots}
\usepackage{listings}
\usepackage{dsfont}
\usepackage{natbib}

\makeatletter
\def\@mb@citenamelist{citep,citet,citetalias,citeauthor,citeyear}
\makeatother

\usepackage[resetlabels]{multibib}
\usepackage[colorlinks=true,linkcolor=black!50!blue,citecolor=black!50!blue]{hyperref}

\numberwithin{equation}{section}
\newcites{appendix}{References (Appendix)}

\usetikzlibrary{shapes.multipart,matrix,calc,positioning,arrows,trees,patterns,intersections,shapes.geometric,decorations.pathreplacing}

\newcommand{\beq}{\begin{equation}}
\newcommand{\eeq}{\end{equation}}
\newcommand{\ba}{\begin{enumerate}[(a)]}
\newcommand{\bi}{\begin{enumerate}[(i)]}
\newcommand{\bn}{\begin{enumerate}[(1)]}
\newcommand{\ee}{\end{enumerate}}
\newcommand{\ve}{\varepsilon}

\DeclareMathOperator*{\argmin}{argmin}

\theoremstyle{definition}

\newtheorem{assumptiona}{Assumption}
\newtheorem{proposition}{Proposition}
\newtheorem{remark}{Remark}

\newtheorem{corollary}{Corollary}

\newlist{steps}{enumerate}{1}
\setlist[steps,1]{leftmargin=*, label = \textbf{Step \arabic*:}}
\newlist{groups}{enumerate}{1}
\setlist[groups,1]{leftmargin=*, label = \textbf{Group \arabic*:}}

\defcitealias{blundell-wp}{Blundell, Kristensen, and Matzkin}

\makeatletter
\newcommand*{\hyperlinkcite}[1]{\hyper@link{cite}{cite.#1\@extra@b@citeb}}
\makeatother

\begin{document}

\title{Consumer Theory with Non-Parametric \\ Taste Uncertainty and Individual Heterogeneity\thanks{The second author gratefully acknowledges financial support from the ACPR Chair (Regulation and Systemic Risks), and the ERC DYSMOIA. We thank Martin Burda, JoonHwan Cho, Rahul Deb, Jiaying Gu, Joann Jasiak, Kyoo il Kim, Yao Luo, Etienne Masson-Makdissi, Ismael Mourifi\'{e}, Jesse Shapiro, Eduardo Souza-Rodrigues, and participants at the \href{http://www.cfenetwork.org/CFE2020/}{$14^{th}$ International Conference on Computational and Financial Econometrics (2020)} and a seminar at the University of Toronto for useful comments. This research was enabled by support from Sci-Net (\href{www.scinethpc.ca}{www.scinethpc.ca}) and Compute Canada (\href{www.computecanada.ca}{www.computecanada.ca}). Researcher(s) own analyses calculated (or derived) based in part on data from The Nielsen Company (US), LLC and marketing databases provided through the Nielsen Datasets at the Kilts Center for Marketing Data Center at The University of Chicago Booth School of Business. The conclusions drawn from the Nielsen data are those of the researcher(s) and do not reflect the views of Nielsen. Nielsen is not responsible for, had no role in, and was not involved in analyzing and preparing the results reported herein. The ACPR is not responsible for any of our results. The conclusions in this paper do not reflect the \mbox{views of the ACPR.}}}
\author{Christopher Dobronyi\thanks{Dobronyi: University of Toronto; christopher.dobronyi@mail.utoronto.ca.} \ \quad and \quad Christian Gouri\'{e}roux\thanks{Gouri\'{e}roux: University of Toronto, Toulouse School of Economics, and CREST.}}

\maketitle

\newpage

\begin{center}
	\phantom{1}\vspace{1cm} \\ \LARGE Consumer Theory with Non-Parametric \\ Taste Uncertainty and Individual Heterogeneity\vspace{0.35cm}
\end{center}

\begin{abstract}
\vspace{0.35cm}
\noindent We introduce two models of non-parametric random utility for demand systems: the stochastic absolute risk aversion (SARA) model, and the stochastic safety-first (SSF) model. In each model, individual-level heterogeneity is characterized by a distribution $\pi\in\Pi$ of taste parameters, and heterogeneity across consumers is introduced using a distribution $F$ over the distributions in $\Pi$. Demand is non-separable and heterogeneity is infinite-dimensional. Both models admit corner solutions. We consider two frameworks for estimation: a Bayesian framework in which $F$ is known, and a hyperparametric (or empirical Bayesian) framework in which $F$ is a member of a known parametric family. Our methods are illustrated by an application to a large U.S. panel of scanner data on alcohol consumption.
\end{abstract}
\vspace{0.5cm}
\noindent \textbf{Keywords:} Consumer Theory, Scanner Data, Stochastic Demand, Taste Heteroge- neity, Non-Parametric Model, Bayesian Approach.

\newpage

\section{Introduction}
\label{sec:1}

The recent availability of databases containing all dated purchases made by a \mbox{large nu-} mber of consumers (28,036 in our application) presents a modern challenge \mbox{for the eco-} nometrics of demand systems, requiring new models and estimation approaches (see, for example, \citeauthor{burda-2008}, \citeyear{burda-2008}, \citeyear{burda-2012}, for discrete choice, and \citeauthor{guha-ng}, \citeyear{guha-ng}, \citeauthor{cher-newey}, \citeyear{cher-newey}, and \citeauthor{dobronyi}, \citeyear{dobronyi}, for the first analyses of such data in the demand literature). This type of data is commonly called \emph{scanner data} because its collection involves retailers or households scanning each purchased good on the date of purchase. This paper introduces two models of random utility for scanner data: the stochastic absolute risk aversion (SARA) model, and the stochastic safety-first (SSF) model. These models have the following advantages in comparison with the existing literature:
\bi
	\item Both models are \emph{consistent with consumer theory}: Every consumer maximizes a strictly increasing and strictly quasi-concave utility function. The \mbox{latter prop-} erty is not accommodated by existing approximations of the utility function like the quadratic approximation of the utility function (\citeauthor{theil-neudecker}, \citeyear{theil-neudecker}; \citeauthor{barten}, \citeyear{barten}), the translog utility model (\citeauthor{johansen-relationship}, \citeyear{johansen-relationship}; \citeauthor{transcendental}, \citeyear{transcendental}), or the Almost Ideal Demand System \citep{aids} and its extensions (\citeauthor{bbl}, \citeyear{bbl}; \citeauthor{moschini}, \citeyear{moschini}).

	\item Both models are non-parametric. In each model, the utility function is indexed by a functional parameter characterizing the individual heterogeneity, allowing for \emph{infinite-dimensional heterogeneity}. In this respect, our paper \mbox{differs from the} existing literature when finite-dimensional heterogeneity is considered (see \citeauthor{beckert-blundell}, \citeyear{beckert-blundell}, \citeauthor{blom-kumar-more}, \citeyear{blom-kumar-more}, \hyperlinkcite{blundell}{Blundell, Horowitz, and Parey}, \hyperlinkcite{blundell}{2017}, and \hyperlinkcite{blundell-wp}{Blundell, Kristensen, and Matzkin}, \hyperlinkcite{blundell-wp}{2017}, for some examples of finite-dimensional restrictions). Our approach is in line with \citet{dette} who write, ``in general the multivariate demand function is a non-monotonic function of an infinite-dimensional unobservable---the individual's preference ordering.''

	\item Both models yield demand functions with \emph{non-separable heterogeneity} (see the discussions in \citeauthor{brown}, \citeyear{brown}, \citeauthor{beckert-blundell}, \citeyear{beckert-blundell}, and \citeauthor{dette}, \citeyear{dette}). They are also endowed with precise \emph{structural interpretations}, as heterogeneity is introduced by means of a distribution $\pi$ of taste parameters, so that we can imagine consumers facing \emph{taste uncertainty}, which they eliminate using expected utility.

	\item Both models are \emph{identified} under weak restrictions. Identification follows from the use of panel data. Without such data, we lose identification (\citeauthor{hn-individual-het}, \citeyear{hn-individual-het}). Of course, the structure of scanner data is extremely important.
\ee

Each model is characterized by a basis of functions. This basis is used to generate a family of utility functions. A distribution is, then, placed over this family. To be precise, we start with a basis of increasing and concave functions. Let $U(x;a)$ denote an element of this basis, where $x$ is a bundle and $a\in \mathscr{A}$ is a finite-dimensional vector of taste parameters. A family of utility functions is generated by taking the convex hull of the basis. Let $U(x;\pi)=\mathbb{E}_{\pi}\big[U(x;a)\big]$ denote an element of this family, where $\pi\in \Pi$ is a distribution on $\mathscr{A}$. This family is indexed by a functional parameter $\pi$, which can be structurally interpreted as taste uncertainty (resolved \mbox{after the consumer} makes her decisions). The heterogeneity across consumers is introduced using a distribution $F$ on the set $\Pi$ of probability distributions $\pi$ on $\mathscr{A}$. Therefore, each model combines uncertainty and heterogeneity: the uncertainty in taste for a given consumer is represented by $\pi$, and the heterogeneity across consumers is captured by $F$.

The paper considers a two-good framework with continuous support for $x$. It is organized as follows: Section \ref{sec:2} introduces the stochastic absolute risk aversion (SARA) model and Section \ref{sec:3} introduces the stochastic safety-first (SSF) model. For each mo- del, we derive conditions on $\Pi$ under which there exists a unique demand system, for each $\pi\in \Pi$. \mbox{In Section \ref{sec:indhetero},} the distribution of heterogeneity $F$ is introduced. When $F$ is known, we obtain a Bayesian framework in which the functional parameter $\pi\in \Pi$ has to be estimated. When $F$ is a member of a known parametric family, indexed by $\theta$, we obtain an empirical Bayesian framework with a hyperparameter $\theta$ that has to be estimated, and a stochastic functional parameter $\pi$ that has to be filtered. In Section \ref{sec:np.ident}, we consider the identification of the taste distribution $\pi$ \emph{within} each model. Next, \mbox{we examine if it is} possible to distinguish \emph{between} stochastic risk aversion and stochastic safety-first. In Section \ref{sec:illustration}, we use the Nielsen Homescan Consumer Panel to illustrate our methodology in an application to the consumption of alcohol. Section \ref{sec:conclusion} concludes. The details of the Dirichlet process are in Appendix \ref{app:dir}; integrability is discussed in Appendix \ref{app:int}; an optimization procedure for filtering the taste distributions $\pi$ after estimating $F$ is in Appendix \ref{app:filter}; details of the data are placed in \mbox{Appendix \ref{app:d}.}

\section{A Model with Stochastic Risk Aversion}
\label{sec:2}

This section introduces the first utility specification that we consider. It first describes the set of utility functions, then derives conditions under which there exists a unique demand system. The taste uncertainty is introduced using risk aversion parameters.

\subsection{The Set of Utility Functions}

There are two goods, denoted $1$ and $2$. Let $\bar{R}=\mathbb{R}_{+}^2$ denote the non-negative orthant with interior $R$. A consumer has preferences over the bundles in $\bar{R}$. Her preferences are summarized by a utility function of the form:
\beq
\label{utility}
	U(x;\pi)=-\mathbb{E}_{\pi}\big[\exp(-A'x)\big],
\eeq
for every $x$ such that $x_1,x_2\geq 0$, where $A=(A_1,A_2)$ is a positive stochastic parameter characterizing the consumer's degrees of absolute risk aversion with respect \mbox{to goods 1} and 2, and $\pi$ is a joint distribution for this pair of stochastic taste parameters. Her preferences are, as a result, contained in a broad family of utility functions, indexed by a functional parameter $\pi$. There are two interpretations of specification \mbox{\eqref{utility}: (i) the} preferences are summarized by a deterministic utility function in the \mbox{convex hull gen-} erated by a parametric family, or (ii) the consumer faces ``taste uncertainty'' and she resolves this uncertainty by using expected utility. We call these preferences \emph{stochastic absolute risk aversion} (SARA) preferences.\footnote{These preferences differ from those used to describe consumer behaviour when facing ambiguity or uncertainty, as in, say, \citet{halevy-feltkamp}.}

If $\pi$ is a point mass at $a=(a_1,a_2)$ such that $a_1,a_2>0$, the stochastic parameters are constant, and $U(x;\pi)$ reduces to $U(x;a)=-\exp(-a'x)$. This function is strictly increasing because we have:\footnote{Here, $>0$ means each component is strictly larger than $0$.}
\beq
	\frac{\partial U(x;a)}{\partial x}=\left[\begin{array}{c}a_1\exp(-a'x) \\ a_2\exp(-a'x)\end{array}\right]>0,
\eeq
at each $x$ such that $x_1,x_2>0$, and concave (although not necessarily strictly concave) because the Hessian associated with the utility function:
\beq
\label{hessian}
	\frac{\partial^2 U(x;a)}{\partial x\partial x'}=-\exp(-a'x)\left(\begin{array}{cc} a_1^2 & a_1a_2 \\ a_1a_2 & a_2^2\end{array}\right),
\eeq
is negative semi-definite, at each $x$ such that $x_1,x_2>0$. This matrix is \mbox{related to a} bivariate measure of absolute risk aversion\footnote{Such a measure can be defined as:
\[
	-\left(\text{diag}\frac{\partial U(x;\pi)}{\partial x}\right)^{-1/2}\frac{\partial^2 U(x;a)}{\partial x\partial x'}\left(\text{diag}\frac{\partial U(x;\pi)}{\partial x}\right)^{-1/2},
\]
where $\text{diag}\frac{\partial U(x;\pi)}{\partial x}$ is the diagonal matrix whose diagonal elements are the first derivatives of $U(x;\pi)$.} (\citeauthor{richard-1975}, \citeyear{richard-1975}; \citeauthor{karni-1979}, \citeyear{karni-1979}, \citeyear{karni-1983}; \citeauthor{grant}, \citeyear{grant}). These properties translate into properties of the more general function: $U(x;\pi)$.

\begin{proposition}
\label{prop:1}
	If preferences are SARA and the consumer's taste distribution $\pi$ is not the mixture of point masses $a,a'\in R$ where $a$ is proportional to $a'$, then the utility \makebox[\textwidth][s]{function $U(x;\pi)$ is strictly increasing with a negative definite Hessian everywhere on $R$.}
\end{proposition}
\begin{proof}
	The utility function $U(x;\pi)$ is strictly increasing on $R$ because:
	\beq
		\frac{\partial}{\partial x}\mathbb{E}_{\pi}\left[U(x;A)\right]=\mathbb{E}_{\pi}\left[\frac{\partial U(x;A)}{\partial x}\right]>0,
	\eeq
	at every $x$ such that $x_1,x_2>0$. Its Hessian is negative definite on $R$ because the sum of two 2-by-2 matrices of rank 1, whose columns are not proportional, \mbox{has full rank.}
\end{proof}

Proposition \ref{prop:1} implies that we have effectively constructed a family of well-behaved utility functions $\{U(x;\pi):\pi\in\Pi\}$ indexed by a functional parameter $\pi$, describing the taste uncertainty, instead of the standard finite-dimensional parameter \mbox{usually con-} sidered in the literature.

Let $g_{\pi}(\cdot)$ denote the function defined by the implicit equation:
\beq
	U(x_1,g_{\pi}(x_1,u);\pi)=u,
\eeq
for every $x_1\geq 0$, and each (attainable) level of utility $u<0$. This implicit equation has a unique solution because $U(x;\pi)$ is strictly increasing on $\bar{R}$. The function $g_{\pi}(\cdot,u)$ is the \emph{indifference curve} associated with the functional parameter $\pi$ and a utility level of $u$---$g_{\pi}(\cdot,u)$ maps every value of $x_1$ to a value of $x_2$ for which $(x_1,x_2)$ attains a utility level of $u$ given $\pi$. The implicit function theorem implies that $g_{\pi}(\cdot)$ is twice-continuo- usly-differentiable with respect to $x_1$ and:
\beq
\label{ode0}
      \frac{\partial g_{\pi}(x_1,u)}{\partial x_1}=-\text{MRS}(x_1,g_{\pi}(x_1,u);\pi),
\eeq
on $R$ where $\text{MRS}(x;\pi)\equiv \frac{\partial U(x;\pi)/\partial x_1}{\partial U(x;\pi)/\partial x_2}$ denotes the \emph{marginal rate of substitution} at $x$---the rate at which the consumer is willing to exchange good 1 for good 2 \mbox{given $x$ and $\pi$.} The indifference curve $g_{\pi}(\cdot,u)$ is strictly convex such that:
\beq
	\frac{\partial^2 g_{\pi}(x_1,u)}{\partial x_1^2}>0,
\eeq
at every $x_1>0$, since the Hessian of $U(x;\pi)$ is negative definite \mbox{everywhere on $R$ (see} Lemma 1 in \citeauthor{dobronyi}, \citeyear{dobronyi}). This property is stronger than the sta- ndard assumption of strict quasi-concavity, which allows this derivative to be zero on a nowhere dense set (\citeauthor{katzner}, \citeyear{katzner}). This distinction is important for what follows.

Note that, after integrating out the taste uncertainty, the absolute risk aversions will depend on the consumption level. For instance, when $A_1$ and $A_2$ are independent with distributions $\pi_1$ and $\pi_2$, the risk aversion for good 1 becomes:
\beq
\label{ara1}
	A_1(x_1)=-\frac{d^2U_1(x_1;\pi_1)/dx_1^2}{dU_1(x_1;\pi_1)/dx_1}=\frac{\mathbb{E}_{\pi_1}[A_1^2\exp(-A_1x_1)]}{\mathbb{E}_{\pi_1}[A_1\exp(-A_1x_1)]},
\eeq
where $U_1(x_1;\pi_1)$ denotes $\mathbb{E}_{\pi_1}[\exp(-A_1x_1)]$, the portion of the utility function $U(x;\pi)$ corresponding to good 1. Clearly, $A_1(x_1)$ depends on $x_1$. Indeed, it is the average of $A_1$ given the following modified density:
\beq
	\frac{A_1\exp(-A_1x_1)}{\mathbb{E}_{\pi_1}[A_1\exp(-A_1x_1)]},
\eeq
with respect to $\pi_1$.

\subsection{The Demand Function}

Let $z\in R$ denote a pair $z=(y,p)$ in which $y$ denotes expenditure and $p$ denotes the price of good 1, both normalized by the price of good 2. The consumer can purchase a bundle $x\in \bar{R}$ if, and only if, $px_1+x_2\leq y$. She chooses a bundle $x\in\bar{R}$ that solves:
\beq
\label{umax}
	\max_{x\in \bar{R}} \; U(x;\pi) \; \; \text{subject to} \; \; px_1+x_2\leq y.
\eeq

Let $X^*(z;\pi)$ denote the solution to:
\beq
\label{unconstrained}
	\max_{x\in\mathbb{R}^2} \; -\mathbb{E}_{\pi}\big[\exp(-A'x)\big] \; \; \text{subject to} \; \; px_1+x_2\leq y.
\eeq
While \eqref{umax} is restricted to bundles in the non-negative orthant, \eqref{unconstrained} allows for neg- ative values. The solution to \eqref{unconstrained} is characterized by the following \mbox{system of first-} order conditions:
\beq
\label{foc}
	\text{MRS}(x;\pi)\equiv\frac{\mathbb{E}_{\pi}\big[A_1\exp(-A'x)\big]}{\mathbb{E}_{\pi}\big[A_2\exp(-A'x)\big]}=p \; \; \text{and} \; \; px_1+x_2-y=0.
\eeq
The first equality says that the marginal rate of substitution equals the relative price $p$. The second equality says that the budget constraint holds with equality. Equivalently, we can solve the equality:
\beq
\label{opt}
	\mathbb{E}_{\pi}\big[(A_1-pA_2)\exp(-(A_1-pA_2)x_1)\exp(-A_2y)\big]=0,
\eeq
for the first component $X_1^*(z;\pi)$, and then use the  budget constraint in \eqref{foc} to solve for $X_2^*(z;\pi)$. As long as $A_1-pA_2$ is not almost surely equal to zero, the first-order partial derivative of the left side of this equality with respect to $x_1$ is strictly negative:
\beq
	-\mathbb{E}_{\pi}\big[(A_1-pA_2)^2\exp(-(A_1-pA_2)x_1)\exp(-A_2y)\big]<0.
\eeq
The function on the left side of \eqref{opt} is, therefore, strictly decreasing in $x_1$, implying that there exists a unique solution $X_1^*(z;\pi)$ to \eqref{opt}, and a unique solution $X^*(z;\pi)$ to \eqref{unconstrained}. If $X^*(z;\pi)$ is in $\bar{R}$, then $X^*(z;\pi)$ coincides with the solution to \eqref{umax}. Else, the solution to \eqref{umax} is on the boundary of $\bar{R}$. Let $X(z;\pi)$ denote the solution to \eqref{umax} given both $z$ and $\pi$. There are three \emph{regimes} of demand in the design space:
\beq
\label{cases}
	X(z;\pi)=
	\begin{cases}
		(0,y)', & \text{if $X_1^*(z;\pi)\leq 0$}, \\
		X^*(z;\pi), & \text{if $0\leq X_1^*(z;\pi)\leq y/p$}, \\
		(y/p,0)', & \text{if $y/p\leq X_1^*(z;\pi)$}.
	\end{cases}
\eeq
Because the utility function $U(x;\pi)$ has strictly convex indifference curves everywhere on $R$, the demand function $X(z;\pi)$ is invertible in the second regime (see Proposition 2 in \citeauthor{dobronyi}, \citeyear{dobronyi}).

\begin{proposition}
\label{prop:2}
	If preferences are SARA and the consumer's taste distribution $\pi$ is not the mixture of point masses $a,a'\in R$ where $a$ is proportional to $a'$, then there exists a unique solution $X(z;\pi)$ to the maximization problem in \eqref{umax} given $z$ and $\pi$, for every $z\in R$, almost surely, for every $\pi$. There are three regimes of demand defined by \eqref{cases}. The resulting demand function $X(z;\pi)$ is invertible in the \mbox{second regime.}
\end{proposition}

As a final remark, let us consider a \emph{risk-neutral} consumer. In particular, \mbox{let us ass-} ume that $A_1$ and $A_2$ tend stochastically to zero, with means that tend to zero so that $\mathbb{E}_{\pi}[A_1]/\mathbb{E}_{\pi}[A_2]$ converges to a non-degenerate $a_0$. By considering the Taylor expansion of utility, it can be shown that these preferences are \mbox{represented by:}
\beq
	x_1+\frac{x_2}{a_0}.
\eeq
This representation is unique up to an increasing transformation. For this risk-neutral consumer, goods are considered to be \emph{perfect substitutes}. It is known that such a consumer will consume only good 1 whenever $p<a_0$, and only \mbox{good 2 whenever $p>a_0$.}

\subsection{Gamma Taste Uncertainty}
\label{example:1}

As an illustration, let us assume that $A_1$ and $A_2$ are independent and that $A_j$ has a Gamma distribution $\gamma(\nu_j,\alpha_j)$ with \emph{degree of freedom} $\nu_j>0$ and \emph{scale factor} $\alpha_j>0$, for $j=1,2$. Under this specification, $\pi=\gamma(\nu_1,\alpha_1)\otimes\gamma(\nu_2,\alpha_2)$, where $\otimes$ denotes the tensor product of distributions. By the Laplace transform of the Gamma distribution:
\beq
	U(x;\pi)=-\left(\frac{\alpha_1}{\alpha_1+x_1}\right)^{\nu_1}\left(\frac{\alpha_2}{\alpha_2+x_2}\right)^{\nu_2}.
\eeq
Under this specification, the absolute risk aversion for good 1 in \eqref{ara1} becomes:
\beq
	A_1(x_1)=\frac{\nu_1}{\alpha_1+x_1},
\eeq
which is hyperbolic in $x_1$. The indifference curve $g_{\pi}(\cdot)$ associated with utility \mbox{level $u$ is:}
\beq
	x_2=g_{\pi}(x_1,u)\equiv
		\alpha_2\left\{\left[-\frac{1}{u}\left(\frac{\alpha_1}{\alpha_1+x_1}\right)^{\nu_1}\right]^{\frac{1}{\nu_2}}-1\right\},
\eeq
for every $x_1\geq 0$ and $u\in(-1,0)$ such that:
\beq
\label{gamma:conditionx1}
	x_1<\alpha_1\left[\left(-\frac{1}{u}\right)^{\frac{1}{\nu_1}}-1\right].
\eeq
It is easily shown that the second derivative of the indifference curve $g_{\pi}(\cdot,u)$ equals:
\beq
	\frac{d^2g_{\pi}(x_1,u)}{dx_1^2}=c\left(\frac{\alpha_1}{\alpha_1+x_1}\right)^{\frac{\nu_1}{\nu_2}+2}>0,
\eeq
for some $c>0$. This inequality confirms that the indifference curve $g_{\pi}(\cdot,u)$ is strictly convex. Furthermore, the MRS is equal to:
\beq
	\text{MRS}(x;\pi)=\frac{\nu_1}{\nu_2}\frac{\alpha_2+x_2}{\alpha_1+x_1}.
\eeq
The unconstrained solution $X_1^*(z;\pi)$ to the first-order condition in \eqref{foc} is equal to:
\beq
	X_1^*(z;\pi)=\frac{\nu_1}{\nu_1+\nu_2}\cdot \frac{y}{p}+\frac{\nu_1\alpha_2}{\nu_1+\nu_2}\cdot\frac{1}{p}-\frac{\nu_2\alpha_1}{\nu_1+\nu_2}.
\eeq
The second component $X_2^*(z;\pi)$ is deduced from the budget constraint in \eqref{foc}. By equation \eqref{cases}, the demand function $X(z;\pi)$ coincides with $X^*(z;\pi)$ over the set $\mathcal{Z}$ of pairs $z$ such that:
\beq
	\min\big\{\nu_1y-p\nu_2\alpha_1+\nu_1\alpha_2,\, \nu_2y+p\nu_2\alpha_1-\nu_1\alpha_2\big\}>0.
\eeq
The three regimes of demand are illustrated in Figure \ref{fig:gammaregimes} in the design space. The strict convexity of the indifference curve $g_{\pi}(\cdot,u)$ on $\mathcal{Z}$ implies that the demand function $X(\cdot;\pi)$ associated with this utility function is invertible on $\mathcal{Z}$.

\begin{figure}
      \centering
      \begin{tikzpicture}[very thick,
                  level 1/.style={sibling distance=20mm},
                  level 2/.style={sibling distance=25mm},
                  level 3/.style={sibling distance=15mm},
                  every circle node/.style={minimum size=1.25mm,inner sep=0mm}]
                  \scriptsize

                  \draw[fill=red!20!white, draw=none] (0,0) -- (2,0) -- (0,2) -- cycle;
                  \draw[fill=blue!20!white, draw=none] (2,0) -- (5,0) -- (5,4) -- (4,4) -- cycle;
          		  \draw[fill=green!20!white, draw=none] (0,2) -- (2,0) -- (4,4) -- (0,4) -- cycle;
                  \draw (0,2) node[left] {$\frac{\nu_1\alpha_2}{\nu_2}$} -- (2,0) node[below] {$\frac{\nu_1\alpha_2}{\nu_2\alpha_1}$} -- (4,4);

                  \node at (0,0) (x1) {};
                  \node at (5.2,0) (x2) {};
                  \draw[->] (x1.center) -- (x2.center) node[right] {$p$};

                  \node at (0,0) (y1) {};
                  \node at (0,4.2) (y2) {};
                  \draw[->] (y1.center) -- (y2.center) node[above] {$y$};

            \end{tikzpicture}
      \caption{\textbf{Regimes for Gamma Taste Uncertainty.} The red region contains all designs $z$ for which $X_1(z;\pi)>0$ and $X_2(z;\pi)=0$; the blue region contains all designs $z$ for which $X_1(z;\pi)=0$ and $X_2(z;\pi)>0$; the green region contains all designs $z$ for which $X_1(z;\pi)>0$ and $X_2(z;\pi)>0$.}
      \label{fig:gammaregimes}
\end{figure}
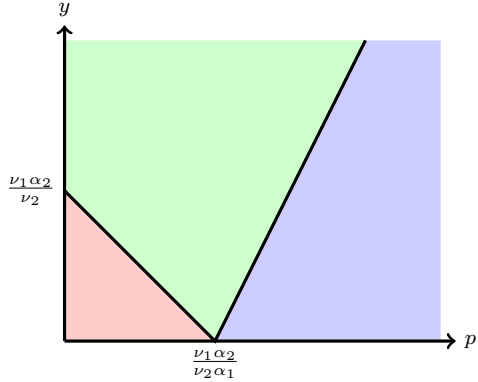

\section{A Model with Stochastic Safety-First}
\label{sec:3}

We now consider a model with taste parameters that have a \mbox{\emph{safety-first} interpretation.}

\subsection{The Set of Utility Functions}

In Section \ref{sec:2}, we constructed a family of well-behaved utility functions by taking the convex hull generated by a particular basis. In this section, we consider another basis, consisting of functions with the form:
\beq
\label{safety:1}
	U(x;a)=(x_1+a_1x_2)-(x_1+a_1x_2-a_2)^+=\min\big\{x_1+a_1x_2,\, a_2\big\},
\eeq
for every $x_1,x_2\geq 0$, where $x^+=\max\{0,x\}$ and $a_1,a_2>0$. This function corresponds to the ``safety-first'' criterion, introduced into the literature on portfolio management by \citet{roy-safety}. In order to illustrate, let us consider the consumption of alcohol, as in \citet{dobronyi}. Suppose that there are two groups of goods: group 1 consisting of drinks with low alcohol by volume such as beers and ciders, and group 2 consisting of drinks with high alcohol by volume such as wines and liquors. Assume that the quantities are measured in identical units such as volume of alcohol---\mbox{that is,} the total volume of the drink in litres multiplied by the alcohol by volume of the drink.\footnote{Quantities could be, alternatively, measured in calories.} We can, then, add these volumes to aggregate two drinks with different sizes and/or percentages of alcohol. Here, $a_1$ is the consumer's relative preference between the two groups of drinks, and $a_2$ is a ``control'' parameter, specifying her attempt to limit her intake of alcohol.

Now, let us introduce a distribution $\pi$ such that $\mathbb{E}_{\pi}[A_j]<\infty$, $j=1,2$, and define:
\beq
	U(x;\pi)=\mathbb{E}_{\pi}\big[(x_1+A_1x_2)-(x_1+A_1x_2-A_2)^+\big].
\eeq
By the law of iterated expectations, we obtain:
\beq
	U(x;\pi)=x_1+\mathbb{E}_{\pi}\big[A_1\big]x_2-\mathbb{E}_{\pi}\mathbb{E}_{\pi}\big[(x_1+A_1x_2-A_2)^+\big|A_1\big].
\eeq
We call these preferences \emph{stochastic safety-first} (SSF) preferences.

Under mild regularity conditions:
\begin{align}
	\frac{\partial U(x;\pi)}{\partial x_1}&=1-\mathbb{E}_{\pi}\mathbb{E}_{\pi}\big[\mathds{1}\{x_1+A_1x_2-A_2> 0\}\big|A_1\big], \\
	&=\mathbb{E}_{\pi}\big[\mathds{1}\{x_1+A_1x_2-A_2< 0\}\big], \\
	\frac{\partial U(x;\pi)}{\partial x_2}&=\mathbb{E}_{\pi}\big[A_1\big]-\mathbb{E}_{\pi}\big[A_1\mathbb{E}_{\pi}\big[\mathds{1}\{x_1+A_1x_2-A_2> 0\}\big|A_1\big]\big], \\ 
	&=\mathbb{E}_{\pi}\big[A_1\mathds{1}\{x_1+A_1x_2-A_2< 0\}\big],
\end{align}
for every $x$ such that $x_1,x_2>0$. These partial derivatives are strictly positive when $\pi$ has full support: $\pi(a_1,a_2)>0$, for $a_1,a_2>0$. By taking the second-order derivatives:
\beq
	\frac{\partial^2 U(x;\pi)}{\partial x\partial x'}=-
	\left(\begin{array}{cc}
		\mathbb{E}_{\pi}[\pi_0]    & \mathbb{E}_{\pi}[A_1\pi_0] \\
		\mathbb{E}_{\pi}[A_1\pi_0] & \mathbb{E}_{\pi}[A_1^2\pi_0]
	\end{array}\right),
\eeq
for every $x$ such that $x_1,x_2>0$, where $\pi_0\equiv\pi(x_1+A_1x_2|A_1)$ in which $\pi(\cdot|A_1)$ denotes the conditional density of $A_2$ given $A_1$, assuming that such a density exists. This matrix is both symmetric and negative definite when $\pi(\cdot|A_1)$ is continuous and $A_1$ is not constant. This result follows from the positivity of $\mathbb{E}_{\pi}[\pi_0]$ and the \mbox{following equality:}
\beq
	\text{det}\frac{\partial^2 U(x;\pi)}{\partial x\partial x'}=\mathbb{E}_{\pi}[\pi_0]V_{\tilde{\pi}}(A_1)>0,
\eeq
which holds for every $x$ such that $x_1,x_2>0$, in which $\tilde{\pi}$ denotes the modified density:
\beq
	\tilde{\pi}(a)=\frac{\pi(x_1+a_2x_2|a_1)\pi(a)}{\mathbb{E}_{\pi}[\pi_0]}.
\eeq

\begin{proposition}
\label{prop:3}
	If preferences are SSF and the consumer's taste distribution \mbox{$\pi$ is con-} tinuous with full support given $A_1$, then the utility function $U(x;\pi)$ is strictly increasing with a negative definite Hessian everywhere on $R$.
\end{proposition}

Consequently, we have constructed another family of well-behaved utility \mbox{functions} $\{U(x;\pi):\pi\in\Pi\}$ indexed by a functional parameter $\pi$, describing taste uncertainty.

\subsection{The Demand Function}

Let us revisit the utility maximization problem in \eqref{umax}. Under the safety-first spec- ification, the analogue of the unconstrained first-order condition in \eqref{opt} \mbox{is given by:}
\beq
\label{foc:ssf}
	\mathbb{E}_{\pi}[(1-pA_1)\mathds{1}\{x_1(1-pA_1)+A_1y-A_2< 0\}]=0.
\eeq
We obtain this equality by equating the marginal rate of substitution with the relative price $p$, and then using the budget constraint to replace $x_2$ with $y-px_1$. Under the regularity conditions from above, the left-hand side is strictly monotone in $x_1$ given $\pi$, so that there exists a unique solution to the first-order condition. As in Section \ref{sec:2}, we let $X_1^*(z;\pi)$ denote this solution, and let $X_2^*(z;\pi)$ denote the quantity $y-pX_1^*(z;\pi)$.

\begin{proposition}
\label{prop:4}
	If preferences are SSF and the consumer's taste \mbox{distribution $\pi$ is con-} tinuous with full support given $A_1$, then there exists a unique solution $X(z;\pi)$ to the maximization problem in \eqref{umax} given $z$ and $\pi$, for every $z\in R$, almost surely, for every $\pi$. There are three regimes of demand defined by \eqref{cases}. The resulting demand function $X(z;\pi)$ is invertible in the second regime.
\end{proposition}

When the consumer's preferences are SSF, the MRS has the form:
\beq
	\text{MRS}(x;\pi)\equiv \frac{\mathbb{E}_{\pi}[\mathds{1}\{x_1+A_1x_2-A_2< 0\}]}{\mathbb{E}_{\pi}[A_1\mathds{1}\{x_1+A_1x_2-A_2< 0\}]}=\frac{1}{\mathbb{E}_{\pi}[A_1|x_1+A_1x_2-A_2<0]}.
\eeq
Thus, the rate at which the consumer is willing to exchange good 1 for good 2 given $x$ and $\pi$  is equal to the inverse of the expectation of her relative preference between goods $A_1$, conditional on not surpassing her control parameter $A_2$.

Some functionals of the distribution $\pi$ can be especially interesting. For instance, in an application to the consumption of alcohol, we might expect the conditional distribution of $A_2$ given $A_1=a_1$ to be concentrated around a single mode, characterizing an implicit alcohol limit for this consumer. Then, we can ask the following questions:
\bi
	\item Is this limit positively correlated with $A_1$? In other words, is there a positive relationship between this limit and a preference for \emph{strong} alcoholic beverages?
	\item Does a change in the maximum blood alcohol level for driving affect this limit?
\ee
These are questions that cannot be answered using classical demand systems like the Almost Ideal Demand System \citep{aids}. In fact, tests based on the Almost Ideal Demand System have rejected rationality in applications to alcohol consumption \citep{all-ferg-stew}. Clearly, it is possible that the Almost Ideal Demand System is misspecified.

\subsection{Exponential Threshold Taste Uncertainty}
\label{example:2}

In general, the first-order condition in \eqref{foc:ssf} has no closed-form solution. \mbox{However, its} expression can be simplified for some taste distributions $\pi$. As an illustration, let us assume that:
\bi
	\item $A_1$ and $A_2$ are independent.
	\item $A_2$ follows an exponential distribution $\gamma(1,\lambda)$ with survival function:
	\beq
		P(A_2>a_2) = \exp(-\lambda a_2).
	\eeq
	\item $A_1$ follows a distribution with Laplace transform: $\Psi(v)=\mathbb{E}[\exp(-vA_1)]$, $v\geq 0$.
\ee
Under this specification, we can first integrate with respect to $A_2$ within the expecta- tion  in \eqref{foc:ssf} in order to obtain the following condition:
\beq
	\mathbb{E}_{\pi}\big[(1-pA_1)\exp\{-\lambda(x_1+(y-x_1p)A_1\}\big]=0.
\eeq
Equivalently, we obtain:
\beq
	\mathbb{E}_{\pi}\big[\exp\{-\lambda (y-x_1p)A_1\}\big]-p\mathbb{E}_{\pi}\big[A_1\exp\{-\lambda (y-x_1p)A_1\}\big]=0.
\eeq
This equation can be written in terms of the Laplace transform $\Psi$ for $A_1$. This yields:
\beq
	\Psi\big[\lambda(y-x_1p)\big]+p\frac{d\Psi}{dv}\big[\lambda(y-x_1p)\big]=0,
\eeq
which can also be written as:
\beq
	\frac{d\log \Psi}{dv}\big[\lambda(y-x_1p)\big]=-\frac{1}{p}.
\eeq
Finally, by inverting this expression and rearranging the terms, we get:
\beq
	X_1^*(z;\pi)=\frac{1}{p}\left[y-\frac{1}{\lambda}\left(\frac{d\log \Psi}{dv}\right)^{-1}\left(-\frac{1}{p}\right)\right],
\eeq
The second component $X_2^*(z;\pi)$ of the unconstrained solution in \eqref{foc:ssf} is deduced from the budget constraint. It follows from equation \eqref{cases} that the demand function $X(z;\pi)$ coincides with $X^*(z;\pi)$ if, and only if:
\beq
	0\leq \frac{1}{\lambda}\left(\frac{d\log \Psi}{dv}\right)^{-1}\left(-\frac{1}{p}\right)\leq y.
\eeq
For instance, if $A_1$ follows a gamma distribution $\gamma(\nu,\alpha)$, then $\log \Psi(v)=-\nu\log(1+v/\alpha)$, and we obtain:
\beq
	\frac{d\log \Psi(v)}{dv}=-\frac{\nu}{\alpha+v},
\eeq
for $v\geq 0$. Moreover, by inverting this function, we get:
\beq
	\left(\frac{d\log \Psi}{dv}\right)^{-1}(\xi)=-\left(\frac{\nu}{\xi}+\alpha\right).
\eeq
Therefore, the solution $X_1^*(z;\pi)$ has the form:
\beq
	X_1^*(z;\pi)=\frac{1}{p}\left[y+\frac{1}{\lambda}(\alpha-\nu p)\right],
\eeq
and demand $X(z;\pi)$ coincides with $X^*(z;\pi)$ if, and only if:
\beq
	0\leq \frac{\nu p-\alpha}{\lambda}\leq y.
\eeq
The regimes of demand are illustrated in Figure \ref{fig:expregimes} in the design space. Note, we can also verify that the \emph{Slutsky coefficient} is strictly negative\footnote{This property holds for any Laplace transform $\Psi$ of $A_1$ (see Appendix \ref{app:int}).} such that:
\beq
	\Delta_x(z)\equiv \frac{\partial X_1(z;\pi)}{\partial p} + X_1(z;\pi)\frac{\partial X_1(z;\pi)}{\partial y}=-\frac{\nu}{\lambda p}<0,
\eeq
ensuring that the demand function $X(\cdot;\pi)$ is invertible over the set $\mathcal{Z}$ of pairs $z$ on which demand is strictly positive (see Section 2 in \citeauthor{dobronyi}, \citeyear{dobronyi}).

\begin{figure}
      \centering
      \begin{tikzpicture}[very thick,
                  level 1/.style={sibling distance=20mm},
                  level 2/.style={sibling distance=25mm},
                  level 3/.style={sibling distance=15mm},
                  every circle node/.style={minimum size=1.25mm,inner sep=0mm}]
                  \scriptsize

                  \draw[fill=red!20!white, draw=none] (0,0) -- (1.5,0) -- (1.5,4) -- (0,4) -- cycle;
                  \draw[fill=blue!20!white, draw=none] (1.5,0) -- (5,0) -- (5,3) -- cycle;
          		  \draw[fill=green!20!white, draw=none] (1.5,0) -- (1.5,4) -- (5,4) -- (5,3) -- cycle;
                  \draw (1.5,4) -- (1.5,0) node[below] {$\frac{\alpha}{\nu}$};
                  \draw (1.5,0) -- (5,3);
                  \node at (4.4,3) {$\frac{\nu p -\alpha}{\lambda}$};

                  \node at (0,0) (x1) {};
                  \node at (5.2,0) (x2) {};
                  \draw[->] (x1.center) -- (x2.center) node[right] {$p$};

                  \node at (0,0) (y1) {};
                  \node at (0,4.2) (y2) {};
                  \draw[->] (y1.center) -- (y2.center) node[above] {$y$};

            \end{tikzpicture}
      \caption{\textbf{Regimes for Exponential Threshold Taste \mbox{Uncertainty.}} The red region contains all designs $z$ for which $X_1(z;\pi) >0$ and $X_2(z;\pi)=0$; the blue region contains all designs $z$ for which $X_1(z;\pi)=0$ and \mbox{$X_2(z;\pi)>0$;} the green region contains all designs $z$ for which $X_1(z;\pi)>0$ and $X_2(z;\pi)>0$.}
      \label{fig:expregimes}
\end{figure}

\section{Individual Heterogeneity}
\label{sec:indhetero}

Sections \ref{sec:2} and \ref{sec:3} introduced two utility specifications, both indexed by the functional parameter $\pi$. Of course, different consumers can have different functional parameters. This individual heterogeneity is introduced in a second layer, by specifying a distribution $F$ over the set $\Pi$ of distributions on $R$, such as the Dirichlet process (see, for example, \citeauthor{navarro}, \citeyear{navarro}, for an application of the Dirichlet process in modelling individual differences). More precisely, we make the following theoretical assumption:

\begin{assumptiona}[Latent Stochastic Model]
	\phantom{1}
	\label{ass:1}
	\bi
		\item There are $n\geq 1$ consumers.
		\item Consumers are segmented into $M$ homogeneous groups.
		\item Consumers in group $m$ have the utility function $U(x;\pi_m)$, for all $m=1,\dots,M$.
		\item The taste parameters $(\pi_m)$ are independently drawn from a Dirichlet \mbox{process $F$.}
	\ee
\end{assumptiona}

Assumption \ref{ass:1} introduces a distribution $F$ over the functional taste parameter $\pi$. This distribution $F$ characterizes the heterogeneity across homogeneous groups. It can encompass, for example, regional or demographic differences in preferences. This infinite-dimensional heterogeneity is non-separable in the stochastic demand equation.

The Dirichlet process can be constructed in three steps:
\begin{steps}
\item Consider the set of (Bernoulli) distributions on $\{0,1\}$. This set is characterized by $q\in \bar{R}$ such that $q_1+q_2=1$. A distribution defined on this set of distributions is a distribution defined on this parameter set. We can, for instance, introduce a \emph{beta distribution}, denoted $B(\alpha_1,\alpha_2)$. The distribution $B(\alpha_1,\alpha_2)$ has a continuous density:
	\beq
		f(q)=\frac{\Gamma(\alpha_1+\alpha_2)q_1^{\alpha_1}q_2^{\alpha_2}}{\Gamma(\alpha_1)\Gamma(\alpha_2)},
	\eeq
	with respect to the Lebesgue measure over the simplex $\{(q_1,q_2)\geq 0:q_1+q_2=1\}$, where $\Gamma$ denotes the gamma function,\footnote{The gamma function $\Gamma$ is defined by $\Gamma(\alpha)=\int_0^{\infty}\exp(-x)x^{\alpha-1}dx$, for each $\alpha>0$.} and $\alpha_1,\alpha_2>0$ are positive scalar parameters.

	\item The beta distribution can be extended to define a distribution on the set of discrete distributions with weights $q_j\geq 0$, $j=1,\dots,J$, such that $\sum_{j=1}^Jq_j=1$. This procedure leads to the \emph{Dirichlet distribution}, denoted $D(\alpha)$. \mbox{The res-} ulting distribution $D(\alpha)$ has continuous density:
	\beq
		f(q)=\frac{\Gamma\big(\sum_{j=1}^J\alpha_j\big)\prod_{j=1}^J q_j^{\alpha_j}}{\prod_{j=1}^J\Gamma(\alpha_j)},
	\eeq
	with respect to the Lebesgue measure over the simplex:
	\beq
		\left\{q\in\mathbb{R}_{+}^J:\text{$\sum_{j=1}^Jq_j=1$ and $q_j\geq 0$, $\forall j$}\right\},
	\eeq
	(see, for example, \citeauthor{kotz}, \citeyear{kotz}, page 485, and \citeauthor{lin-dirichlet}, \citeyear{lin-dirichlet}, for details).

	\item Then, the Dirichlet distribution can be extended to define a distribution on a large set of distributions\footnote{The realizations of a Dirichlet process are, almost surely, discrete distributions. Although we assumed continuity to prove the existence of a unique demand system in Section \ref{sec:3}, these realizations can approximate any continuous distribution. This discrepancy has no practical implications.} defined on $\bar{R}$ (see Appendix \ref{app:dir}). This procedure leads to the \emph{Dirichlet process}. The Dirichlet process is characterized by a distribution $\mu$ on $\bar{R}$ and a scaling parameter $c>0$. The distribution $\mu$ can be thought of as the mean of the Dirichlet process, while the parameter $c$ manages its degree of discretization (see Appendix \ref{app:dir}). This extension of the Dirichlet distribution is much more complicated than the Dirichlet distribution, especially because the notion of the Lebesgue measure on the set of distributions, and the notion of a density, no longer exist (see \citeauthor{ferguson}, \citeyear{ferguson}, \citeauthor{rolin}, \citeyear{rolin}, and \citeauthor{sethuraman}, \citeyear{sethuraman}).
\end{steps}

Let us now discuss implications of Assumption \ref{ass:1}: If the functional and scaling parameters of the Dirichlet process are known, then we are in a Bayesian framework (see, for example, \citeauthor{geweke}, \citeyear{geweke}, for a Bayesian analysis of revealed preference) in which the taste distribution $\pi\in \Pi$ has to be estimated. Otherwise, we can assume that the mean $\mu$ of our process $F$ is characterized by a finite-dimensional hyperparameter $\theta$. Naturally, the hyperparametric model has two types of parameters: the hyperparameter $\theta$ to be estimated, and the functional parameters $(\pi_m)$ to be filtered.

\section{Non-Parametric Identification}
\label{sec:np.ident}

In this section, we consider the identification of the functional parameter $\pi$ \mbox{\emph{within} each} model from the observation of a demand function. Then, we examine if we can distin- guish \emph{between} the SARA and SSF models.

Intuitively, a consumer's demand function is identified if we observe her making a lot of consumption decisions at a variety of designs $z$. Clearly, we can \mbox{identify her} demand function if (i) her preferences are constant over time and we observe a large panel or experiment,\footnote{In this case, when the number of dates $T$ is large, we can have a segment $m$ for each \mbox{consumer $i$.}} or (ii) she belongs to a large homogeneous \mbox{segment of consum-} ers with identical preferences. This explains the form of Assumption \ref{ass:1} (as it allows for either interpretation). Later, we apply the segmented approach to scanner data in the application to the consumption of alcohol in Section \ref{sec:illustration}.

\begin{figure}
      \centering
      \begin{subfigure}[b]{0.48\linewidth}
      \centering
      \begin{tikzpicture}[very thick,
                  level 1/.style={sibling distance=20mm},
                  level 2/.style={sibling distance=25mm},
                  level 3/.style={sibling distance=15mm},
                  every circle node/.style={minimum size=1.25mm,inner sep=0mm}]
                  \scriptsize

                  \draw[red] plot [smooth, tension=0.9] coordinates {(0,0) (1,1) (3,0.5) (5,2)};
                  \draw plot [smooth, tension=0.9] coordinates {(0,0) (1,1.5) (3,1) (5,3)};
                  \draw[blue] plot [smooth, tension=0.9] coordinates {(0,0) (1,2) (3,1.5) (5,4)};

                  \node at (0,0) (x1) {};
                  \node at (5.2,0) (x2) {};
                  \draw[->] (x1.center) -- (x2.center) node[right] {$y$};

                  \node at (0,0) (y1) {};
                  \node at (0,4.2) (y2) {};
                  \draw[->] (y1.center) -- (y2.center) node[above] {$x$};

            \end{tikzpicture}
      \end{subfigure}
      \begin{subfigure}[b]{0.48\linewidth}
      \centering
      \begin{tikzpicture}[very thick,
                  level 1/.style={sibling distance=20mm},
                  level 2/.style={sibling distance=25mm},
                  level 3/.style={sibling distance=15mm},
                  every circle node/.style={minimum size=1.25mm,inner sep=0mm}]
                  \scriptsize

                  \draw[red] plot [smooth, tension=0.9] coordinates {(0,0) (1,1) (3,0.5) (5,2)};
                  \draw plot [smooth, tension=0.9] coordinates {(0,0) (1,1.5) (3,1) (5,3)};
                  \draw[blue] plot [smooth, tension=0.9] coordinates {(0,0) (1,0.5) (3,2.5) (5,3.5)};

                  \node at (0,0) (x1) {};
                  \node at (5.2,0) (x2) {};
                  \draw[->] (x1.center) -- (x2.center) node[right] {$y$};

                  \node at (0,0) (y1) {};
                  \node at (0,4.2) (y2) {};
                  \draw[->] (y1.center) -- (y2.center) node[above] {$x$};

            \end{tikzpicture}
      \end{subfigure}
      \caption{\textbf{Monotonicity in Heterogeneity.} Each figure displays the Engel curves for three consumers at a fixed price. Monotonicity is satisfied on the left. Monotonicity is violated on the right because the curves cross.}
      \label{fig:engel}
\end{figure}
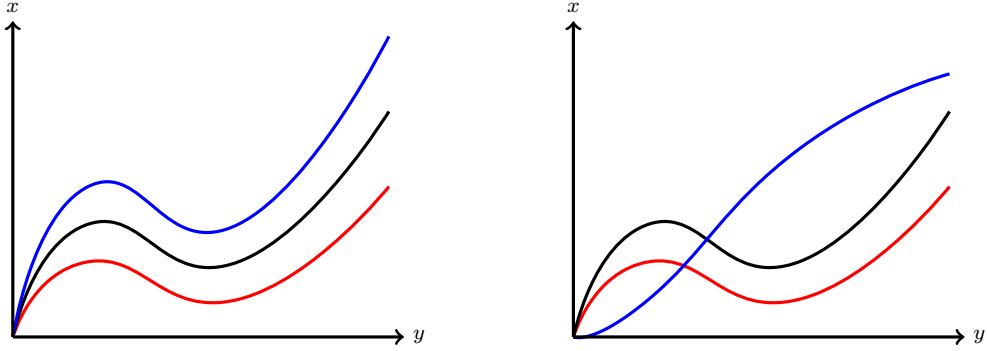

With panel data, one no longer requires the assumption that demand is monotonic with respect to unobserved heterogeneity in order to achieve identification (see Figure \ref{fig:engel}, and the role of this assumption in \citeauthor{brown-matzkin}, \citeyear{brown-matzkin}, \citeauthor{matzkin-nonadd}, \citeyear{matzkin-nonadd}, and \citeauthor{hn-individual-het}, \citeyear{hn-individual-het}).

\subsection{Within Model Identification}
\label{sec:within}

In the models introduced in Sections \ref{sec:2} and \ref{sec:3}, and for any $\pi$ such that \mbox{demand is inv-} ertible, we can derive the inverse demand function, whose second component coincides with the MRS which can be integrated to obtain a unique preference ordering. Indeed, by construction, the integrability conditions (needed to recover a \emph{unique} well-behaved preference ordering) are satisfied, implying that preferences are recoverable (see \citeauthor{samuelson1948}, \citeyear{samuelson1948}, for a seminal discussion of integrability in the case of two goods, and \citeauthor{samuelson}, \citeyear{samuelson}, \citeauthor{hurwicz-uzawa}, \citeyear{hurwicz-uzawa}, and \citeauthor{hosoya2016}, \citeyear{hosoya2016}, for general approaches). However, the possibility to recover preferences from a consumer's demand function does \mbox{not imply that} the distribution of taste uncertainty $\pi$ is identified. Indeed, two distinct taste distributions could produce an identical MRS.

For identification, we only consider the information contained in the demand function $X(\cdot;\pi)$ on the set $\mathcal{Z}$ of designs $z$ for which the components of the demand function are strictly positive. This restriction disregards some information that may be available in the first or third regimes of \eqref{cases}. In most datasets, when a component of the demand function equals zero, the price $p$ is not observed.

\subsubsection{Stochastic Absolute Risk Aversion}
\label{within:SARA}

In the stochastic absolute risk aversion (SARA) model, the identification \mbox{condition is:}
\beq
	\left\{\frac{\mathbb{E}_{\pi}[A_1\exp(-A'x)]}{\mathbb{E}_{\pi}[A_2\exp(-A'x)]}=\frac{\mathbb{E}_{\pi'}[A_1\exp(-A'x)]}{\mathbb{E}_{\pi'}[A_2\exp(-A'x)]}, \; \forall x\in R\right\} \; \; \Rightarrow \; \; \pi=\pi'.
\eeq
In the degenerate case in which $A$ is deterministic and equal to $(a_1,a_2)$, \mbox{the MRS red-} uces to $a_1/a_2$. Thus, in this special case, the two-dimensional parameter $a=(a_1,a_2)$ is identified up to a positive factor. This reasoning leads us to a question: Does this lack of identification also exist in an extended setting?

Let us first remark that the utility function $U(x;\pi)$ is equal to the moment generating function for $\pi$ with a negative sign: $\Phi(x;\pi)=-U(x;\pi)$. Because this moment generating function characterizes $\pi$ when the stochastic parameter $A$ is non-negative (see Theorem 1a in Chapter 13 on Tauberian Theorems in \citeauthor{feller-1968}, \citeyear{feller-1968}), it is equivalent to consider the identification of either $\pi$, or $\Phi(x;\pi)$.\footnote{Note, the existence of the moment generating function does not imply the existence of all power moments and, even if all power moments exist, they do not necessarily characterize the distribution. A known example is the log-normal distribution used in the application \citep{heyde}.} As mentioned, we can always integrate the MRS to recover a unique preference ordering. That is, we can recover $U(x;\pi)$ up to a monotonic transformation. We still need to discern the conditions on $\pi$ under which we can recover $\Phi(x;\pi)$. Indeed, moment generating functions have properties that are not necessarily preserved under monotonic transformations.

We obtain the following result:

\begin{proposition}
	If preferences are SARA, then $\Phi(x;\pi)$ and $\Phi(x;\pi)^{\nu}$ lead to the same preference ordering, for all positive scalars $\nu>0$.
\end{proposition}
\begin{proof}
	Let $U(x;\pi)=-\Phi(x;\pi)$ and $\tilde{U}(x;\pi)=-\Phi(x;\pi)^{\nu}$ denote the utility functions associated with $\Phi(x;\pi)$ and $\Phi(x;\pi)^{\nu}$, respectively. Then, by definition, we must have:
	\beq
		\tilde{U}(x;\pi)=-\Phi(x;\pi)^{\nu}=-(-U(x;\pi))^{\nu}=\phi_{\nu}(U(x;\pi)),
	\eeq
	where $\phi_{\nu}(u)=-(-u)^{\nu}$ is strictly increasing for $u<0$. Since $\tilde{U}(x;\pi)$ is a monotonic transformation of $U(x;\pi)$, these utility functions yield the same preference ordering.
\end{proof}

This means that we can, at most, identify the class of moment generating functions $\mathscr{C}(\Phi)=\{\Phi^{\nu}:\nu>0\}$. Note that, for any moment generating function $\Phi$, the transf- ormed function $\Phi^{\nu}$ is also a moment generating function.

Let us now consider identification when $A_1$ and $A_2$ are independent:

\begin{proposition}
\label{prop:6}
	Let $\Phi_j$ denote the marginal moment generating function for $A_j$, for $j=1,2$. If preferences are SARA, and $A_1$ and $A_2$ are independent, then $(\Phi_1,\Phi_2)$ and $(\Phi_1^*,\Phi_2^*)$ lead to the same preference ordering if, and only if, for some $\nu >0$, we have:
	\[
		\Phi_1^*=\Phi_1^{\nu} \; \; \text{and} \; \; \Phi_2^*=\Phi_2^{\nu}.
	\]
\end{proposition}
\begin{proof}
	The identification criterion becomes:
	\[
		\left(\frac{\partial\Phi_1(x_1)}{\partial x_1}\Phi_2(x_2)\right)\left(\Phi_1(x_1)\frac{\partial\Phi_2(x_2)}{\partial x_2}\right)^{-1}=\left(\frac{\partial\Phi_1^*(x_1)}{\partial x_1}\Phi_2^*(x_2)\right)\left(\Phi_1^*(x_1)\frac{\partial\Phi_2^*(x_2)}{\partial x_2}\right)^{-1},
	\]
	for all $x\in R$. This criterion can, then, be written as:
	\[
		\frac{\partial\log \Phi_1(x_1)}{\partial x_1}\left(\frac{\partial\log \Phi_1^*(x_1)}{\partial x_1}\right)^{-1}=\frac{\partial\log \Phi_2(x_2)}{\partial x_2}\left(\frac{\partial\log \Phi_2^*(x_2)}{\partial x_2}\right)^{-1},
	\]
	for all $x\in R$. Thus, we deduce that, if these distributions yield the same MRS, then:
	\[
		\frac{\partial\log \Phi_j^*(x_j)}{\partial x_j}=\nu\frac{\partial\log \Phi_j(x_j)}{\partial x_j},
	\]
	for some $\nu >0$, at every $x_j\geq 0$, for both $j=1,2$. Because the log-transform of the moment generating function at zero equals zero, by integrating this equation, we get:
	\beq
		\log \Phi_j^*(x_j)=\nu\log\Phi_j(x_j),
	\eeq
	at every $x_j\geq 0$, for both $j=1,2$. Equivalently, $\Phi_1^*=\Phi_1^{\nu}$ and $\Phi_2^*=\Phi_2^{\nu}$.
\end{proof}

Proposition \ref{prop:6} implies that $\mathscr{C}(\Phi)$ is identified under the independence of $A_1$ and $A_2$. Indeed, we can recover the consumer's preference ordering using traditional methods, and use the fact that all admissible preference orderings map to a unique class $\mathscr{C}(\Phi)$.

Of course, independence is a strong restriction. In the SARA model, it is equivalent to the additive separability of the utility function.\footnote{In the case of two goods, additive separability is stronger than separability.} To see this \mbox{result, notice that,} under independence, we obtain:
\beq
	U(x;\pi)=-\mathbb{E}_{\pi}\big[\exp(-A_1x_1)]\cdot\mathbb{E}_{\pi}\big[\exp(-A_2x_2)].
\eeq
Since utility functions are unique up to strictly increasing transformations, this utility function is equivalent to:
\beq
	\tilde{U}(x;\pi)\equiv -\log(-U(x;\pi))=-\log \mathbb{E}_{\pi}\big[\exp(-A_1x_1)]-\log\mathbb{E}_{\pi}\big[\exp(-A_2x_2)],
\eeq
which is an additively separable utility function. In Appendix \ref{app:dependence}, we prove a generalization of Proposition \ref{prop:6} where stochastic taste parameters have a common component. 

\subsubsection{Stochastic Safety-First}

In the SSF model, the identification condition is:
\beq
\begin{gathered}
	\bigg\{\frac{\mathbb{E}_{\pi}\big[A_1\big|x_1+A_1x_2-A_2<0\big]}{\mathbb{E}_{\pi'}\big[A_1\big|x_1+A_1x_2-A_2<0\big]}=1, \; \forall x\in R\bigg\} \; \; \Rightarrow \; \; \pi=\pi'.
\end{gathered}
\eeq
Let us now consider the validity of this condition under an independence assumption. Note, in the SSF model, independence is no longer equivalent to additive separability.

\begin{proposition}
\label{prop:7}
	If preferences are SSF, $A_1$ and $A_2$ are independent, and the marginal distribution of $A_2$ is continuous, then $\mathbb{E}[A_1]$ is identified, and the marginal distribution of $A_2$ is identified up to some positive power transformation of its survival function.
\end{proposition}
\begin{proof}
	In the SSF model, the MRS is identified, and it satisfies:
	\beq
	\label{eq:survival}
		\mathbb{E}_{\pi}[A_1S(x_1+A_1x_2)]=\text{MRS}(x;\pi)\mathbb{E}_{\pi}[S(x_1+A_1x_2)],
	\eeq
	where $S(\cdot)$ denotes the survival function of $A_2$.
	\bi
		\item The expectation $\mathbb{E}_{\pi}[A_1]$ is identified because $\text{MRS}(x_1,0;\pi)=\mathbb{E}_{\pi}[A_1]$.
		\item By differentiating \eqref{eq:survival} with respect to $x_2$, we get:
		\[
		\begin{gathered}
			\mathbb{E}_{\pi}[A_1^2S'(x_1+A_1x_2)]=\text{MRS}(x;\pi)\mathbb{E}_{\pi}[A_1S'(x_1+A_1x_2)] \\
			+\frac{\partial \text{MRS}}{\partial x_2}(x;\pi)\mathbb{E}_{\pi}[S(x_1+A_1x_2)].
		\end{gathered}
		\]
		When $x_2=0$, this equation becomes:
		\[
			S'(x_1)\mathbb{E}_{\pi}[A_1^2]=\text{MRS}(x_1,0;\pi)S'(x_1)\mathbb{E}_{\pi}[A_1]+\frac{\partial \text{MRS}}{\partial x_2}(x_1,0;\pi)S(x_1).
		\]
		By rearranging, we get:
		\[
			\frac{\partial \text{MRS}}{\partial x_2}(x_1,0;\pi)=\frac{S'(x_1)}{S(x_1)}\left(\mathbb{E}_{\pi}[A_1^2]-\text{MRS}(x_1,0;\pi)\mathbb{E}_{\pi}[A_1]\right)=\frac{S'(x_1)}{S(x_1)}V(A_1).
		\]
		Because the partial derivative of the MRS with respect to $x_2$ is identified, the hazard function $\lambda(x_1)=-S'(x_1)/S(x_1)$ of the distribution of $A_2$ is identified up to a positive factor. Since $S(x_1)=\exp\{-\Lambda(x_1)\}$, where $\Lambda(x_1)=\int_0^{x_1}\lambda(t)dt$ is the cumulative hazard function of the distribution of $A_2$, we can identify $S(\cdot)$ up to a positive power transformation.
	\ee
\end{proof}

Proposition \ref{prop:7} provides no information on the identifiability of the distribution of $A_1$ beyond its first moment. It seems difficult to obtain a general identification result, but insights into our identification problem can be obtained by considering the two primary families of distributions that are invariant to positive power transformations, that are, the exponential family and the Pareto family.
\bi
	\item \textbf{Exponential family:} Suppose that the marginal distribution of $A_2$ belongs to the exponential family, and that we have identified its survival function up to a positive power transformation such that $S(x)=\exp\{-cx\}$, for some unknown $c>0$. The MRS in \eqref{eq:survival} becomes:
	\beq
	\label{MRS:exp}
		\text{MRS}(x;\pi)=\frac{\mathbb{E}_{\pi}[A_1\exp\{-cx_2A_1\}]}{\mathbb{E}_{\pi}[\exp\{-cx_2A_1\}]}\equiv G_0(x_2).
	\eeq
	This expression does not depend on $x_1$. Now, let $\Psi(u)=\mathbb{E}_{\pi}[\exp\{-uA_1\}]$ denote the Laplace transform of $A_1$. Under this notation, the equality in \eqref{MRS:exp} implies:
	\[
		G_0(x_2)=\frac{d\log \Psi}{du}(cx_2).
	\]
	Or, equivalently, $G_0(u/c)=d\log \Psi(u)/du$. By integrating, we obtain:
	\[
		\log\Psi(u)=c\big[H(u/c)-H(0)\big],
	\]
	where $H(\cdot)$ is a primitive of the MRS. Therefore:

	\begin{corollary}
	\label{cor:exp}
		Under the conditions of Proposition \ref{prop:7}, if the marginal distribution of $A_2$ belongs to the exponential family, the following results hold:
		\ba
			\item The power transform $c$ is not identified.
			\item The distribution of $A_1$ is identified under an identification restriction on $c$.
		\ee
	\end{corollary}

	We conclude that, under the conditions of Corollary \ref{cor:exp}, the distributions of $A_1$ and $A_2$ are non-parametrically identified up to a single scalar parameter $c>0$.

	\item \textbf{Pareto family:} Let us now examine whether a similar result can be obtained for the Pareto family, in which $S(x)=x^{-\alpha}$, for some $\alpha>0$. The parameter $\alpha$ characterizes the fat tails of the distribution of $A_2$ and the power transformation on the MRS. This survival function produces:
	\[
		\text{MRS}(x;\pi)=\frac{\mathbb{E}_{\pi}[A_1(x_1+A_1x_2)^{-\alpha}]}{\mathbb{E}_{\pi}[(x_1+A_1x_2)^{-\alpha}]}=\frac{\mathbb{E}_{\pi}[A_1(x_0+A_1)^{-\alpha}]}{\mathbb{E}_{\pi}[(x_0+A_1)^{-\alpha}]},
	\]
	where $x_0\equiv x_1/x_2$ denotes a ratio of quantities. Equivalently, we get:
	\beq
	\label{eq:gstar}
		\text{MRS}(x;\pi)=\frac{\mathbb{E}_{\pi}[(x_0+A_1)^{-\alpha+1}]}{\mathbb{E}_{\pi}[(x_0+A_1)^{-\alpha}]}-x_0\equiv G_0(x_0),
	\eeq
	which only depends on the ratio $x_0$. Therefore, we have constructed homothetic preferences. By equation \eqref{eq:gstar}:
	\[
		e(x)\equiv\frac{d}{dx}\log \mathbb{E}_{\pi}[(x+A_1)^{-\alpha+1}],
	\]
	is identified up to a multiplicative constant. Therefore, by integration, $\mathbb{E}_{\pi}[(x+A_1)^{-\alpha+1}]$ is identified up to $\alpha$ and a multiplicative constant $\kappa$. However, as $x$ tends to infinity, this expression is equivalent to $\kappa x^{-\alpha+1}\exp E(x)$, where $E(\cdot)$ is a primitive of $e(\cdot)$. This tail behaviour provides both the identification of $\alpha$ and $\kappa$. This analysis is summarized by the following result:

	\begin{corollary}
	\label{cor:pareto}
		Under the conditions of Proposition \ref{prop:7}, if the marginal distribution of $A_2$ belongs to the Pareto family, the distributions of $A_1$ and $A_2$ are both non-parametrically identified.
	\end{corollary}
\ee

\subsection{Between Model Identification}
\label{sec:betweens}

Once the identification of the consumer's taste distribution $\pi$ within each model is solved, we still need to consider the identification between the models. This analysis is needed to test whether preferences are consistent with SARA, or SSF, or both. It is important to know whether these two classes of preferences are nested or non-nested. If they are non-nested, we need to characterize their intersection and define a general class encompassing both types of preferences.

To illustrate, suppose that the consumer has SSF preferences:
	\beq
	\label{ex:1:1}
		U(x;\pi)=\mathbb{E}_{\pi}\big[\min\big\{x_1+A_1x_2,\, A_2\big\}\big],
	\eeq
	where (i) $A_1$ and $A_2$ are independent, (ii) $A_1$ has distribution $\pi_2$, and (iii) $A_2$ follows an exponential distribution (with unit intensity). Under this specification, we obtain:
	\beq
	\label{ex:1:2}
		U(x;\pi)=1-\mathbb{E}_{\pi}\big[\exp(-x_1-A_1x_2)\big].
	\eeq
	To clarify this result, observe that, by conditioning on $A_1$, we are left with the expec- tation of the minimum of a set containing a constant and a random variable with an exponential distribution. This utility function is a strictly increasing transformation of a SARA utility function:
	\beq
		\tilde{U}(x;\pi)=-\mathbb{E}_{\pi}\big[\exp(-B'x)\big],
	\eeq
	where (i) $B_1$ follows a point mass at $1$, and (ii) $B_2$ has distribution $\pi_2$. Consequently, these utility functions, one SARA, and the other SSF, induce the same preference ordering over the consumption set.

\subsection{Discussion}

The possible lack of identification of each consumer's taste distribution $\pi_m$ has to be taken into account in the economic interpretation of the results. However, it has to be noted that it does not create difficulties for structural inference, where the (scalar or functional) parameters of interest are the parameters characterizing the MRS, rather than the parameters characterizing the utility function.

The lack of identification is due to the special structure of the cone of increasing and concave functions defined on $R$, and of the extremal elements of this cone. For finite increasing concave functions defined on $\mathbb{R}_+$, it is well-known that the extremal functions are of the type:
\beq
	h_0(x)=\min\big\{\alpha_1x+\beta_1, \, \alpha_2 x+\beta_2\big\},
\eeq
in which $(\alpha_j,\beta_j)\in \bar{R}$, for $j=1,2$ (see \citeauthor{blaschke-pick}, \citeyear{blaschke-pick}), and that any finite positive increasing concave function can be written as:
\beq
\label{extremalfunctions}
	b+\mathbb{E}_{\pi}\big[\min(A,x)\big],
\eeq
where $b$ is a positive scalar and $\pi$ is the distribution of $A$. Such functions are charac- terized by $b$ and $\pi$. The set of extremal functions in \eqref{extremalfunctions} is a minimal set of extremal points generating the cone.

Such a property no longer holds for finite positive increasing concave functions defined on $\bar{R}$. \citet{johansen} has described a large set of extremal \mbox{points of the type:}
\beq
	h_1(x)=\min\big\{\alpha_1x+\beta_1, \dots, \alpha_n x+\beta_n\big\},
\eeq
for which $h_1(\cdot)$ induces a covering with vertices of order 3 (see page 62 in \citeauthor{johansen}, \citeyear{johansen}), and has shown that this set is dense in the cone of finite continuous convex functions defined on a convex set in $R$ (see Theorem 2 in \citeauthor{johansen}, \citeyear{johansen}). A minimal set of extremal points generating this cone does not exist. This argument explains why Sections \ref{sec:2} and \ref{sec:3} consider specific convex subsets generated by parametric functions.

While we restrict our attention to SARA and SSF preferences (because the stochastic taste parameters have clear interpretations in these models), other \mbox{convex hulls co-} uld have been considered. For example:
\bi
	\item The convex hull generated by the union of the SARA and SSF models---that is, the smallest structural model containing both of the models in Sections \mbox{\ref{sec:2} and \ref{sec:3}.}
	\item The convex hull generated by a basis of the form:
	\beq
		U(x;a,\nu)=\frac{a_1}{\nu_1}x_1^{\nu_1}+\frac{a_2}{\nu_2}x_2^{\nu_2},
	\eeq
	for every $x\in\bar{R}$ in which $a\in R$ and $\nu\in(0,1)^2$. This basis corresponds to a first-order expansion of a utility function \citep[see][]{johansen-relationship}, and contains a Stone-Geary utility function as a limiting case. Indeed, as $\nu$ approaches zero, we obtain: $U(x;a)=a_1\log x_1+a_2\log x_2$. However, the convex hull generated by this basis is not flexible enough because it only contains weighted combinations of $x_1^{\nu_1}$ and $x_2^{\nu_2}$. Similarly, the convex hull generated by a Stone-Geary basis only contains Stone-Geary utility functions, where the weights are the means of the taste parameters:
	\beq
		U(x;\pi)=\mathbb{E}_{\pi}\big[A_1\log x_1+A_2\log x_2\big] = \mathbb{E}_{\pi}\big[A_1]\log x_1+\mathbb{E}_{\pi}\big[A_2]\log x_2.
	\eeq
\ee

The Stone-Geary basis $U(x;a)$ above can be adjusted to define \mbox{another parametric} basis. In particular, let us apply the transformation $\varphi(x)=-\exp(-x)$ to the Stone-Geary utility function. This transformation yields:
\beq
	\tilde{U}(x;a)\equiv \varphi(U(x;a))=-\frac{1}{x_1^{a_1}x_2^{a_2}}.
\eeq
This utility function forms a well-behaved basis because it is strictly increasing with a negative semi-definite Hessian. While $U(x;a)$ and $\tilde{U}(x;a)$ represent the same preference ordering, they will generate different families due to the strict \mbox{concavity of $\varphi(\cdot)$.} To illustrate, suppose that the stochastic parameters, $A_1$ and $A_2$, are independently distributed with respect to uniform distributions on $[0,1]$. This specification produces:
\beq
	U(x;\pi)=\frac{1}{2}\log x_1 +\frac{1}{2}\log x_2 \; \; \text{and} \; \; \tilde{U}(x;\pi)=-\left(\frac{x_1-1}{x_1 \log x_1 }\right)\left( \frac{x_2-1}{x_2 \log x_2 }\right).
\eeq
While $U(x;\pi)$ is a Stone-Geary utility function, $\tilde{U}(x;\pi)$ is a complicated non-linear function of $x_1$ and $x_2$. Consequently, an uninteresting basis has been \mbox{transformed into} an interesting one. This procedure can be completed for any increasing, concave, and twice-differentiable transformation $\varphi(\cdot)$.

\section{An Illustration}
\label{sec:illustration}

This section shows how to use the SARA and SSF models in a non-parametric framework. First, we specify the statistical model by introducing an assumption on the obs- ervations, and then we discuss statistical inference. The methodology is \mbox{illustrated in} an application to alcohol consumption using scanner data concerning individual purchase histories.

\subsection{Assumptions on Observations}

The behavioural models introduced in the previous sections can be completed with an assumption on the available observations. We consider panel data, indexed by the consumer $i$ and date $t$. After a preliminary treatment of the purchase histories, we have a large number $n$ of consumers and a fixed number $T$ of observed dates. In the preliminary treatment, the goods are aggregated into two groups using a common quantity unit and the dated purchases are aggregated by month (see Section \ref{sec:data}). Recall that, under Assumption \ref{ass:1}, we have $M$ segments of homogeneous consumers.

We introduce the following assumption on the observations:

\begin{assumptiona}[Observations]
\label{ass:2}
\phantom{1}
\bi
	\item We jointly observe $(x_{it},z_{it})$, for all $i=1,\dots,n$ and $t=1,\dots,T$, when $x_{it}>0$.
	\item The individual histories $(x_{it},z_{it})_{t=1}^T$ are independent given all $\pi_m$, $m=1,\dots,M$.
	\item Designs $(z_{it})$ are exogenous (independent of taste distributions $\pi_m$).
\ee
\end{assumptiona}

Assumption \ref{ass:2} describes the structure of the observations. It implies that we can imagine taste parameters $(\pi_m)$ being independently drawn from a Dirichlet process $F$, designs $(z_{it})$ being independently drawn from some distribution, and consumption $x_{it}$ satisfying $x_{it}=X(z_{it};\pi_{m_i})$, where $m_i$ is the group of consumer $i$. Many papers assume that consumption $x_{it}$ is positive (see Section IV.A in \hyperlinkcite{blundell}{Blundell, Horowitz, and Parey}, \hyperlinkcite{blundell}{2017}, for this assumption in an application to gasoline demand, as well as Assumption A5 in \citeauthor{dobronyi}, \citeyear{dobronyi}, for this assumption in an application to the consumption of alcohol); the SARA and SSF models \mbox{allow for corner} solutions. However, in many datasets (including the dataset used in the application in Section \ref{sec:illustration}), there is a problem of partial observability. Let $\tilde{y}$ denote the expenditure (prior to normalization), and let $\tilde{p}_j$ denote the price of good $j$ (prior to normalization). Usually, we only observe the price $\tilde{p}_j$ of a good $j$ when the consumer buys a positive quantity of good $j$. Then, we only observe (normalized) expenditure $y$ when the consumer buys a positive quantity of good 2, and we only observe the (normalized) price $p$ when the consumer buys a positive quantity of both goods (see \citeauthor{craw-pol}, \citeyear{craw-pol}, for an approach to revealed preference that deals with this partial observability problem).\, This problem explains the specific form of Assumption \hyperref[ass:2]{2(i)}.

For deriving the asymptotic properties of estimators, it is also necessary to specify the type of asymptotics to be considered:

\begin{assumptiona}
\label{ass:3}
	Let $n_m$ denote the size of the $m^{th}$ homogeneous group.
	\bi
		\item $n_mT\rightarrow\infty$, as $n\rightarrow\infty$, for all $m=1,\dots,M$.
		\item $n_mT\sim \lambda_m n$, for some $\lambda_m\in(\lambda_{\ell},\lambda_h)$, where $0<\lambda_{\ell}<\lambda_h<1$, for $m=1,\dots,M$.
		\item $M\rightarrow\infty$, as $n\rightarrow\infty$.
	\ee
\end{assumptiona}

Assumptions \hyperref[ass:3]{A3(i)} and \hyperref[ass:3]{A3(ii)} ensure that there are enough observations to non-parametrically estimate the demand function associated with the functional parameter $\pi_m$ on a sufficiently large subset $\mathcal{Z}_m$ of designs $z$. Assumption \hyperref[ass:3]{A3(iii)} guarantees enough filtered parameters $\hat{\pi}_m$ to estimate the underlying Dirichlet process $F$. In some special circumstances, $T$ is large, and Assumption \ref{ass:3} can be used with $m=i$ and $M=n$---that is, a single consumer per group. Otherwise, grouping of homogeneous consumers is needed to identify the demand functions on sufficiently \mbox{large subsets $\mathcal{Z}_m$.}

\subsection{Estimation Method}

The Dirichlet process is common in Bayesian estimation (see, for instance, \citeauthor{ferguson}, \citeyear{ferguson}, and \citeauthor{li-2019}, \citeyear{li-2019}). This process is useful because it is flexible and, if observations are independently and identically drawn from an unknown distribution, the posterior distribution of this distribution has a closed-form expression. However, our framework is much more complicated for two reasons:
\bi
	\item The observed consumption choices $(X_{ijt})$ are not identically distributed because consumers make decisions at \emph{different expenditures and prices}.
	\item It is difficult to derive a closed-form expression for the demand, as a \mbox{function of} the expenditure, the price, and the functional parameter characterizing taste uncertainty $\pi$. It is, therefore, difficult to derive a closed-form \mbox{expression for the} distribution of $X_{ijt}$ conditional on $Z_{it}$.
\ee
These features of our model explain why estimation requires specific numerical algor- ithms. These specific algorithms have to be able to deal with the non-linear \mbox{and high-} dimensional features of the models. In the Bayesian framework, the Dirichlet process is fixed. In the hyperparametric framework, it is parameterized by a vector $\theta$. These parameters have to be estimated and the functional parameters $(\pi_m)$ have to be filter- ed. These estimation approaches are described below.

\subsubsection{Bayesian Framework}
\label{bayesian}

In a pure Bayesian framework, a Dirichlet process is fixed by selecting a mean distribution $\mu$ and a scaling parameter $c$ (see Appendix \ref{app:dir}). This distribution defines the common prior for the functional taste parameters $(\pi_m)$. After, the data are used to compute the posterior distribution for the functional parameters $(\pi_m)$. Under Assu- mption \ref{ass:2}, the posterior distribution can be computed separately for each homogeneous group of consumers:
\[
	\ell(\pi_m|x_{it},z_{it},x_{it}>0,i\in\Lambda_m, t=1,\dots,T),
\]
where $\Lambda_m$ denotes the group of consumers with preferences characterized by the taste parameter $\pi_m$. This approach does not have to account for the potential identification problem discussed in Section \ref{sec:np.ident}. If a specific characteristic of $\pi_m$ is weakly identified, its posterior distribution will be close to the prior distribution.

In our framework, the observations $(x_{it},z_{it})$, conditional on $x_{it}>0$, \mbox{must satisfy} the deterministic first-order conditions implied by the model. These conditions have the following form:
\beq
\label{mrsrestrictions}
	\text{MRS}(x_{it};\pi)=p_{it},
\eeq
for any observed pair $(x_{it},z_{it})$. Equivalently:
\beq
\label{mrsrestrictions2}
	\mathbb{E}_{\pi}\left[\frac{\partial U(x_{it};A)}{\partial x_1}\right]=p_{it}\, \mathbb{E}_{\pi}\left[\frac{\partial U(x_{it};A)}{\partial x_2}\right],
\eeq
for any observed pair $(x_{it},z_{it})$. These conditions are moment restrictions, called \emph{MRS restrictions}. In our big data framework, the number of MRS restrictions is very large, typically several hundred to a thousand. The posterior of $\pi_m$ is simply the distribution of $\pi_m$ given these deterministic restrictions on $\pi_m$. If the taste parameters, $A_1$ and $A_2$, are independent with marginal distributions, $\pi_1$ and $\pi_2$, respectively, then the MRS restrictions are bilinear in $\pi_1$ and $\pi_2$---specifically, these restrictions are linear in $\pi_1$ given $\pi_2$, and linear in $\pi_2$ given $\pi_1$. Later, this property is used to construct a numerically efficient optimization algorithm for filtering all the $\pi_m$ \mbox{(see Appendix \ref{app:filter}).}

\subsubsection{Hyperparametric (or Empirical Bayesian) Framework}
\label{sec:hyper}

The hyperparametric (or empirical Bayesian) framework is a complicated non-linear state-space model with two layers of latent state variables. Such a framework can be characterized as follows:
\bi
	\item \textbf{Deep layer:} Functional parameters $(\pi_m)$, drawn from $F$ (parameterized by $\theta$);
	\item \textbf{Surface layer:} Demand functions $X(\cdot;\pi_m)$ deduced from $\pi_m$;
	\item \textbf{Measurement equations:} Observed pairs $(x_{it},z_{it})$, given $x_{it}>0$.
\ee
We have partial observability of the demand function because the value of demand $X(z;\pi_m)$ is observed at finitely many designs $z$. Furthermore, unlike most state-space models, the state variables are infinite-dimensional.

\subsubsection{Estimating the Hyperparameter}

While it is difficult to derive analytically the distribution of $X_{it}$ given $Z_{it}$, it is easy to simulate its distribution for a given value of $\theta$ (see Appendix \ref{app:dir} for simulations from the Dirichlet distribution). Therefore, $\theta$ can be estimated by the \mbox{method of simulated} moments (MSM), or indirect inference (see \citeauthor{mcfadden-msm}, \citeyear{mcfadden-msm}, \citeauthor{pakes-optimization}, \citeyear{pakes-optimization}, and \citeauthor{gou-mon}, \citeyear{gou-mon}). That is, $\theta$ is estimated by matching some sample and simulated moments of the pair $(X_{it},Z_{it})$.

To illustrate, consider a pure panel such that $M=n$.\footnote{When $M<n$, we simulate $n_mT$ observations for the $m^{th}$ draw from the Dirichlet process.} The steps are the following:
\begin{steps}
	\item Simulate $s=1,\dots,n$ draws from a Dirichlet process given the \mbox{parameter $\theta$.} Each draw $\pi^s(\theta)$ is associated with an individual consumer $i$ such that $s=i$.
	\item Compute simulated consumption $x_{it}^s(\theta)$ by solving the first-order condition in \eqref{foc} with respect to $x_1$ and applying the transformation in \eqref{cases} given $z_{it}=(y_{it},p_{it})$ and $\pi^i(\theta)$.
	\item Construct a collection of $K$ moments from the observed and simulated data:
	\[
		m\equiv\left[\frac{1}{nT}\sum_{i=1}^n\sum_{t=1}^T m_k(x_{it},z_{it})\right]_k \; \; \text{and} \; \; m(\theta)\equiv\left[\frac{1}{nT}\sum_{i=1}^n\sum_{t=1}^T m_k(x_{it}^s(\theta),z_{it})\right]_k.
	\]
	Then, numerically solve the following problem:
	\beq
		\argmin_{\theta} \; \big|\big|m-m(\theta)\big|\big|,
	\eeq
	in which $||\cdot||$ is a Euclidean norm with the form $||m||^2=m'\Omega m$, for some positive-definite $K\times K$ matrix $\Omega$.
\end{steps}

\noindent Given the estimated hyperparameter $\hat{\theta}$, the taste distributions $(\pi_m)$ must be filtered. This step is equivalent to applying the Bayesian approach with the estimated Dirichlet distribution as the prior distribution (see Appendix \ref{app:filter}).

Under Assumptions \ref{ass:1} to \ref{ass:3}, the estimator for $\theta$ is consistent and asymptotically normal, and it converges at a speed of $1/\sqrt{nT}$. The derivation of the asymptotic pro- perties of the filtered functional parameter $\hat{\pi}_m$ is out of the scope of this paper and left for future research.

\subsubsection{Filtering the Taste Distributions}
\label{sec:filtration}

Once the hyperparameter $\theta$ is estimated, we can filter $\pi_m$ by using the following steps:
\begin{steps}
	\item Draw a taste distribution $\tilde{\pi}_m$ from the Dirichlet process given $\hat{\theta}$. Then, by construction, the taste distribution $\tilde{\pi}_m$ is a draw from the \mbox{prior distribution.}
	\item Discretize $\tilde{\pi}_m$ on a grid of values for the taste parameters, $A_1$ and $A_2$. Let $\bar{\pi}_m$ denote the result. The aim of this step is to put $\tilde{\pi}_m$ on a \mbox{grid for optimization.}
	\item Solve the minimization problem:
	\[
		\min_{\pi} \; ||\pi-\bar{\pi}_m|| \; \; \text{s.t. MRS restictions \eqref{mrsrestrictions} and unit mass restrictions.}
	\]
	Let $\hat{\pi}_m^*$ denote the solution. This solution approximates a drawing from the posterior.
	\item Replicate these steps to obtain a sequence of solutions: $\hat{\pi}_{m,s}^*$, $s=1,\dots,S$, where $S$ is the number of replications. The filtered $\hat{\pi}_m$ is obtained by averaging over all simulations such that:
	\[
		\hat{\pi}_m=\frac{1}{S}\sum_{s=1}^S\hat{\pi}_{m,s}^*
	\]
\end{steps}
This procedure involves a high-dimensional argument $\pi_m$, and a very large number of MRS restrictions. Indeed, we need several hundred grid points for $\pi_m$, and, in the application, we have about one-thousand MRS restrictions, for each $m=1,\dots,M$. If the taste parameters, $A_1$ and $A_2$, are independent, this procedure can be numerically simplified by using the fact that these restrictions are bilinear (see Section \ref{bayesian} and Appendix \ref{app:filter}).

\subsection{The Data}
\label{sec:data}

We use the Nielsen Homescan Consumer Panel (NHCP). Nielsen provides a sample of households with barcode scanners. Households are asked to scan all purchased goods on the date of each purchase. The prices are entered by the households or linked to retailer data by The Nielsen Company. The households that agree to participate are compensated through benefits and lotteries.

We focus on the consumption of alcoholic drinks (see \citeauthor{manning-et-al}, \citeyear{manning-et-al}, for an application to alcohol consumption in economics). We classify drinks by type. Good 1 contains beers and ciders.\footnote{The NHCP classifies ciders as wine, by default. We reclassify these beverages using \mbox{product desc-} riptions because most ciders have a low alcohol by volume (ABV).} Good 2 contains wines and liquors. We disregard all non-alcoholic beers, ciders, and wines. We are left with 30,635 beers and ciders, and 108,439 wines and liquors, for a total of 139,074 drinks. We convert all measurement units to litres of alcohol by first converting all units to litres and then multiplying by the standard alcohol by volume (ABV) in each subgroup---specifically, 4.5\% for beer and cider, 11.6\% for wine, and 37\% for liquor. For example, if a household buys two packs of six bottles of beer and each bottle contains 355 millilitres of beer, then the household buys 4.26 litres of this beer, or $4.26\times 0.045=0.231$ litres of alcohol. We use the standard ABV in each subgroup as a result of data limitations. \mbox{Our sample only} contains purchases made at stores, not purchases made at bars, or restaurants.

Measuring quantities in litres of alcohol has at least three advantages: (i) it can account for a quality effect, (ii) it is appropriate for analyzing most relevant structural objects (e.g. the effect of a change in taxation on alcohol consumption), and \mbox{(iii) it yi-} elds continuous quantities, permitting the application of standard tools in consumer theory (which could not be used if quantities were measured in, for example, bottles), and avoiding some common identification issues in the literature.

We restrict our sample to purchases made from August to November in 2016. This relatively short window is used to diminish the impact of changing tastes and product availability, and to avoid most federal holidays in the United States that are often associated with alcohol consumption such as Independence Day, Christmas Day, and New Year's Eve. Our sample contains 28,036 households. Some additional details of this restricted sample are placed in Appendix \ref{app:d}.

\begin{table}
      \centering
      \caption{Mean $m$, standard deviation $\sigma$, the ratio $\sigma/m$, \mbox{and quantiles for} expenditure $\tilde{y}$, prices $\tilde{p}_j$, normalized expenditure $y$, and normalized price $p$. Normalized expenditures $y$ and prices $p$ are conditional on being defined.}
      \begin{tabular}{crrrrrrrrr}
            \midrule \midrule
                          &        &           & & \multicolumn{5}{c}{Quantiles}       & \\ \cmidrule{5-9}
            Var.      & $m$   & $\sigma$ & $\sigma/m$ & \; 0\%  & \; 25\% & \; 50\% & \; 75\% & \; 100\%  & $N$ \\ \midrule
            $\tilde{y}$   & 52.29 & 70.62    & 1.35 & 0.00 & 13.28 & 28.17 & 62.97 & 2,767.76 & 63,972 \\
            $\tilde{p}_1$ & 70.36 & 44.19 & 0.62 & 0.05 & 46.41 & 62.55 & 83.41 & 2,893.65 & 33,077 \\
            $\tilde{p}_2$ & 61.33  & 326.93      & 5.33 & 0.03  & 29.29 & 48.56 & 75.14 & 39,900.85 & 45,518 \\ \midrule
            $y$           & 1.60 & 3.48 & 2.17 & 0.00 & 0.27 & 0.74 & 1.83 & 228.57 & 45,518 \\
            $p$           & 2.20 & 4.06 & 1.54 & 0.00 & 0.89 & 1.37 & 2.29 & 139.76 & 14,659 \\ \midrule \midrule
      \end{tabular}
      \label{table:data:1}
\end{table}

The dated purchases are aggregated by month. For each household and \mbox{month, the} prices are constructed by dividing the total expenditure for each aggregate good (after accounting for the value of coupons) by the amount of alcohol of that aggregate good purchased by the household, when this amount is strictly positive. \mbox{Then, we norm-} alize by the price of good 2. This procedure yields four monthly \mbox{observations per hou-} sehold for a total of 112,144. A total of 63,936 observations have positive consumption such that $x_{it}> 0$. Table \ref{table:data:1} gives summary statistics conditional on $x_{it}\neq 0$. \mbox{The prices} $(p_{it})$ are conditional on being well-defined (see the discussion of partial observability on pages 23 and 24). For the interpretation of the results, recall that $\tilde{y}$ is the expenditure (prior to normalization), and that $\tilde{p}_j$ is the price of \mbox{good $j$ (prior to normal-} ization).

\begin{table}
      \centering
      \caption{Proportion of observations by type.}
      \begin{tabular}{ccc}
            \midrule \midrule
                    & $x_2=0$   & $x_2>0$   \\ \midrule
            $x_1=0$ & 0.4298    & 0.2751    \\
            $x_1>0$ & 0.1642    & 0.1307    \\ \midrule \midrule
      \end{tabular}
      \label{table:data:2}
\end{table}

There are four regimes of observations: (i) zero expenditure on all goods, \mbox{(ii) zero} expenditure on good 1 and strictly positive expenditure on good 2, (iii) \mbox{strictly pos-} itive expenditure on good 1 and zero expenditure on good 2, and (iv) \mbox{strictly positive} expenditure on all goods. Table \ref{table:data:2} provides the proportion of observations \mbox{in each regi-} me, and shows a large proportion of observations with zero expenditure. Recall, under Assumption \ref{ass:2}, designs $z_{it}$ are drawn from a distribution. Therefore, we can interpret this result as a mass at zero in the marginal distribution of expenditure.

Figure \ref{fig:data:exp} displays the sample distribution of expenditure $\tilde{y}_{it}$ by regime: the distribution of expenditure $\tilde{y}_{it}$ conditional on $x_{it}>0$ is on the left; the sample distributions of expenditure $\tilde{y}_{it}$ for the two other regimes with positive expenditure are on the right. The shape of the sample distribution of expenditure $\tilde{y}_{it}$ does not appear to vary all that much with the regime. That being said, the sample distribution conditional on $x_{it}>0$ has more probability attributed to higher expenditures.

Figure \ref{fig:data:prices} compares the sample distributions of prices $\tilde{p}_j$ by regime: the \mbox{sample dis-} tributions of $\tilde{p}_1$ are on the left; the sample distributions of $\tilde{p}_2$ are on the right. Although the sample distribution of $\tilde{p}_1$ differs from the sample distribution of $\tilde{p}_2$, these distributions do not seem to be affected by the regime.

Figure \ref{fig:data:consumption} displays the sample distributions of (normalized) designs $z_{it}=(y_{it},p_{it})$ and the components of consumption $x_{it}$ given $x_{it}>0$. As expected, the components of consumption $x_{it}$ are increasing in expenditure $y_{it}$. Furthermore, the first component of consumption $x_{it}$ is more affected by changes in the price $p_{it}$ than the second component.

\begin{figure}
      \centering
      \begin{subfigure}[b]{0.48\linewidth}
      \centering
      \includegraphics[scale=0.4]{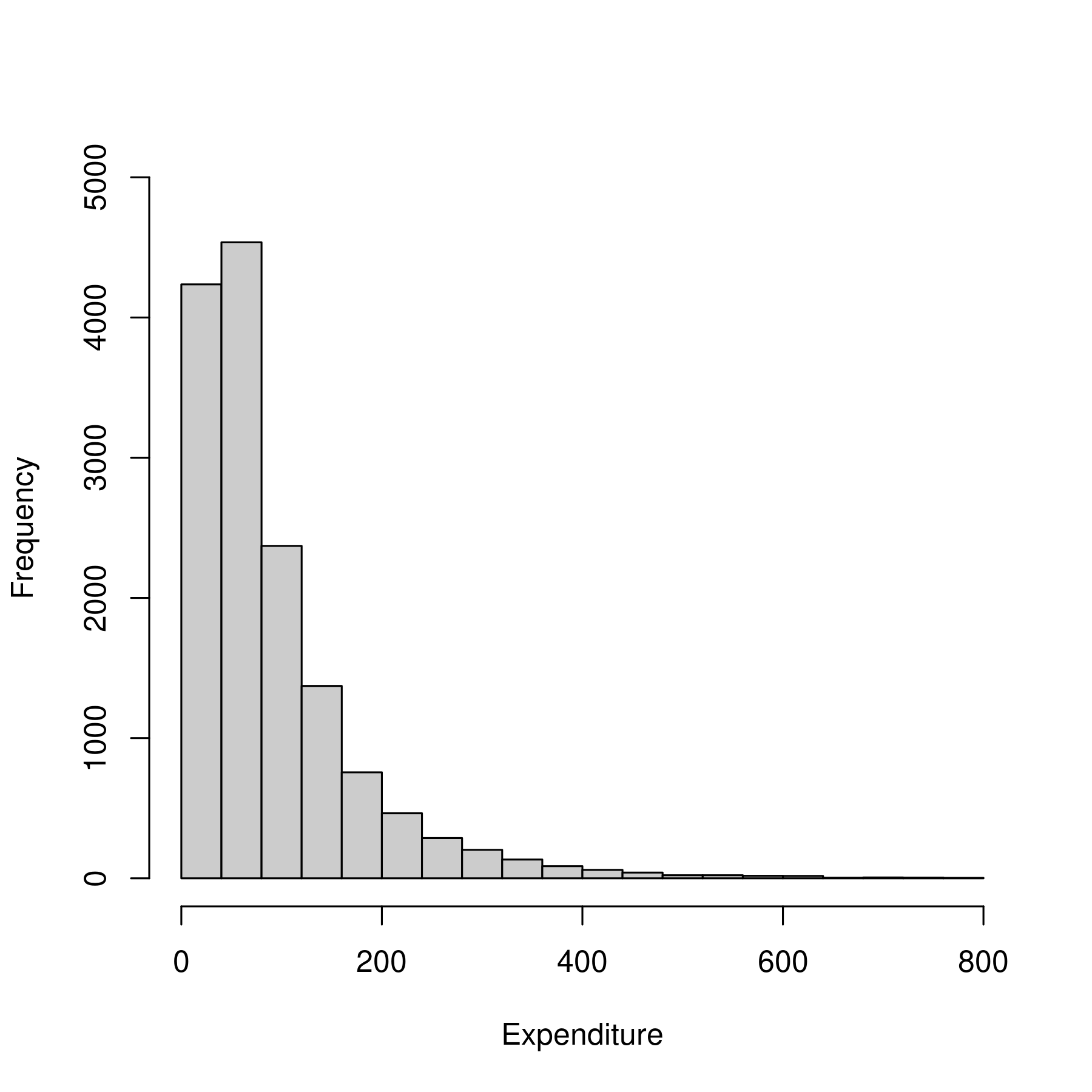}
      \end{subfigure}
      \begin{subfigure}[b]{0.48\linewidth}
      \centering
      \includegraphics[scale=0.4]{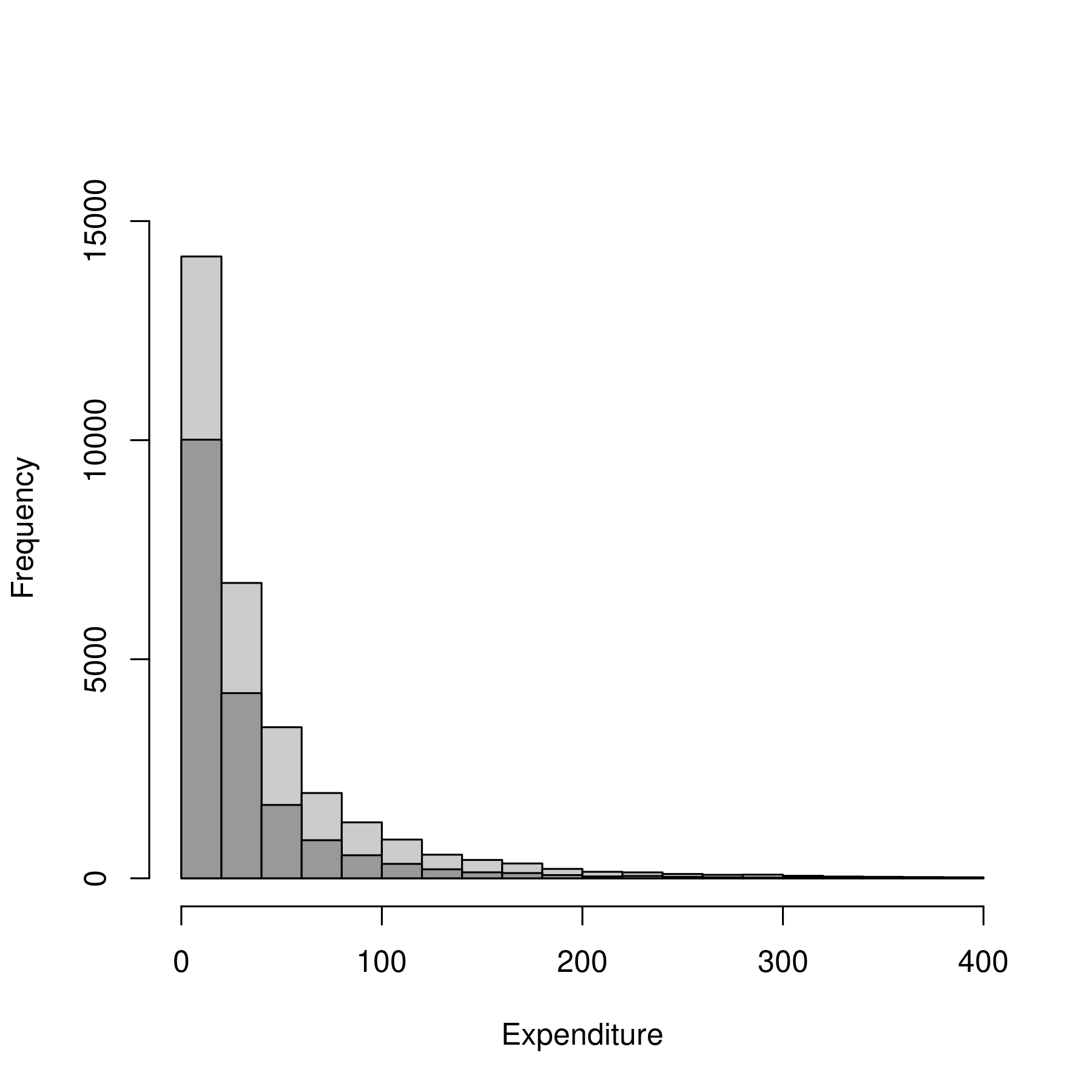}
      \end{subfigure}
      \caption{\textbf{Sample Distributions of Expenditure $\tilde{y}$ by Regime.} On the left, we illustrate the sample distribution conditional on $x_1>0$ and $x_2>0$; on the right, the light histogram illustrates the sample distribution conditional on $x_1=0$ and $x_2>0$, and the dark histogram illustrates the sample distribution conditional on $x_1>0$ and $x_2=0$.}
      \label{fig:data:exp}
\end{figure}

\begin{figure}
      \centering
      \begin{subfigure}[b]{0.48\linewidth}
      \centering
      \includegraphics[scale=0.4]{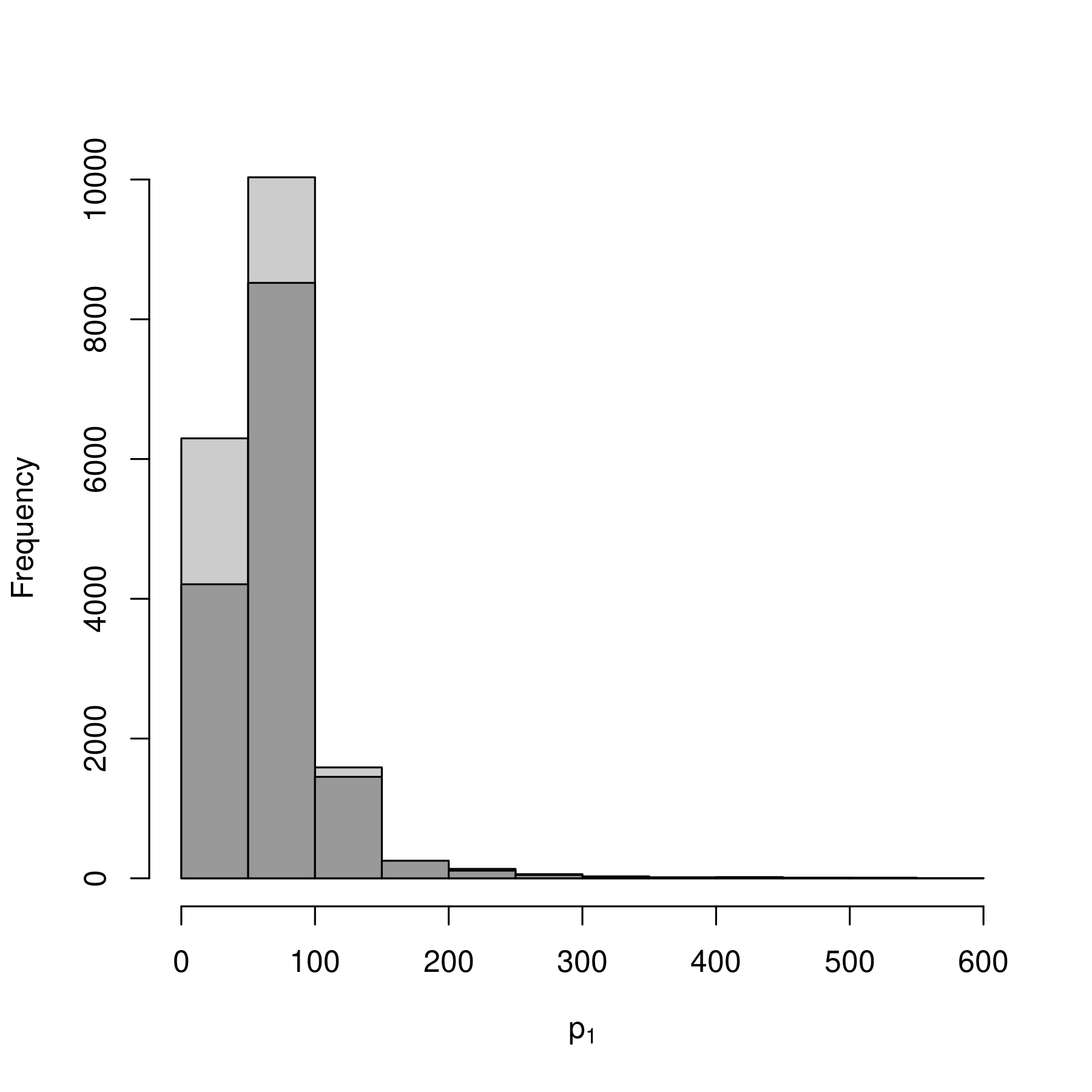}
      \end{subfigure}
      \begin{subfigure}[b]{0.48\linewidth}
      \centering
      \includegraphics[scale=0.4]{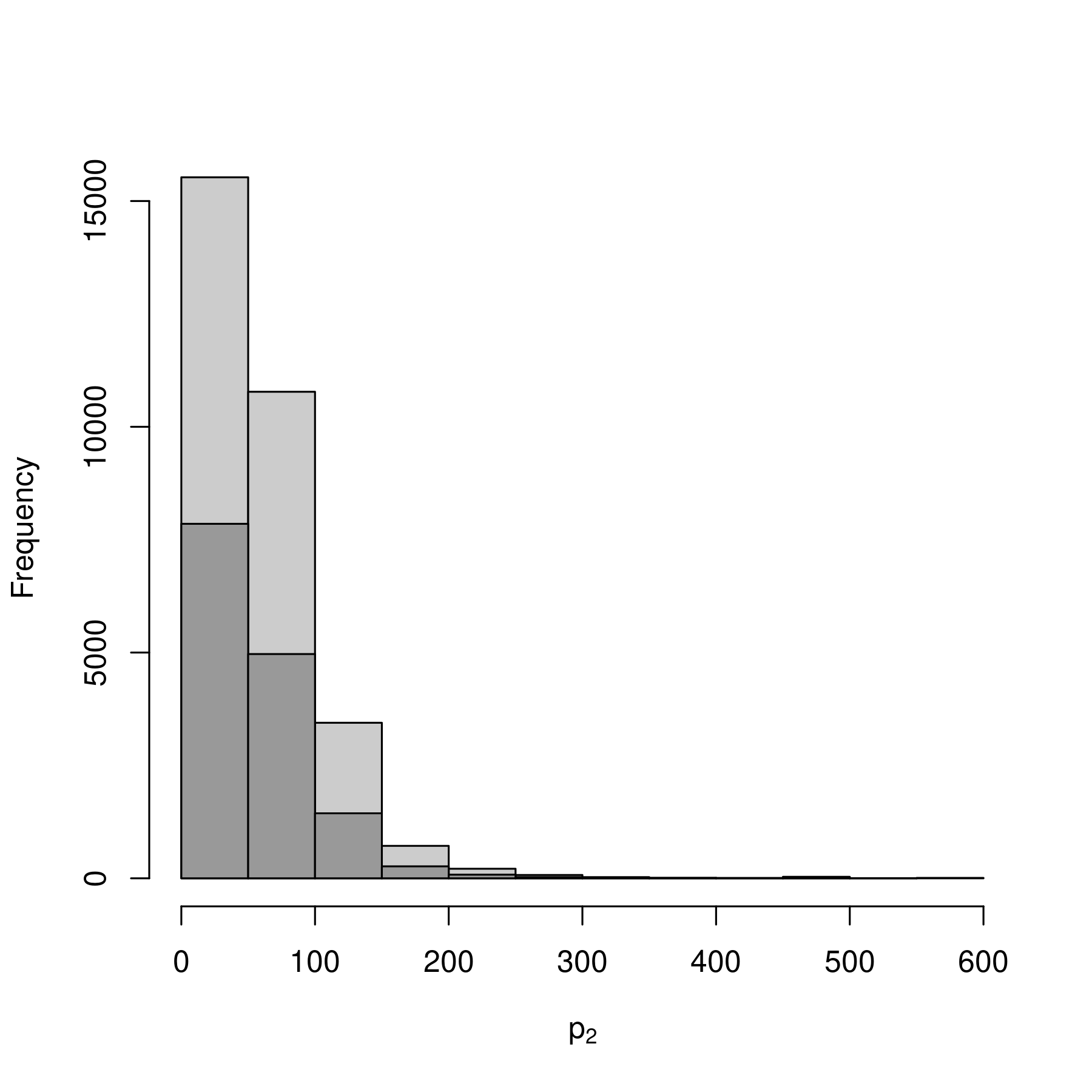}
      \end{subfigure}
      \caption{\textbf{Sample Distributions of Prices $\tilde{p}_j$ by Regime.} On the left, the light histogram illustrates the sample distribution of $\tilde{p}_1$ conditional on $x_1>0$ and $x_2=0$; on the right, the light histogram illustrates the sample distribution of $\tilde{p}_2$ conditional on $x_1=0$ and $x_2>0$; in each plot, the dark histogram illustrates the sample distribution conditional on \mbox{$x_1>0$ and $x_2>0$.}}
      \label{fig:data:prices}
\end{figure}

\begin{figure}
      \centering
      \begin{subfigure}[b]{0.48\linewidth}
      \centering
      \includegraphics[scale=0.4]{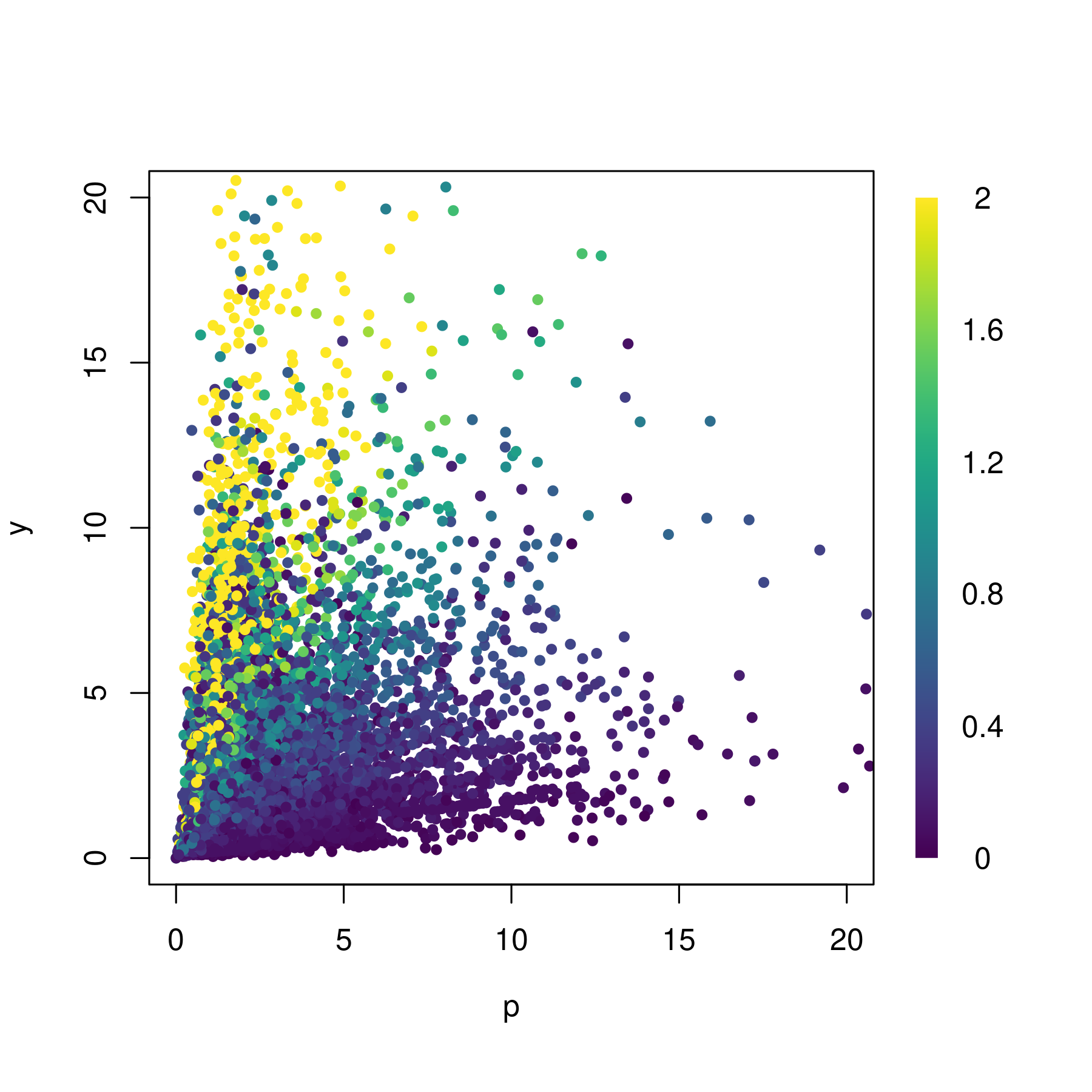}
      \end{subfigure}
      \begin{subfigure}[b]{0.48\linewidth}
      \centering
      \includegraphics[scale=0.4]{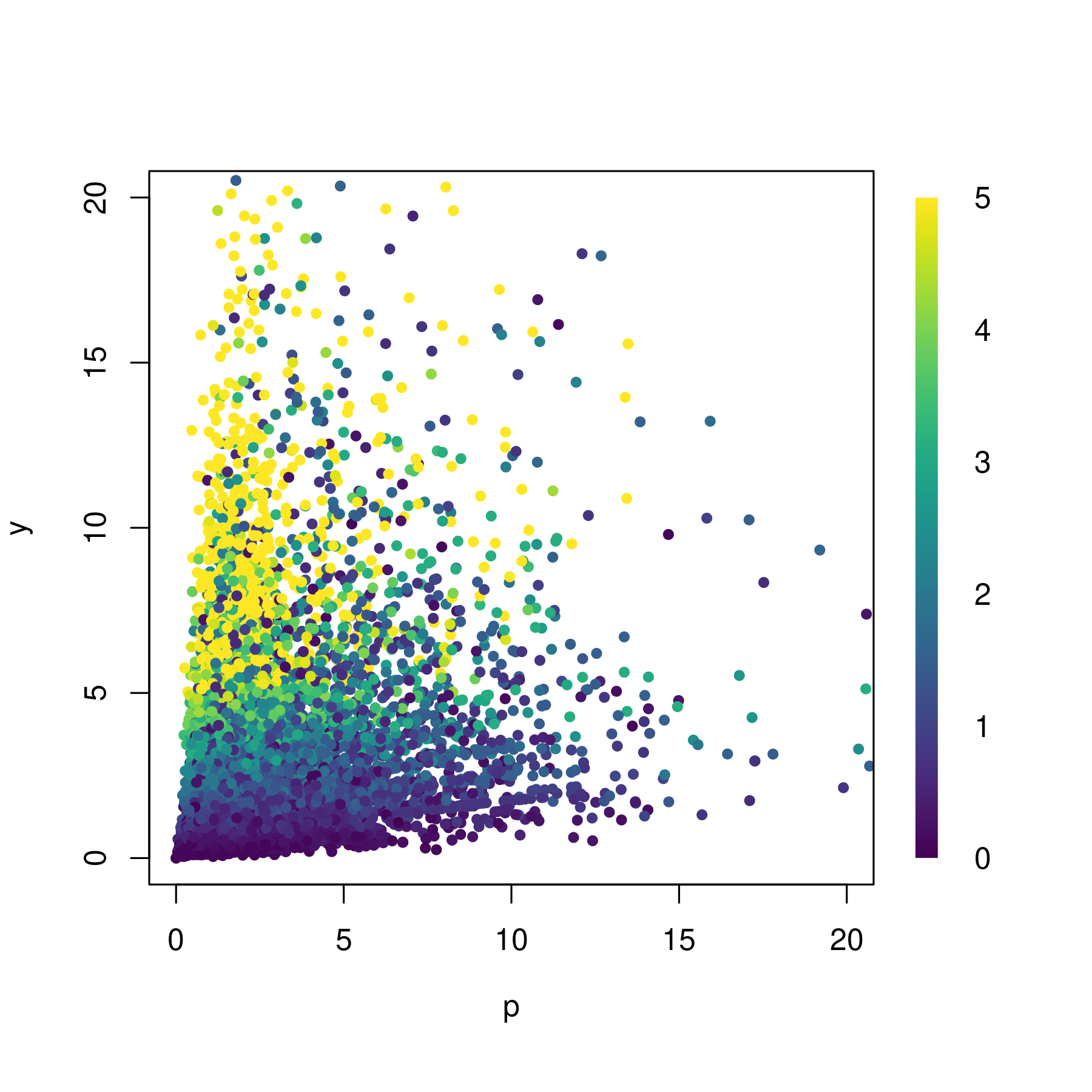}
      \end{subfigure}
      \caption{\textbf{Sample Distributions of Designs $z_{it}$ and Consumption.} These figures are conditional on $x_{it}>0$. On the left, colour describes the quantity of good 1; on the right, colour describes the quantity of good 2.}
      \label{fig:data:consumption}
\end{figure}

Since we consider a rather short window of time, we follow the segmented population approach. We segment the population by state. Large states (e.g. California) are segmented again by county. Specifically, a county is given its own segment if it has more than 70 observations with positive consumption and it is in a \mbox{state with more} than 1,000 observations with positive consumption. We are left with a total \mbox{of 65 seg-} ments, each corresponding to a state or county. The smallest segment is Wyoming, containing 15 observations with positive consumption; the largest state is Florida (after removing Broward, Hillsborough, Palm Beach, Pinellas, and Miami-Dade counties), containing 880 observations with positive expenditure; the mean number of observations with positive consumption per segment is approximately 226.

Figure \ref{fig:data:consumption-state} displays the sample distributions of (normalized) designs $z_{it}=(y_{it},p_{it})$ and the first component of consumption $x_{it}$ given $x_{it}>0$ in two of the larger segments: California (after removing Almeda, Los Angeles, Orange, Riverside, Sacramento, San Bernardino, and San Diego counties), and Florida (after removing Broward, Hillsborough, Palm Beach, Pinellas, and Miami-Dade counties).

\begin{figure}
      \centering
      \begin{subfigure}[b]{0.48\linewidth}
      \centering
      \includegraphics[scale=0.4]{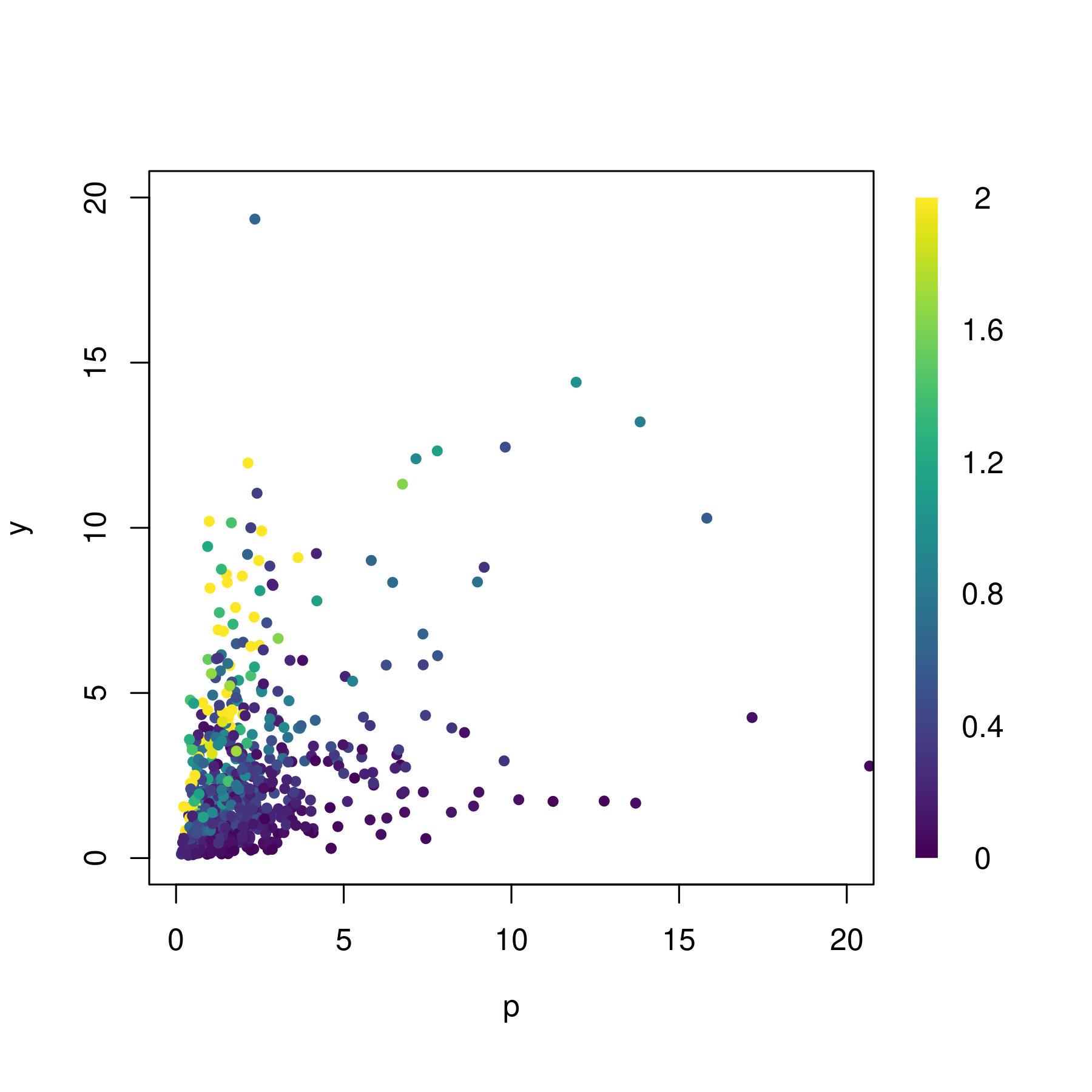}
      \end{subfigure}
      \begin{subfigure}[b]{0.48\linewidth}
      \centering
      \includegraphics[scale=0.4]{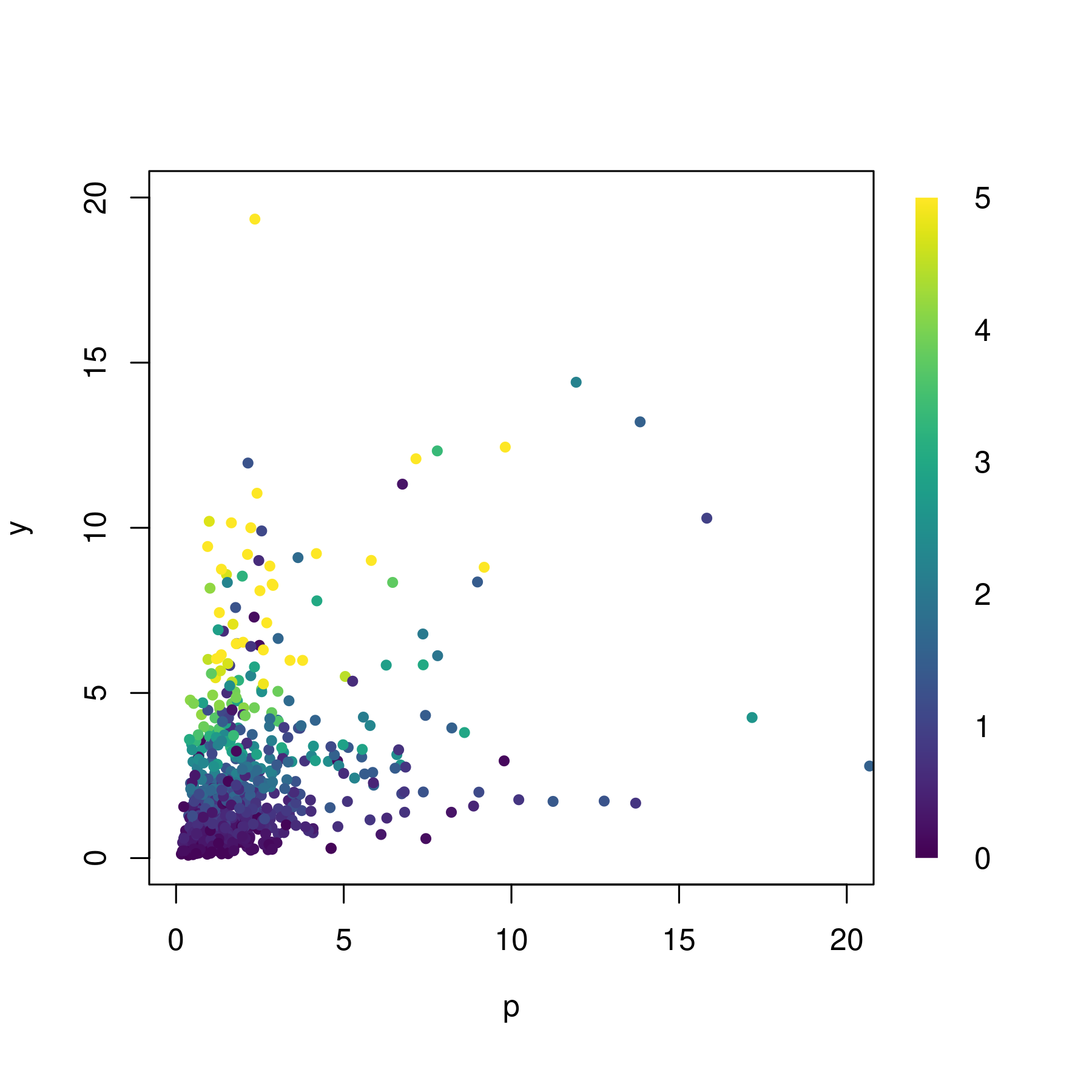}
      \end{subfigure}
      \begin{subfigure}[b]{0.48\linewidth}
      \centering
      \includegraphics[scale=0.4]{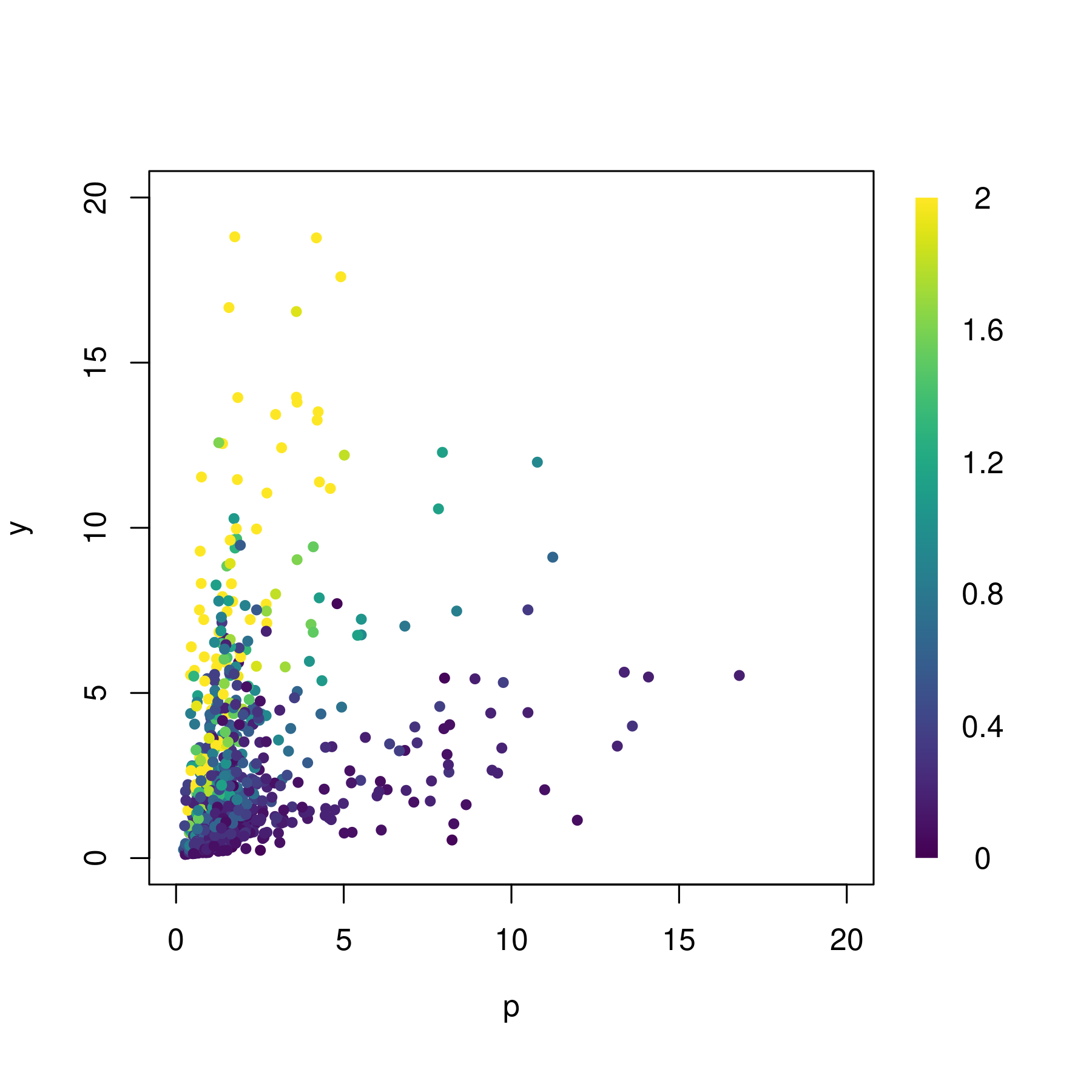}
      \end{subfigure}
      \begin{subfigure}[b]{0.48\linewidth}
      \centering
      \includegraphics[scale=0.4]{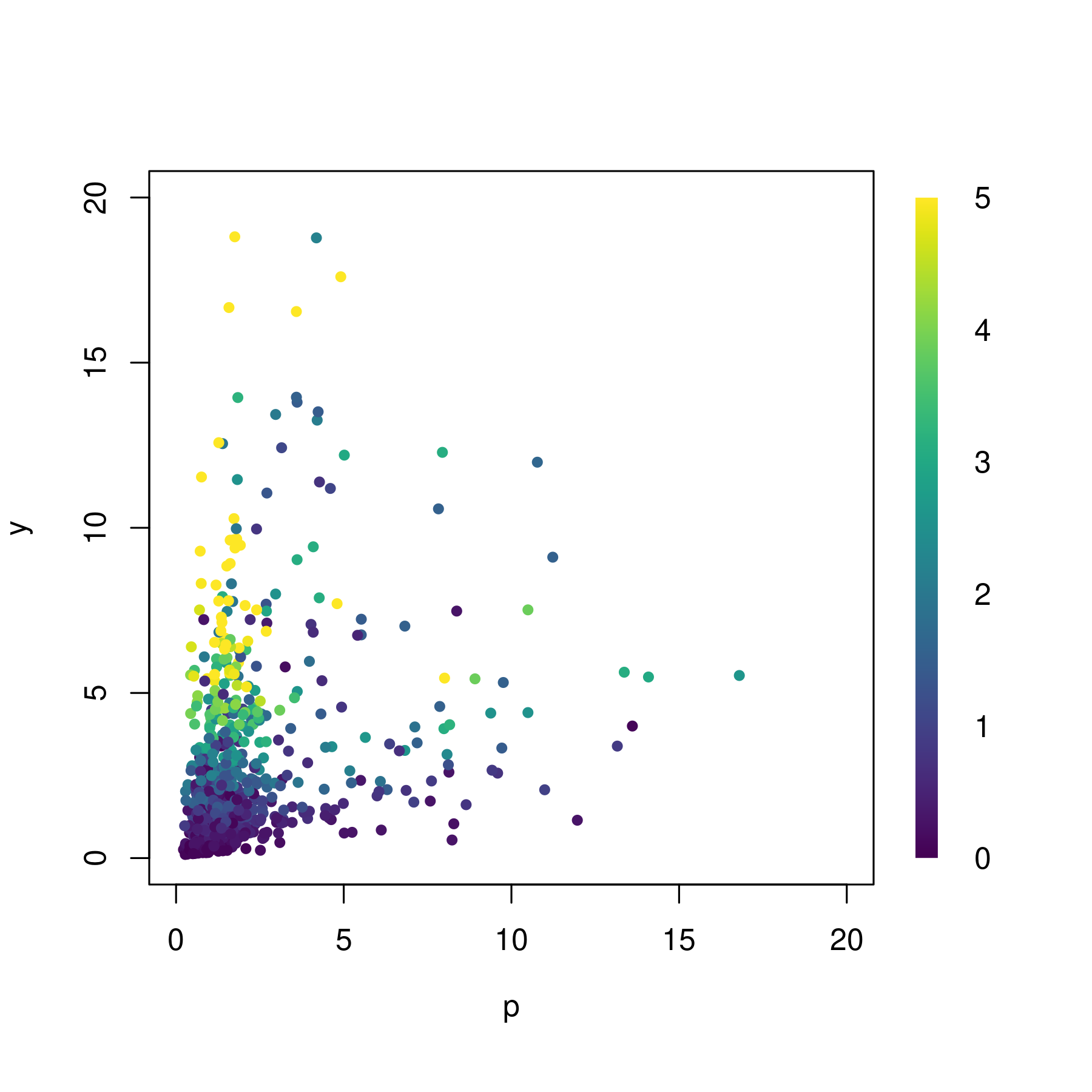}
      \end{subfigure}
      \caption{\textbf{Sample Distributions by State.} These figures are conditional on $x_{it}>0$. California is shown on the top, and Florida is shown on the bottom. On the left, colour describes the quantity of good 1; on the right, colour describes the quantity of good 2.}
      \label{fig:data:consumption-state}
\end{figure}

Figure \ref{fig:data:nw-state} displays the Nadaraya-Watson (kernel) estimates of the demand function for beer conditional on $x_{it}>0$ in California and Florida over a subset of the domain of designs. Demand for beer in California is lower and less responsive to price changes than in Florida. Figure \ref{fig:slice} displays Engel curves for good 1 in California and Florida given $p\equiv\tilde{p}_1/\tilde{p}_2=4$.\footnote{This price is chosen to be in a sufficiently dense region of the sample distribution (see Figure \ref{fig:data:consumption-state}).} These Engel curves cross.

\begin{figure}
      \centering
      \begin{subfigure}[b]{0.48\linewidth}
      \centering
      \includegraphics[scale=0.4]{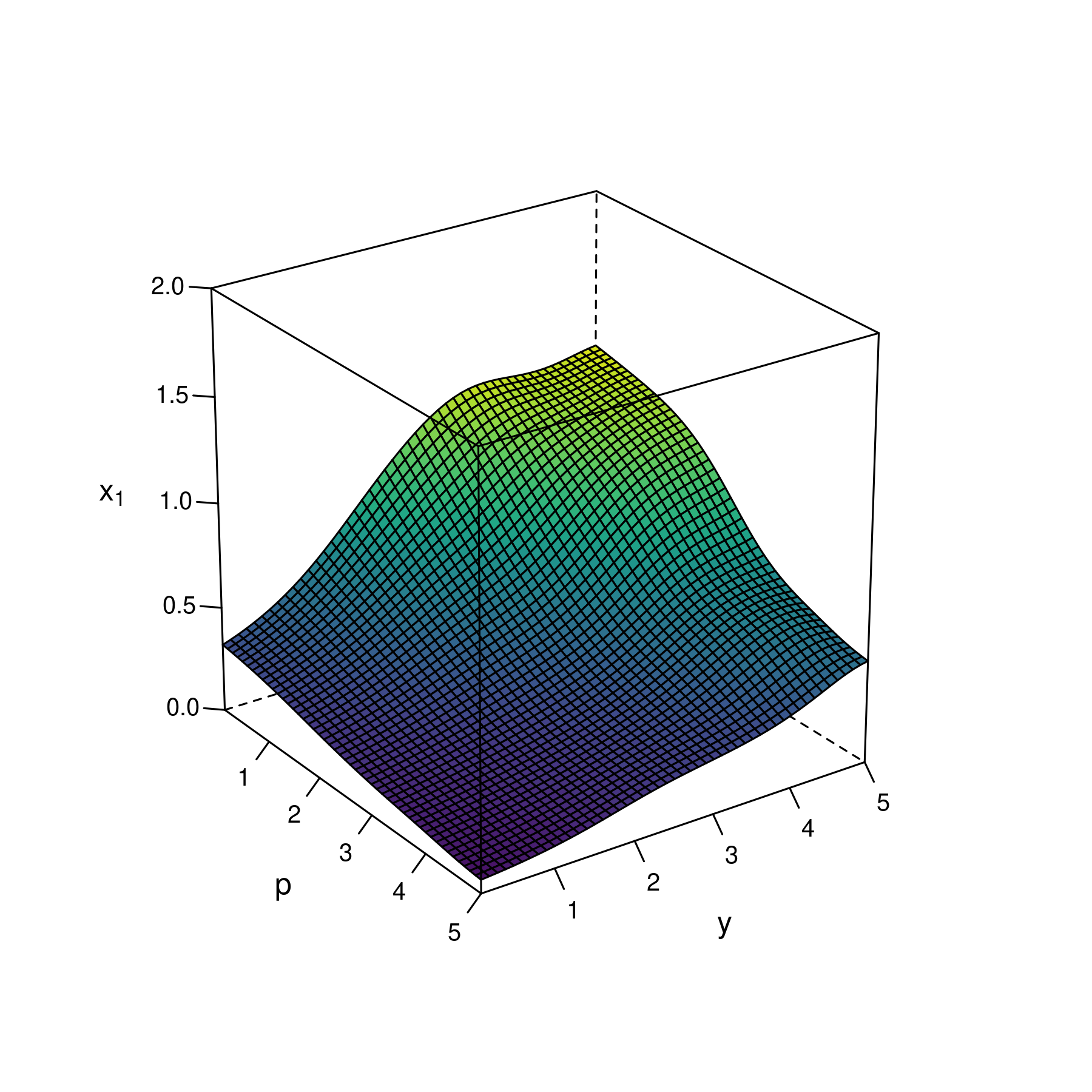}
      \end{subfigure}
      \begin{subfigure}[b]{0.48\linewidth}
      \centering
      \includegraphics[scale=0.4]{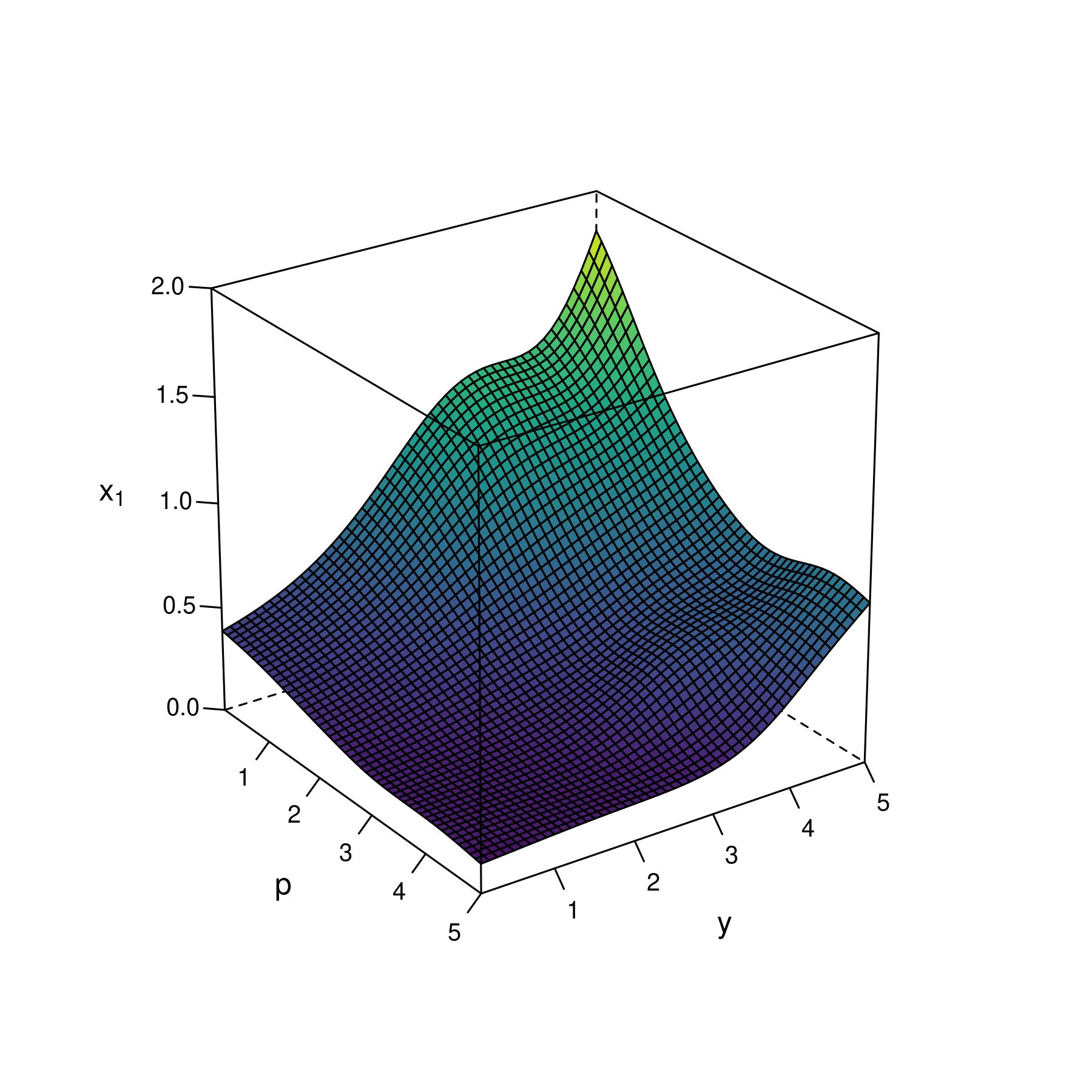}
      \end{subfigure}
      \caption{\textbf{Demand.} Nadaraya-Watson estimates of the demand function for good 1 conditional on $x_{it}>0$ in California (left) and Florida (right).}
      \label{fig:data:nw-state}
\end{figure}

\begin{figure}
      \centering
      \includegraphics[scale=0.5]{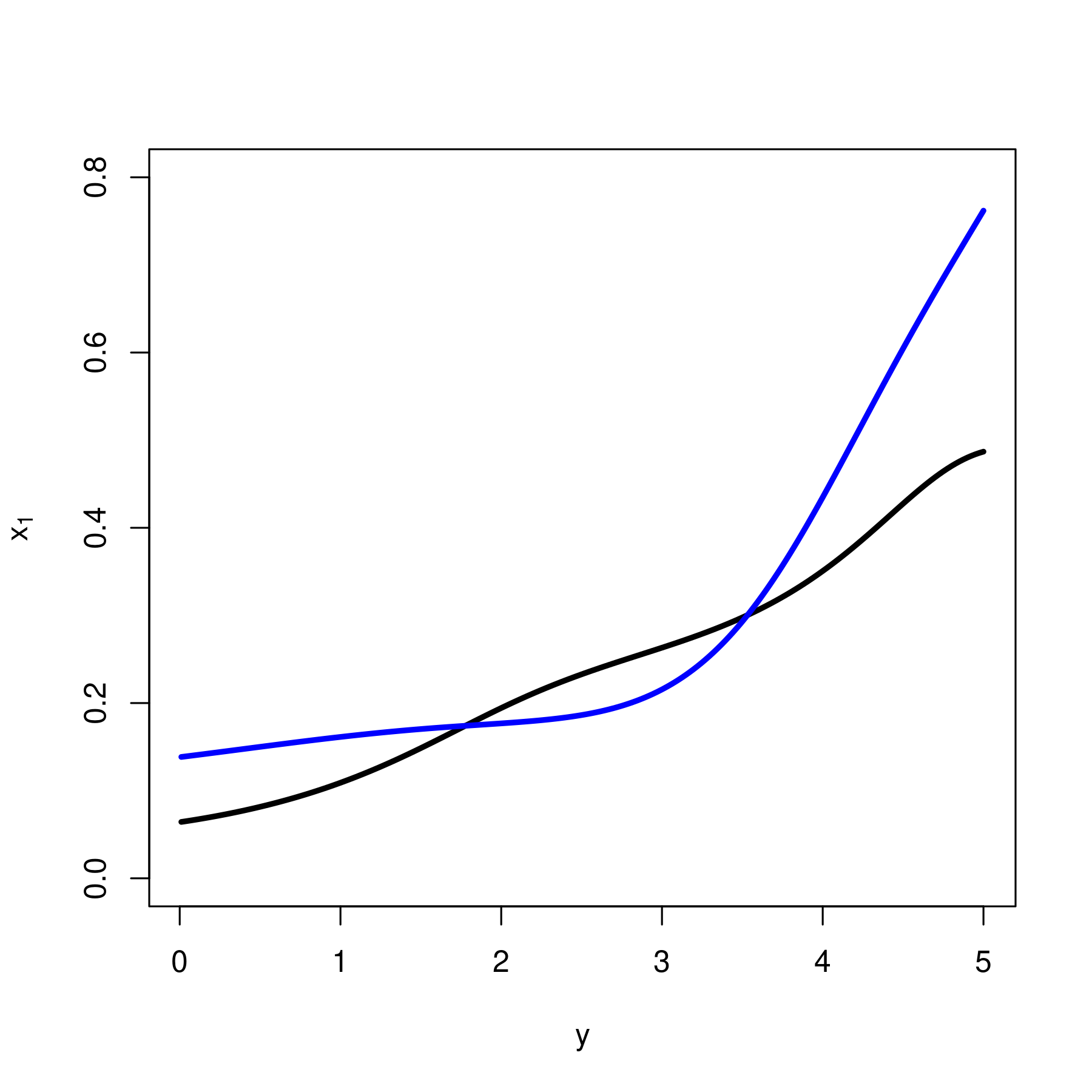}
      \caption{\textbf{Engel Curves.} Engel curves for good 1 in California (black) and Florida (blue). These curves violate monotonicity because they cross.}
      \label{fig:slice}
\end{figure}

\subsection{Estimation Results}

As an illustration, we consider the SARA model in the hyperparametric framework. We assume that the taste parameters, $A_1$ and $A_2$, are independent. \mbox{Under this assum-} ption, the taste uncertainty is characterized by the marginal distributions, $\pi_1$ and $\pi_2$. The marginal distribution $\pi_j$ of $A_j$ is independently drawn from a Dirichlet process $F_j$, $j=1,2$. The mean of $F_j$ is a log-normal distribution with parameters $\mu_j$ and $\sigma_j$, and the scale parameter of $F_j$ is $c_j$. The utility function corresponding to this log-normal mean distribution, say $\bar{\pi}_j$, has a quasi closed-form expression. \mbox{Indeed, under} this distribution, we obtain:
\[
	\log(A_j)=\mu_j+\sigma_j\ve_j, \; \; \forall j=1,2,
\]
where $\ve_j$ is distributed with respect to a standard normal distribution. Then:
\[
\begin{gathered}
	\mathbb{E}_{\bar{\pi}_j}[\exp(-A_jx_j)]=\mathbb{E}_{\bar{\pi}_j}[\exp(-\exp(\mu_j+\sigma_j\ve_j)x_j)], \\
	=\frac{1}{\sqrt{1+w(x_j\exp(\mu_j)\sigma_j^2)}}\exp\left\{-\frac{1}{2\sigma_j^2}w(x_j\exp(\mu_j)\sigma_j^2)^2-\frac{1}{\sigma^2}w(x_j\exp(\mu_j)\sigma_j^2)\right\},
\end{gathered}
\]
where $w(x)$ is the Lambert function, defined by the implicit equation:
\[
	w(x)\exp(w(x))=x,
\]
[see equation (1.3) in \citeauthor{laplace-lognormal}, \citeyear{laplace-lognormal}]. By drawing from the Dirichlet process, we will draw a stochastic utility function around the closed-form expression above. The hyperparameter $\theta$ has six components such that:
\[
	\theta=(\mu_1,\sigma_1,c_1,\mu_2,\sigma_2,c_2).
\]

\subsubsection{The Hyperparameter}
\label{sec:est:hyper}

As described in Section \ref{sec:hyper}, the first step involves estimating the hyperparameter $\theta$ using the Method of Simulated Moments (MSM). The hyperparameter $\theta$ is calibrated by using the following (sample and simulated) moments computed for all of the 63,936 observations with positive consumption:
\bi
	\item marginal moments of $(X_{it})$;
	\item cross-moments of $(\log X_{it},\log P_{it})$ and $(\log X_{it},\log Y_{it})$;
	\item cross-moments of $(X_{it}, P_{it})$, $(X_{it},Y_{it})$, $(X_{it},\log P_{it})$, and $(X_{it},\log Y_{it})$.
\ee
The moments in (ii) are the moments used in the Almost Ideal Demand System \citep[see][]{aids}; the moments in (iii) are introduced in order to capture risk effects by comparison with the moments in (ii). The optimum is found using a random search algorithm over a sufficiently big support.\footnote{Random search is more efficient than grid search in hyperparameter optimization \citep{bergstra}.}

To apply MSM, it is necessary to compute simulated consumption $x_{it}^s(\theta)$ for every observation, at each step of the optimization algorithm. This procedure is computati- onally costly. Note that, the number of simulated observations with positive consumption is stochastic, and not necessarily equal to the number of observations with positive consumption in the sample. This aspect has no impact on the \mbox{consistency of the} MSM estimator.

The estimated hyperparameter is:
\beq
	\hat{\theta}=(0.7987, 3.5516, 45.0951, 0.1201, 3.6597, 3.5544).
\eeq
Therefore, the median level of risk aversion for the mean of the Dirichlet process\footnote{This is not the absolute risk aversion of the utility function for the log-normal mean distribution which depends on the consumption level and has to be computed with a modified density.} for $A_1$ is $\exp(0.5495)\simeq 2.2226$, and the median level of risk aversion for the mean of the Dirichlet process for $A_2$ is $\exp(0.8738)\simeq 1.1276$. The fact that $\mu_1$ is smaller than $\mu_2$ is expected: Since quantities are measured in terms of volume of alcohol, this result is consistent with the faster overall intake of alcohol when consuming drinks with a higher ABV. Moreover, the distribution $\pi_1$ of $A_1$ is much more concentrated around its mean than the distribution $\pi_2$ of $A_2$, as the scaling parameter $c_1=45.0951$ for $\pi_1$ is much larger than the scaling parameter $c_2=3.5544$ for $\pi_2$.

We do not report any standard errors because they are automatically small from the large number of observations. Indeed, the standard significance test procedures (such as comparing a t-statistic to the critical value of a standard normal at the 5\% significance level) are not relevant in this big data framework. The highest degree of uncertainty concerns the filtered functional parameters $(\hat{\pi}_m)$ since $\pi_m$ is a high-dim- ensional parameter and the number of observations in each segment is much smaller.

The means of these Dirichlet processes are displayed in the left panel in Figure \ref{fig:lognorm}. The right panel displays the indifference curves associated with utility levels $-0.1000$, $-0.0800$, and $-0.0680$ for a draw from the Dirichlet process given $\hat{\theta}$.

\begin{figure}
      \centering
      \begin{subfigure}[b]{0.48\linewidth}
      \centering
      \includegraphics[scale=0.1]{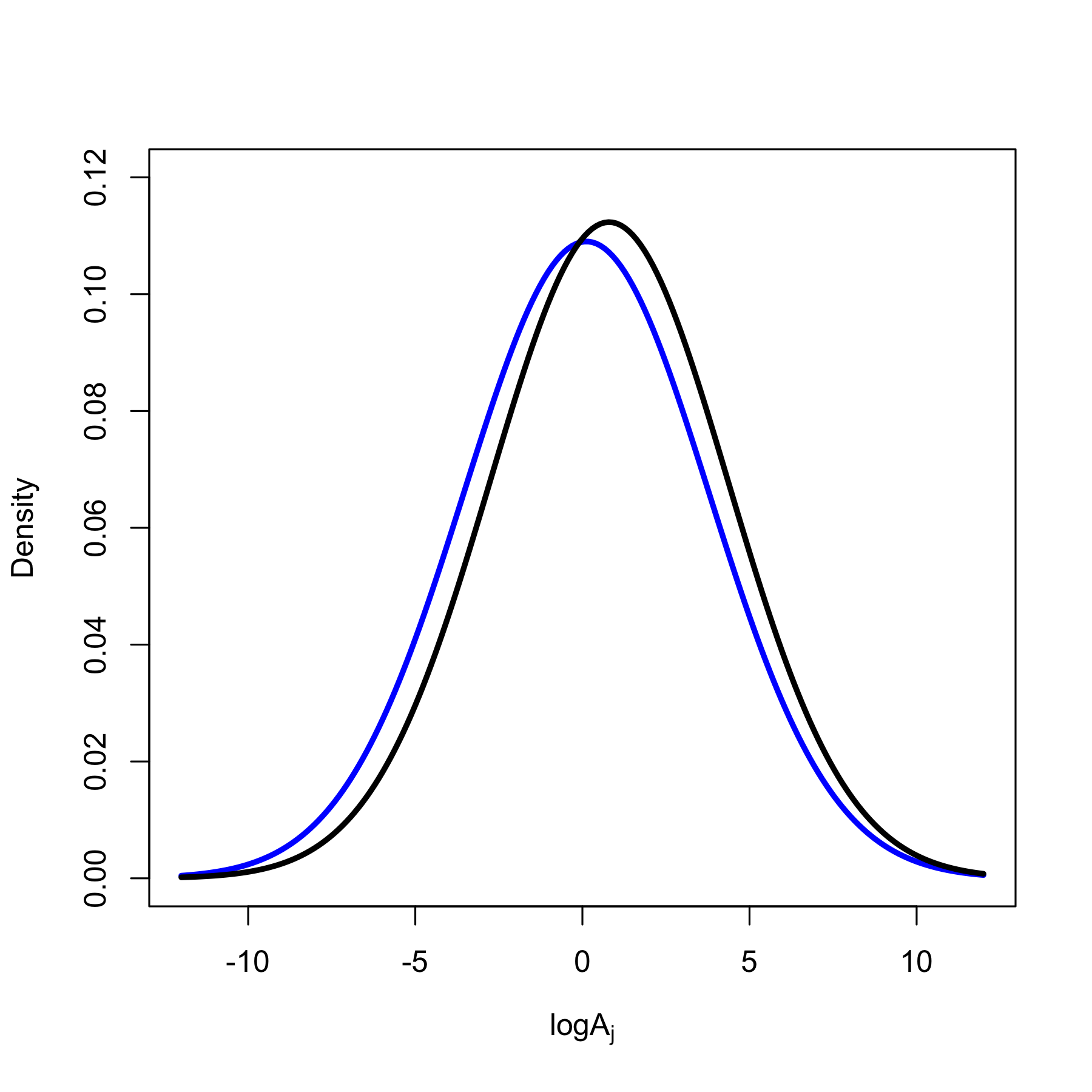}
      \end{subfigure}
      \begin{subfigure}[b]{0.48\linewidth}
      \centering
      \includegraphics[scale=0.1]{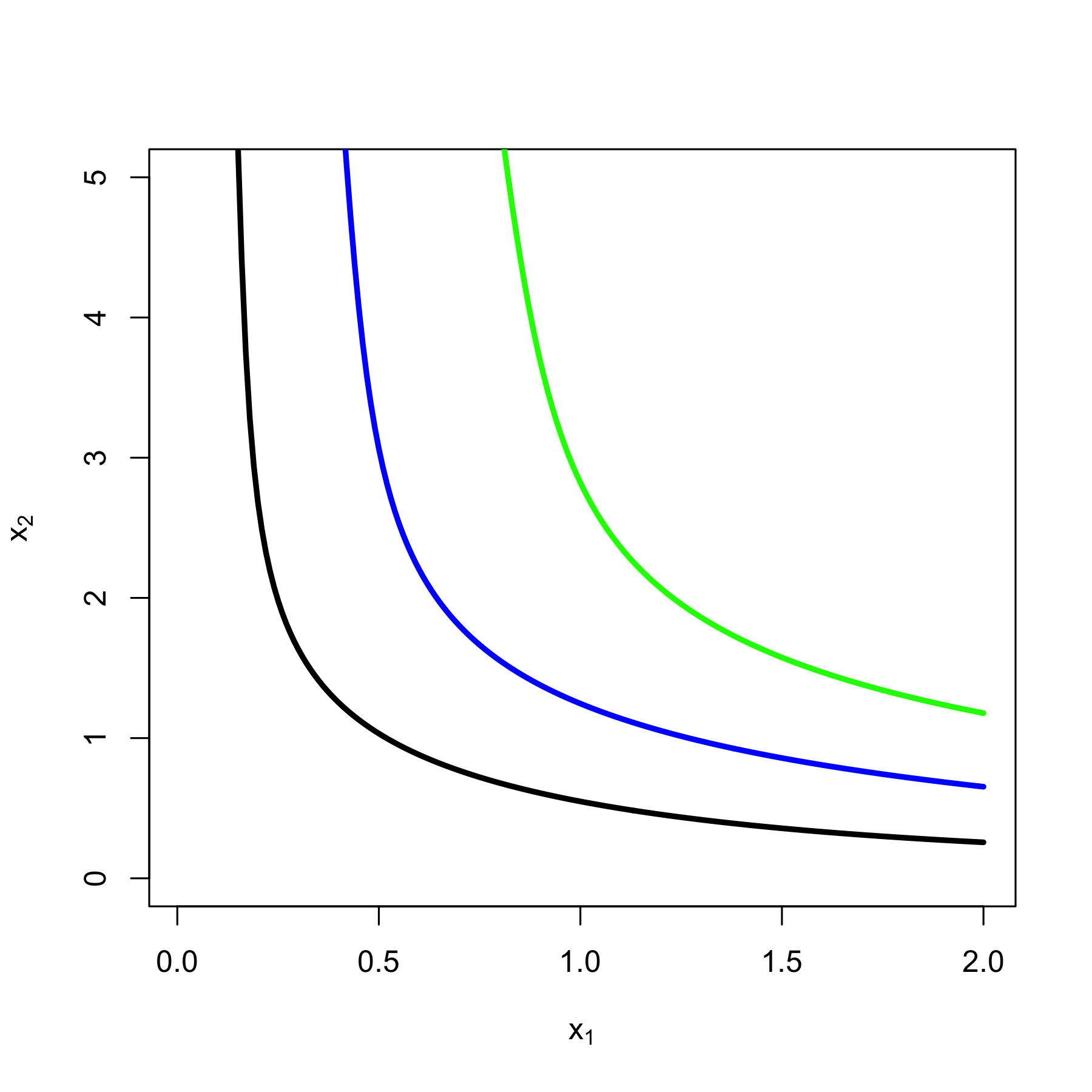}
      \end{subfigure}
      \caption{On the left, black shows the density of the mean of the Dirichlet process for $A_1$, and blue shows the density of the mean of the Dirichlet process for $A_2$. The $x$-axis is in log-scale. For scale: $\exp(-5)\simeq0.0067$ and $\exp(5)\simeq 148.4131$. The figure on the right displays indifference curves associated with these distributions.}
      \label{fig:lognorm}
\end{figure}

Figure \ref{fig:qqdir} displays the Q-Q plots for two draws $(\pi_1^s,\pi_2^s)$, $s=1,2$, from the Dirichlet process given $\hat{\theta}$. In particular, we plot the quantiles of the realization of the distribution $\pi_j$ of $\log(A_j)$ against the quantiles of the normal distribution given the estimated hyperparameters $(\hat{\mu}_j,\hat{\sigma}_j)$, for $j=1,2$. If these quantiles coincide exactly, they will lie on the 45-degree line. As expected, these Q-Q plots lie approximately around the 45-degree line. The draws $(\pi_1^s)$, $s=1,2$, for $\pi_1$ are closer the 45-degree line and ``more continuous'' than the draws $(\pi_2^s)$, $s=1,2$, for $\pi_2$ since $c_1>c_2$.

\begin{figure}
      \centering
      \begin{subfigure}[b]{0.48\linewidth}
      \centering
      \includegraphics[scale=0.1]{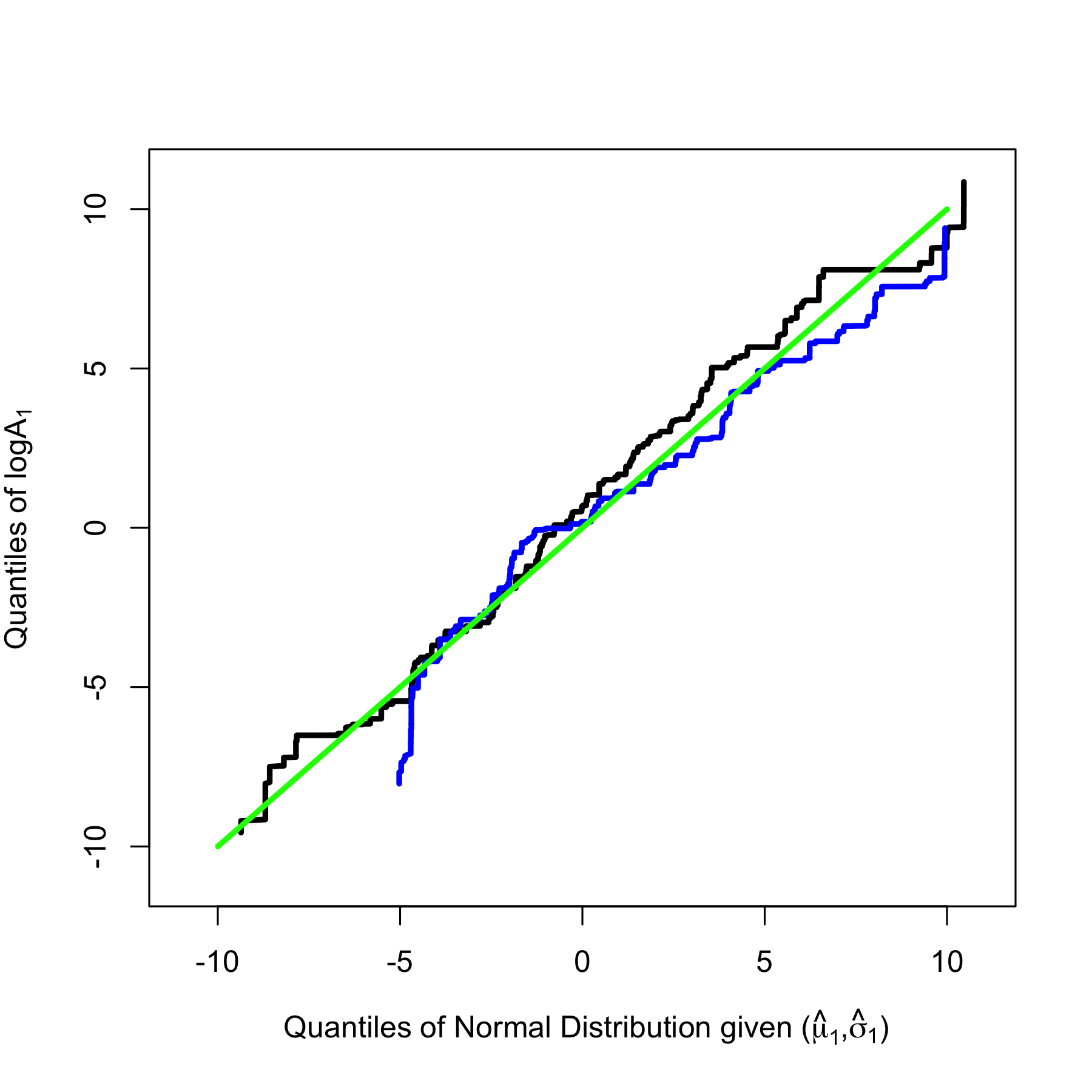}
      \end{subfigure}
      \begin{subfigure}[b]{0.48\linewidth}
      \centering
      \includegraphics[scale=0.1]{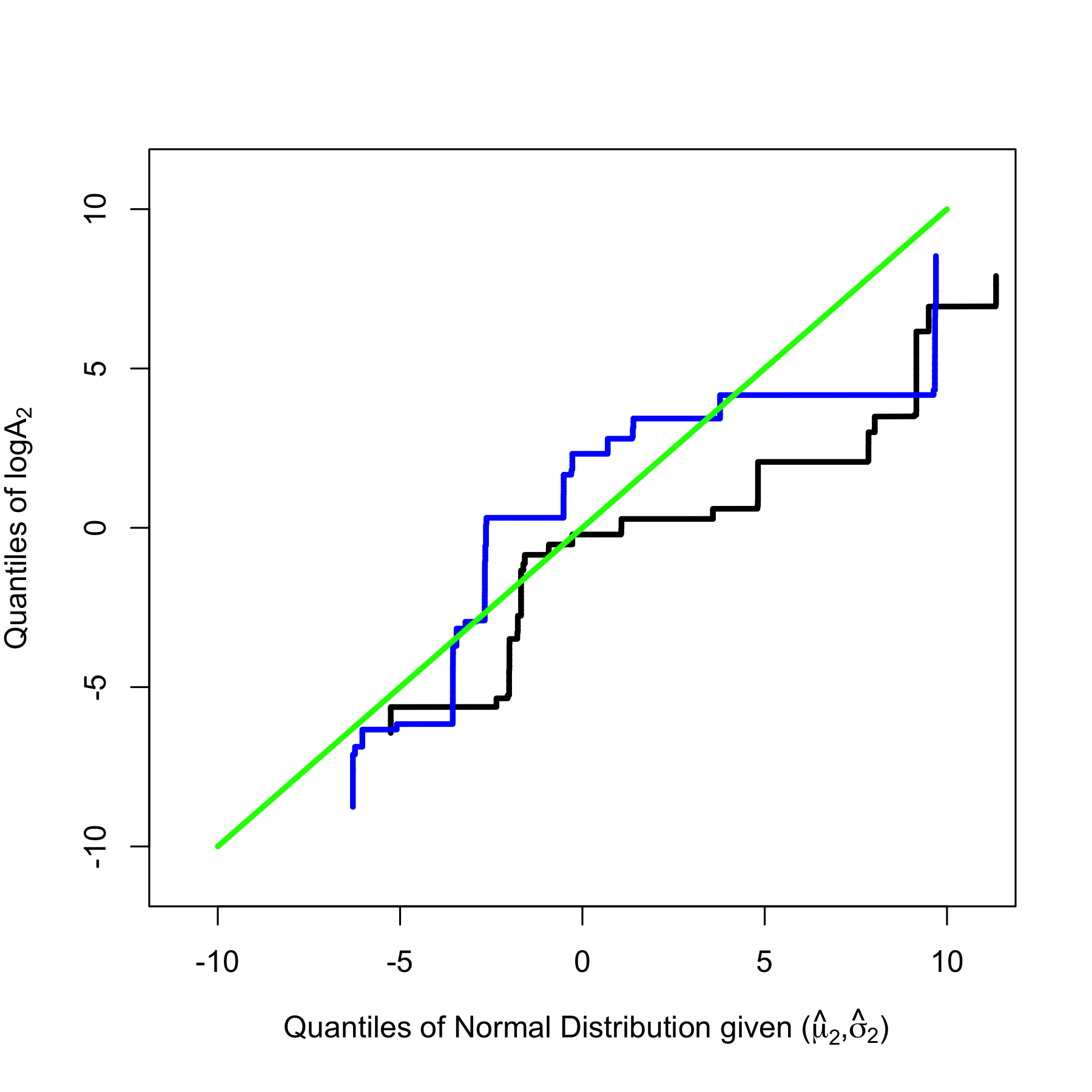}
      \end{subfigure}
      \caption{The Q-Q plots for two draws from the Dirichlet process given $\hat{\theta}$: On the left, the quantiles of $\log(A_1)$ are plotted against the quantiles of the normal distribution given $(\hat{\mu}_1,\hat{\sigma}_1)$; on the right, the quantiles of $\log(A_2)$ are plotted against the quantiles of the normal distribution given $(\hat{\mu}_2,\hat{\sigma}_2)$. In each figure, the green line is the 45-degree line.}
      \label{fig:qqdir}
\end{figure}

Figure \ref{fig:qqdir} illustrates how one might use the (estimated) hyperparameter for interpretation. Specifically, it is used to deduce the mean of the Dirichlet process, which is used as a benchmark for comparison with a drawn or filtered functional \mbox{parameter $\pi_m$.}

\subsection{Taste Distributions}

This section uses the filtering approach described in Section \ref{sec:filtration} to recover $\pi_m$. In the SARA model, the MRS restriction in \eqref{mrsrestrictions2} is:
\[
	\mathbb{E}_{\pi}[A_1\exp(-A'x_{it})]=p_{it}\, \mathbb{E}_{\pi}[A_2\exp(-A'x_{it})].
\]
When $A_1$ and $A_2$ are independent, this expression becomes:
\beq
\begin{gathered}
	\mathbb{E}_{\pi_1}[A_1\exp(-A_1x_{i1t})]\mathbb{E}_{\pi_2}[\exp(-A_2x_{i2t})] \\ =p_{it}\, \mathbb{E}_{\pi_1}[\exp(-A_1x_{i1t})]\mathbb{E}_{\pi_2}[A_2\exp(-A_2x_{i2t})].
\end{gathered}
\eeq
To filter $\pi_m$, these restrictions have to be imposed for every observation with positive consumption $x_{it}$ associated with segment $\Lambda_m$. In California, there are 688 MRS rest- rictions, and, in Florida, there are 880. Appendix \ref{app:filter} shows how to numerically solve the resulting optimization problem given the bilinearity of the MRS restrictions under independence.

The marginal taste distributions were filtered using a grid with 500 points between $\exp(-10)$ and $\exp(10)$, equally spaced on the log-scale. All draws from the estimated prior were simulated by the stick-breaking method given $J=1000$ \mbox{breaks (see Appen-} dix \ref{app:dir}). Exactly $S=100$ draws from the posterior were used to \mbox{filter each distribution.}

\begin{figure}
      \centering
      \begin{subfigure}[b]{0.48\linewidth}
      \centering
      \includegraphics[scale=0.4]{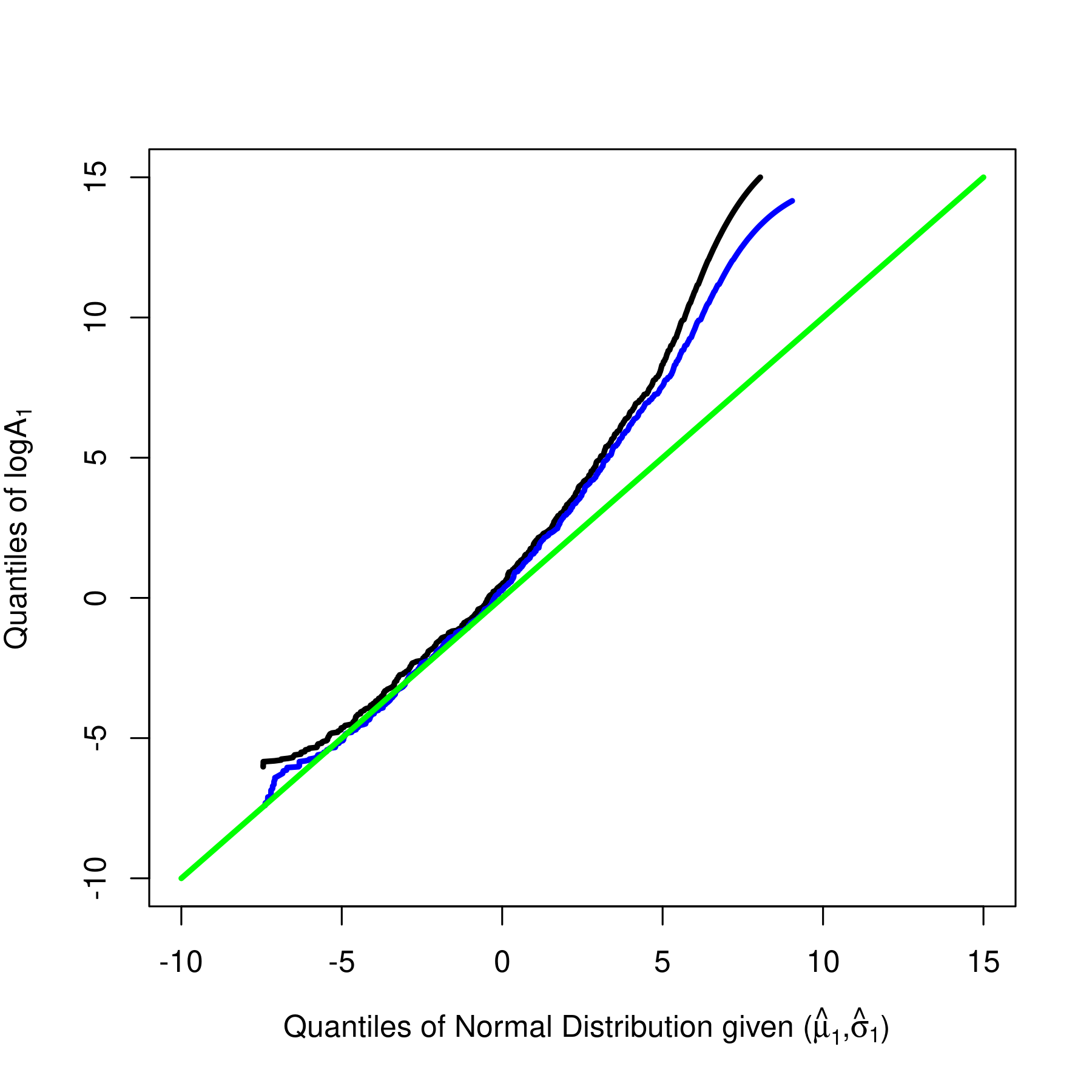}
      \end{subfigure}
      \begin{subfigure}[b]{0.48\linewidth}
      \centering
      \includegraphics[scale=0.4]{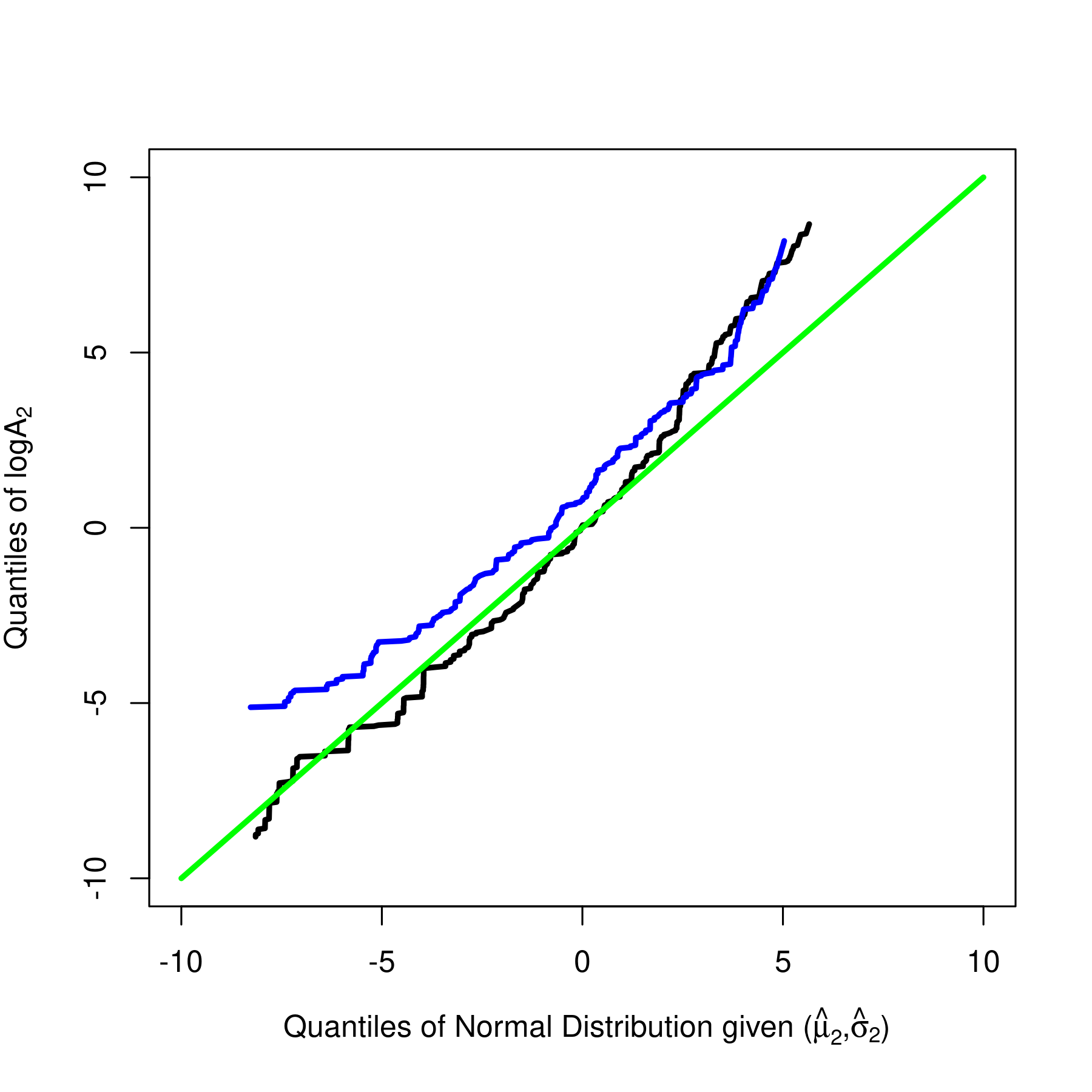}
      \end{subfigure}
      \caption{The Q-Q plots for the filtered taste distributions for California (black) and Florida (blue): On the left, the quantiles of $\log(A_1)$ are plotted against the quantiles of the normal distribution given $(\hat{\mu}_1,\hat{\sigma}_1)$; on the right, the quantiles of $\log(A_2)$ are plotted against the quantiles of the normal distribution given $(\hat{\mu}_2,\hat{\sigma}_2)$. In each figure, the green line is the 45-degree line.}
      \label{fig:qqstate}
\end{figure}

Figure \ref{fig:qqstate} displays the Q-Q plots for the filtered taste parameters $\hat{\pi}_m$ for California and Florida. As in Figure \ref{fig:qqdir}, the (estimated) hyperparameter is used to construct a benchmark for comparison. As expected, the filtered taste parameters \mbox{are rather diff-} erent from this benchmark. Here, the role of the estimated prior distribution diminishes with the number of observations. In both states, the slope on the left is steeper than the 45-degree line, suggesting that the posterior mean distribution for $A_1$ is more ``dispersed'' than its estimated prior mean distribution. The convexity of \mbox{these curves} also suggests fatter tails.

For the structural interpretation of these plots, assume that (i) the \mbox{preferences are} SARA, (ii) the taste parameters are independent, (iii) the marginal \mbox{distribution of $A_1$} is the same in both states, and (iv) the marginal distribution of $A_2$ ``shifts'' such that $\pi_2^*(A_2) = \pi_2(cA_2)$, where $\pi_2$ and $\pi_2^*$ denote the marginal distributions of $A_2$ in these states. Under these assumptions:
\beq
	U(x_1,x_2;\pi^*)=U(x_1,cx_2;\pi),
\eeq
and solving the utility maximization problem in \eqref{opt} yields:
\beq
\label{trans1}
	X_1(z;\pi^*)=X_1(cz;\pi) \; \; \text{and} \; \; X_2(z;\pi^*)=\left(\frac{1}{c}\right)X_2(cz;\pi).
\eeq
Similarly, if there is a ``shift'' in the marginal distribution of $A_1$ and the \mbox{marginal dist-} ribution of $A_2$ is the same in both states, we obtain:
\beq
\label{trans2}
	X_1(z;\pi^*)=\left(\frac{1}{c}\right)X_1\left(y,\frac{p}{c};\pi\right) \; \; \text{and} \; \; X_2(z;\pi^*)=X_2\left(y,\frac{p}{c};\pi\right).
\eeq
The relationships given in \eqref{trans1} and \eqref{trans2} suggest that there exists a complicated non-linear relationship between such demand functions. Therefore, we cannot immediately deduce from Figure \ref{fig:qqstate} which state has a higher demand for beer. For a more formal analysis, the utility functions associated with each posterior mean taste distribution must be used to derive a posterior MRS, or a posterior demand function.

\begin{figure}
	\centering
	\begin{subfigure}[b]{0.48\linewidth}
      \centering
      \includegraphics[scale=0.4]{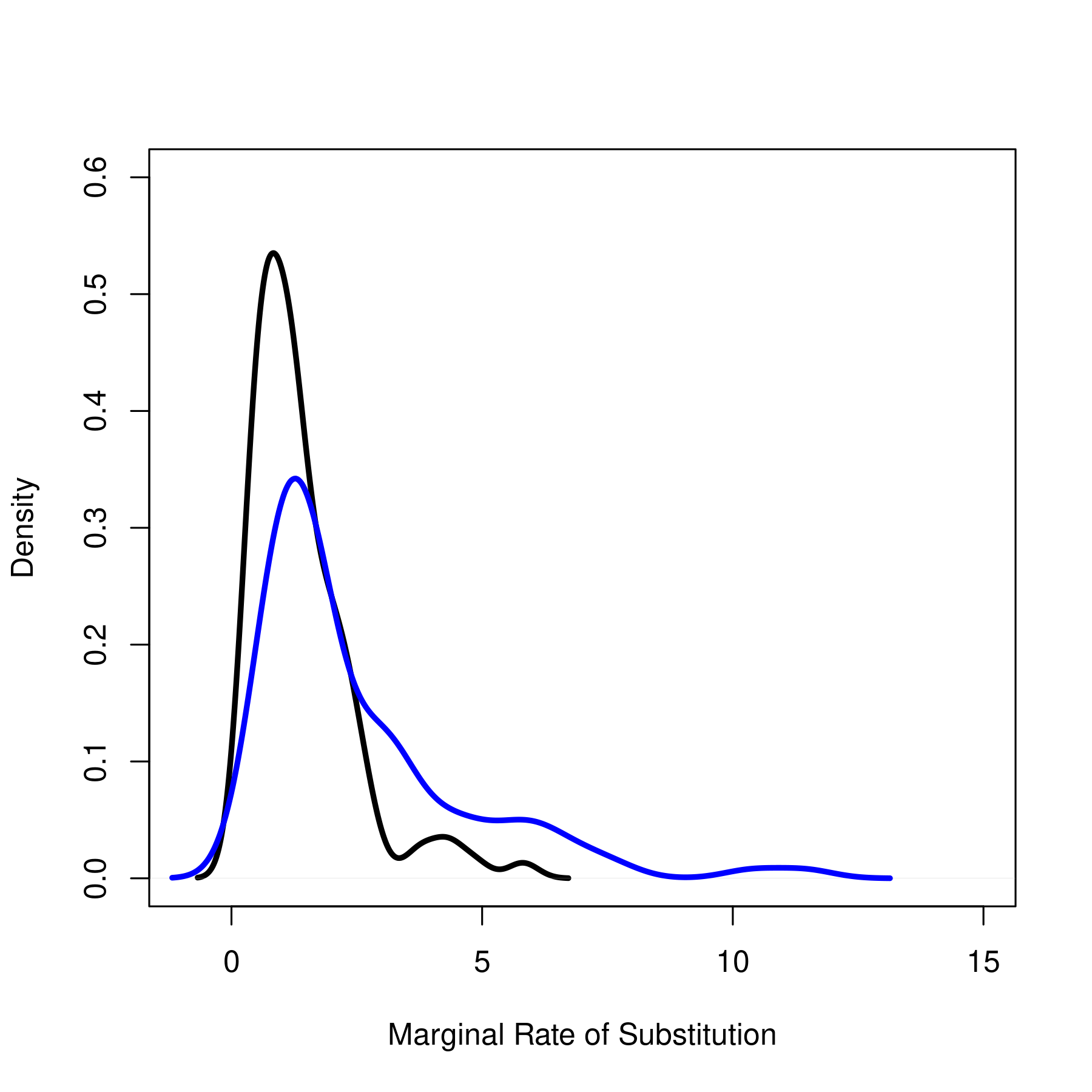}
      \end{subfigure}
      \begin{subfigure}[b]{0.48\linewidth}
      \centering
      \includegraphics[scale=0.4]{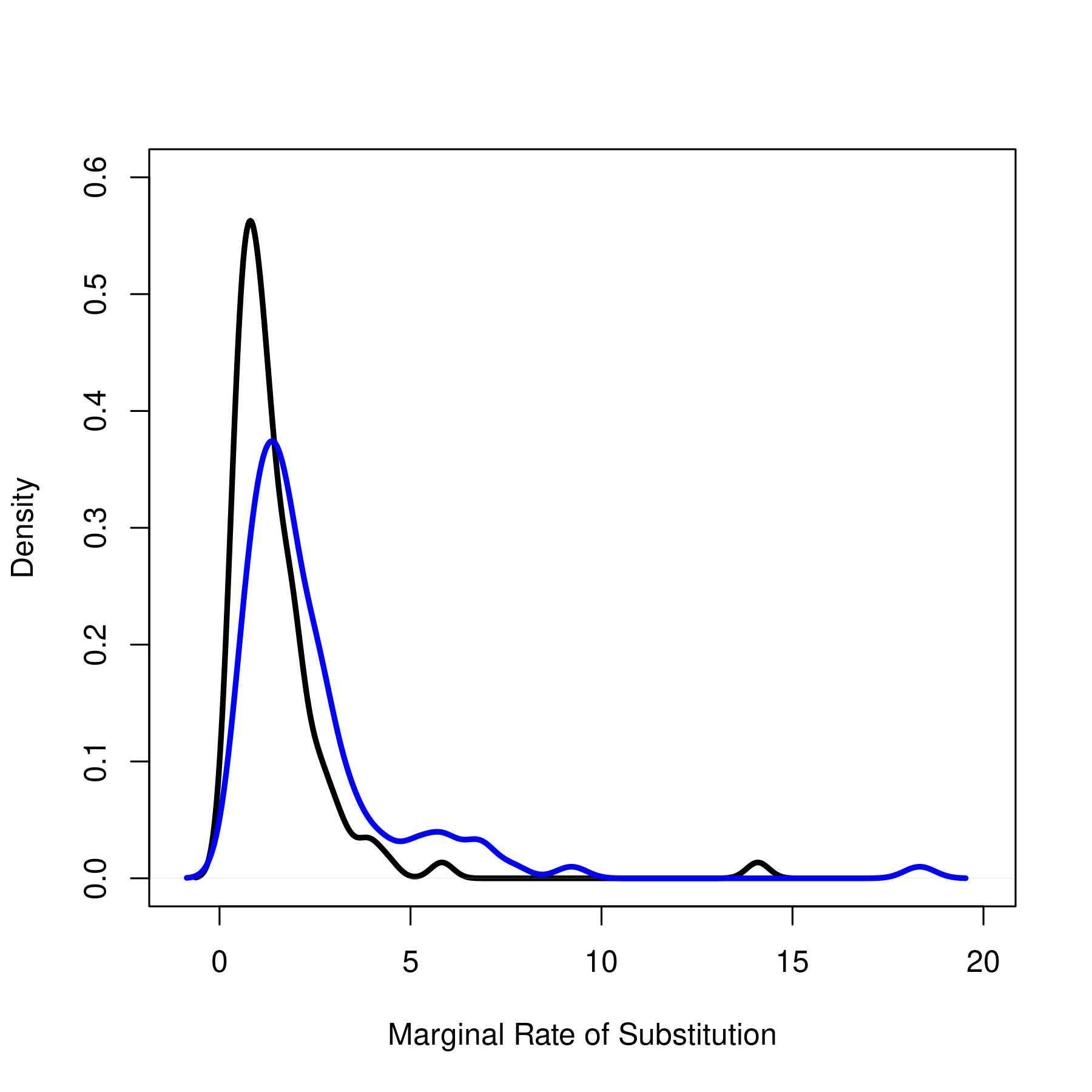}
      \end{subfigure}
	\caption{\textbf{Posterior Marginal Rate of Substitution.} The posterior distributions for $\text{MRS}(1,1;\pi)$ (black) and $\text{MRS}(1,2;\pi)$ (blue) for California (left) and Florida (blue).}
	\label{fig:postmrs}
\end{figure}

This analysis has to be completed with a discussion of accuracy. In \mbox{this non-param-} etric framework, the posterior distributions of $\pi_1$ and $\pi_2$ are infinite-indimesional and cannot be represented. However, posterior distributions of any scalar transformation of $\pi_1$ and $\pi_2$ can be derived using simulation. In this respect, it is important to know which scalar objects are of interest. Typically, we are interested in \mbox{the MRS evalu-} ated at a specific bundle, say $x_0$, or counterfactual demand, corresponding to a particular design, say $z_0=(y_0,p_0)$. Figure \ref{fig:postmrs} displays the posterior distributions of the MRS, evaluated at two bundles, $(1,1)$ and $(1,2)$, for California and Florida. In both states, these distributions are approximately log-normal (with is consistent with \citeauthor{dobronyi}, \citeyear{dobronyi}), and the posterior for $\text{MRS}(1,2;\pi)$ has a much longer tail than the posterior for $\text{MRS}(1,1;\pi)$, implying that, the quantity of good 2 that must be given to the consumer in order to compensate her for one unit of good 2 (and keep her just as happy) is larger, on average, when she has more of good 2. This tail is longer in California.

The filtered taste distributions in Figure \ref{fig:qqstate} are obtained by applying the algorithm in Appendix \ref{app:filter} and forcing the density $\pi_m$ to be non-negative at each iteration. The existence of negative ``probabilities'' can be a result of numerical uncertainty, the choice of grid, or misspecification. Specifically, it can arise if the consumer in segment $\Lambda_m$ does not maximize her SARA/SSF utility function (or any utility function) subject to the linear \mbox{budget constraint.} By analyzing these negative probabilities, we can construct a measure of the deviation from rationality. To illustrate, let $\pi_k^+=\max\{0,\pi_k\}$ and $\pi_k^-=\max\{0,-\pi_k\}$, respectively, denote the positive and negative components of the elementary probability $\pi_k$ associated with the $k^{th}$ grid point. The following ratio:
\beq
	\text{BR}=\frac{\sum_k\pi_k^-}{\sum_k(\pi_k^-+\pi_k^+)},
\eeq
is a measure of bounded rationality. This ratio ranges between 0 and 1. The closer this ratio is to 1, the less compatible the data are with the hundreds of MRS restrictions imposed by the chosen model. This ratio is related to a subset of the literature conce- rned with such measures. Existing measures include Afriat's Efficiency Index (\citeauthor{afriat}, \citeyear{afriat}; \citeauthor{varian-1990}, \citeyear{varian-1990}), and the Money Pump Index (\citeauthor{echenique}, \citeyear{echenique}). In general, these indices are used to measure a single consumer's deviation from rationality by evaluating how ``close'' her choices are to satisfying the Generalized Axiom of Revealed Preference (GARP), a necessary and sufficient condition for a finite number of choices to be consistent with the maximization of any locally non-satiated utility function. In our framework, the BR ratio can be used to measure the violation of the homogeneous segment assumption. Table \ref{tab:br} displays the BR ratios for California and Florida. The BR ratio for $\pi_1$ is smaller than the ratio for $\pi_2$ in each state; these ratios are roughly the same across states.

\begin{table}
      \centering
      \caption{BR ratios for California and Florida.}
      \begin{tabular}{ccc}
            \midrule \midrule
            State & $\pi_1$ & $\pi_2$ \\ \midrule
            California  & 0.15  & 0.20 \\
            Florida  & 0.17   & 0.21 \\ \midrule \midrule
      \end{tabular}
      \label{tab:br}
\end{table}

\section{Concluding Remarks}
\label{sec:conclusion}

This paper is one among pioneering papers attempting to tackle the challenges of performing structural demand analysis with scanner data (see also \citeauthor{burda-2008}, \citeyear{burda-2008}, \citeyear{burda-2012}, \citeauthor{craw-pol}, \citeyear{craw-pol}, \citeauthor{guha-ng}, \citeyear{guha-ng}, \citeauthor{cher-newey}, \citeyear{cher-newey}, and \hyperlinkcite{dobronyi}{Dobronyi and Gouri\'{e}roux}, \citeyear{dobronyi}). The recent availability of scanner \mbox{data permits} new developments in the analysis of consumer behaviour. Here, we have shown that, by introducing homogeneous segments of consumers, we can consider a model of consumption with non-parametric preferences and infinite-dimensional heterogeneity, not only from a theoretical point-of-view, but also from a practical one. The distribution of individual heterogeneity in the population can be estimated, and the underlying non-parametric preferences can be filtered by using appropriate algorithms.

We developed an analysis for \emph{two goods} for exposition. This feature of \mbox{our analysis} leaves the question: Can the methods developed in this paper be extended to a framework with, say, 100 goods? A completely \emph{unconstrained} non-parametric analysis wou- ld encounter the curse of dimensionality. Specifically, we would need to estimate the distribution of the utility function (a non-parametric function with, in this scenario, 100 arguments). This task would be infeasible, even in our big data framework. But, the SARA model with independent taste parameters is a \emph{constrained} non-parametric model. The structure of the SARA model reduces the non-parametric dimension of the problem, making it feasible. Indeed, when taste parameters are independent, we only need to estimate 100 one-dimensional distributions. A similar remark \mbox{applies to} the algorithm used to filter the taste distributions: The two steps based on the bilinear form of the MRS restrictions in a two good setting can be replaced with 100 successive steps based on the multilinear form of MRS restrictions in a 100 good setting, without increasing the numerical complexity.

Many of the results in this paper require taste parameters to be independent, but this requirement can be relaxed. For example, we can always consider a SARA model with the following form:
\beq
	U(x;A)=-\exp(-(A_c+A_{1})x_1-(A_c+A_{2})x_2),
\eeq
where $A_c$ is a common component, and $A_j$ is a good-specific taste parameter, for each $j=1,2$. In such a framework, independence between $A_c$, $A_1$, and $A_2$ does not imply independence between the parameters:
\beq
	A_1^*=A_c+A_1 \; \; \text{and} \; \; A_2^*=A_c+A_2,
\eeq
but it does reduce the dimensionality of the problem: Instead of introducing a joint distribution $\pi$ on a space of dimension 2, the model only depends on three distributions on a space of dimension 1. This specification avoids the curse of dimensionality. (See Appendix \ref{app:dependence} for a discussion of identification in this case with taste dependence.)

In this paper, consumers are assumed to be rational, and divided into homogeneous segments. Since, in each segment, the demand function can be non-parametrically estimated over a subset of its domain, the analysis can be continued to develop a test of the homogeneity of each segment, or, more generally, a non-parametric method for constructing homogeneous segments.

The approach developed in this paper uses standard ideas from consumer theory to make inference. This approach is appropriate when both quantities and prices have continuous supports. This feature makes this approach valid for some \mbox{levels of good,} consumer, and date aggregation. Hence, this approach can be used for, say, evaluating the effect of alcohol tax on alcohol consumption, but unreasonable for analyzing how a particular consumer will choose between hundreds of different \mbox{brands of whiskey. To} our knowledge, the tools needed to solve such a problem have not been developed yet.

\nocite{rock,manning-et-al,brown,karni-1983,grant,blundell-wp,cher-newey,ng,guha-ng,bilinear,bilinear-vanrosen,ledoit,heyde,laplace-lognormal,blundell}

\bibliography{references}

\begin{thebibliography}{62}
\newcommand{\enquote}[1]{``#1''}
\expandafter\ifx\csname natexlab\endcsname\relax\def\natexlab#1{#1}\fi

\bibitem[\protect\citeauthoryear{Afriat}{Afriat}{1967}]{afriat}
\textsc{Afriat, S.} (1967): \enquote{The Construction of Utility Functions from
  Expenditure Data,} \emph{International Economic Review}, 8, 67--77.

\bibitem[\protect\citeauthoryear{Alley, Ferguson, and Stewart}{Alley
  et~al.}{1992}]{all-ferg-stew}
\textsc{Alley, A., D.~Ferguson, and K.~Stewart} (1992): \enquote{An Almost
  Ideal Demand System for Alcoholic Beverages in British Columbia,}
  \emph{Empirical Economics}, 17, 401--418.

\bibitem[\protect\citeauthoryear{Asmussen, Jensen, and
  Rojas-Nandayapa}{Asmussen et~al.}{2016}]{laplace-lognormal}
\textsc{Asmussen, S., J.~Jensen, and L.~Rojas-Nandayapa} (2016): \enquote{On
  the Laplace Transform of the Lognormal Distribution,} \emph{Methodology and
  Computing in Applied Probability}, 18, 441--458.

\bibitem[\protect\citeauthoryear{Banks, Blundell, and Lewbel}{Banks
  et~al.}{1997}]{bbl}
\textsc{Banks, J., R.~Blundell, and A.~Lewbel} (1997): \enquote{Quadratic Engel
  Curves and Consumer Demand,} \emph{The Review of Economics and Statistics},
  79, 527--539.

\bibitem[\protect\citeauthoryear{Barten}{Barten}{1968}]{barten}
\textsc{Barten, A.} (1968): \enquote{Estimating Demand Equations,}
  \emph{Econometrica}, 36, 213--251.

\bibitem[\protect\citeauthoryear{Beckert and Blundell}{Beckert and
  Blundell}{2008}]{beckert-blundell}
\textsc{Beckert, W. and R.~Blundell} (2008): \enquote{Heterogeneity and the
  Non-Parametric Analysis of Consumer Choice: Conditions for Invertibility,}
  \emph{The Review of Economic Studies}, 75, 1069--1080.

\bibitem[\protect\citeauthoryear{Bergstra and Yoshua}{Bergstra and
  Yoshua}{2012}]{bergstra}
\textsc{Bergstra, J. and Y.~Yoshua} (2012): \enquote{Random Search for
  Hyper-Parameter Optimization,} \emph{Journal of Machine Learning Research},
  13, 281--305.

\bibitem[\protect\citeauthoryear{Blaschke and Pick}{Blaschke and
  Pick}{1916}]{blaschke-pick}
\textsc{Blaschke, W. and G.~Pick} (1916): \enquote{Distanzsch\"{a}tzungen im
  Funktionenraum II,} \emph{Mathematische Annalen}, 77, 277--302.

\bibitem[\protect\citeauthoryear{Blomquist, Kumar, Liang, and Newey}{Blomquist
  et~al.}{2015}]{blom-kumar-more}
\textsc{Blomquist, S., A.~Kumar, C.~Liang, and W.~Newey} (2015):
  \enquote{Individual Heterogeneity, Nonlinear Budget Sets, and Taxable
  Income,} \emph{CESifo Working Paper}.

\bibitem[\protect\citeauthoryear{Blundell, Horowitz, and Parey}{Blundell
  et~al.}{2017{\natexlab{a}}}]{blundell}
\textsc{Blundell, R., J.~Horowitz, and M.~Parey} (2017): \enquote{Nonparametric
  Estimation of a Nonseparable Demand Function under the Slutsky Inequality
  Restriction,} \emph{The Review of Economics and Statistics}, 99, 291--304.

\bibitem[\protect\citeauthoryear{Blundell, Kristensen, and Matzkin}{Blundell
  et~al.}{2017{\natexlab{b}}}]{blundell-wp}
\textsc{Blundell, R., D.~Kristensen, and R.~Matzkin} (2017):
  \enquote{Individual Counterfactuals with Multidimensional Unobserved
  Heterogeneity,} \emph{CeMMAP Working Paper}.

\bibitem[\protect\citeauthoryear{Brown and Walker}{Brown and
  Walker}{1989}]{brown}
\textsc{Brown, B. and M.~Walker} (1989): \enquote{The Random Utility Hypothesis
  and Inference in Demand Systems,} \emph{Econometrica}, 57, 815--829.

\bibitem[\protect\citeauthoryear{Brown and Matzkin}{Brown and
  Matzkin}{1995}]{brown-matzkin}
\textsc{Brown, D. and R.~Matzkin} (1995): \enquote{The Random Utility Model
  from Data on Consumer Demand,} \emph{Discussion Paper, Yale University}.

\bibitem[\protect\citeauthoryear{Burda, Harding, and Hausman}{Burda
  et~al.}{2008}]{burda-2008}
\textsc{Burda, M., M.~Harding, and J.~Hausman} (2008): \enquote{A Bayesian
  Mixed Logit-Probit Model for Multinomial Choice,} \emph{Journal of
  Econometrics}, 147, 232--246.

\bibitem[\protect\citeauthoryear{Burda, Harding, and Hausman}{Burda
  et~al.}{2012}]{burda-2012}
---\hspace{-.1pt}---\hspace{-.1pt}--- (2012): \enquote{A Poisson Mixture Model
  of Discrete Choice,} \emph{Journal of Econometrics}, 166, 184--203.

\bibitem[\protect\citeauthoryear{Chernozhukov, Hausman, and Newey}{Chernozhukov
  et~al.}{2020}]{cher-newey}
\textsc{Chernozhukov, V., J.~Hausman, and W.~Newey} (2020): \enquote{Demand
  Analysis with Many Prices,} \emph{Working Paper}.

\bibitem[\protect\citeauthoryear{Christensen, Jorgenson, and Lau}{Christensen
  et~al.}{1975}]{transcendental}
\textsc{Christensen, L., D.~Jorgenson, and L.~Lau} (1975):
  \enquote{Transcendental Logarithmic Utility Functions,} \emph{The American
  Economic Review}, 65, 367--383.

\bibitem[\protect\citeauthoryear{Crawford and Polisson}{Crawford and
  Polisson}{2016}]{craw-pol}
\textsc{Crawford, I. and M.~Polisson} (2016): \enquote{Demand Analysis with
  Partially Observed Prices,} \emph{University of Leicester Working Paper},
  15/12.

\bibitem[\protect\citeauthoryear{Deaton and Muellbauer}{Deaton and
  Muellbauer}{1980}]{aids}
\textsc{Deaton, A. and J.~Muellbauer} (1980): \enquote{An Almost Ideal Demand
  System,} \emph{The American Economic Review}, 70, 312--326.

\bibitem[\protect\citeauthoryear{Dette, Hoderlein, and Neumeyer}{Dette
  et~al.}{2016}]{dette}
\textsc{Dette, H., S.~Hoderlein, and N.~Neumeyer} (2016): \enquote{Testing
  Multivariate Economic Restrictions Using Quantiles: The Example of Slutsky
  Negative Semidefiniteness,} \emph{Journal of Econometrics}, 191, 129--144.

\bibitem[\protect\citeauthoryear{Dobronyi and Gouri\'{e}roux}{Dobronyi and
  Gouri\'{e}roux}{2020}]{dobronyi}
\textsc{Dobronyi, C. and C.~Gouri\'{e}roux} (2020): \enquote{Stochastic
  Revealed Preference: A Non-Parametric Analysis,} \emph{Working Paper}.

\bibitem[\protect\citeauthoryear{Echenique, Lee, and Shum}{Echenique
  et~al.}{2011}]{echenique}
\textsc{Echenique, F., S.~Lee, and M.~Shum} (2011): \enquote{The Money Pump as
  a Measure of Revealed Preference Violations,} \emph{Journal of Political
  Economy}, 119, 1201--1223.

\bibitem[\protect\citeauthoryear{Feller}{Feller}{1968}]{feller-1968}
\textsc{Feller, W.} (1968): \emph{An Introduction to Probability Theory and its
  Applications}, vol.~2, Wiley.

\bibitem[\protect\citeauthoryear{Ferguson}{Ferguson}{1974}]{ferguson}
\textsc{Ferguson, T.} (1974): \enquote{Prior Distributions on Spaces of
  Probability Measures,} \emph{The Annals of Statistics}, 2, 615--629.

\bibitem[\protect\citeauthoryear{Geweke}{Geweke}{2012}]{geweke}
\textsc{Geweke, J.} (2012): \enquote{Nonparametric Bayesian Modelling of
  Monotone Preferences for Discrete Choice Experiments,} \emph{The Journal of
  Econometrics}, 171, 185--204.

\bibitem[\protect\citeauthoryear{Gouri\'{e}roux and Monfort}{Gouri\'{e}roux and
  Monfort}{1996}]{gou-mon}
\textsc{Gouri\'{e}roux, C. and A.~Monfort} (1996): \emph{Simulation-Based
  Econometric Methods}, New York: Oxford University Press.

\bibitem[\protect\citeauthoryear{Gouri{\'{e}}roux, Monfort, and
  Renault}{Gouri{\'{e}}roux et~al.}{1990}]{bilinear}
\textsc{Gouri{\'{e}}roux, C., A.~Monfort, and E.~Renault} (1990):
  \enquote{Bilinear Constraints: Estimation and Tests,} \emph{Essays in Honor
  of Edmond Malinvaud, Empirical Economics, MIT Press}, 166--191.

\bibitem[\protect\citeauthoryear{Grant}{Grant}{1995}]{grant}
\textsc{Grant, S.} (1995): \enquote{A Strong (Ross) Characterization of
  Multivariate Risk Aversion,} \emph{Theory and Decision}, 38, 131--152.

\bibitem[\protect\citeauthoryear{Guha and Ng}{Guha and Ng}{2019}]{guha-ng}
\textsc{Guha, R. and S.~Ng} (2019): \enquote{A Machine Learning Analysis of
  Seasonal and Cyclical Sales in Weekly Scanner Data,} \emph{NBER 25899}.

\bibitem[\protect\citeauthoryear{Halevy and Feltkamp}{Halevy and
  Feltkamp}{2005}]{halevy-feltkamp}
\textsc{Halevy, Y. and V.~Feltkamp} (2005): \enquote{A Bayesian Approach to
  Uncertainty Aversion,} \emph{Review of Economic Studies}, 72, 449--466.

\bibitem[\protect\citeauthoryear{Hausman and Newey}{Hausman and
  Newey}{2016}]{hn-individual-het}
\textsc{Hausman, J. and W.~Newey} (2016): \enquote{Individual Heterogeneity and
  Average Welfare,} \emph{Econometrica}, 84, 1225--1248.

\bibitem[\protect\citeauthoryear{Heyde}{Heyde}{1963}]{heyde}
\textsc{Heyde, C.} (1963): \enquote{On a Property of the Lognormal
  Distribution,} \emph{Journal of the Royal Statistical Society Series B:
  Methodological}, 29, 16--18.

\bibitem[\protect\citeauthoryear{Hosoya}{Hosoya}{2016}]{hosoya2016}
\textsc{Hosoya, Y.} (2016): \enquote{On First-Order Partial Differential
  Equations: An Existence Theorem and its Applications,} \emph{Advances in
  Mathematical Economics}, 20, 77--87.

\bibitem[\protect\citeauthoryear{Hurwicz and Uzawa}{Hurwicz and
  Uzawa}{1971}]{hurwicz-uzawa}
\textsc{Hurwicz, L. and H.~Uzawa} (1971): \enquote{On the Integrability of
  Demand Functions,} in \emph{Preferences, Utility, and Demand: A Minnesota
  symposium}, New York, chap.~6, 114--148.

\bibitem[\protect\citeauthoryear{Johansen}{Johansen}{1969}]{johansen-relationship}
\textsc{Johansen, L.} (1969): \enquote{On the Relationships Between Some
  Systems of Demand Functions,} \emph{Liiketaloudellinen Aikakauskirga}, 1,
  30--41.

\bibitem[\protect\citeauthoryear{Johansen}{Johansen}{1974}]{johansen}
\textsc{Johansen, S.} (1974): \enquote{The Extremal Convex Functions,}
  \emph{Mathematica Scandinavica}, 34, 61--68.

\bibitem[\protect\citeauthoryear{Karni}{Karni}{1979}]{karni-1979}
\textsc{Karni, E.} (1979): \enquote{On Multivariate Risk Aversion,}
  \emph{Econometrica}, 47, 1391--1401.

\bibitem[\protect\citeauthoryear{Karni}{Karni}{1983}]{karni-1983}
---\hspace{-.1pt}---\hspace{-.1pt}--- (1983): \enquote{On the Correspondence
  Between Multivariate Risk Aversion and Risk Aversion with State-Dependent
  Preferences,} \emph{Journal of Economic Theory}, 30, 230--242.

\bibitem[\protect\citeauthoryear{Katzner}{Katzner}{1968}]{katzner}
\textsc{Katzner, D.} (1968): \enquote{A Note on the Differentiability of
  Consumer Demand Functions,} \emph{Econometrica}, 36, 415--418.

\bibitem[\protect\citeauthoryear{Kitamura and Stutzer}{Kitamura and
  Stutzer}{1997}]{kitamura-stutzer}
\textsc{Kitamura, Y. and M.~Stutzer} (1997): \enquote{An Information-Theoretic
  Alternative to Generalized Method of Moments Estimation,}
  \emph{Econometrica}, 65, 861--874.

\bibitem[\protect\citeauthoryear{Kotz, Balakrishnan, and Johnson}{Kotz
  et~al.}{2000}]{kotz}
\textsc{Kotz, S., N.~Balakrishnan, and N.~Johnson} (2000): \emph{Continuous
  Multivariate Distributions, Volume 1: Models and Applications}, New York:
  Wiley, 2nd ed.

\bibitem[\protect\citeauthoryear{Ledoit and Wolf}{Ledoit and
  Wolf}{2004}]{ledoit}
\textsc{Ledoit, O. and M.~Wolf} (2004): \enquote{A Well-Conditioned Estimator
  for Large-Dimensional Covariance Matrices,} \emph{Journal of Multivariate
  Analysis}, 88, 365--411.

\bibitem[\protect\citeauthoryear{Li, Schofield, and G\"{o}nen}{Li
  et~al.}{2019}]{li-2019}
\textsc{Li, Y., E.~Schofield, and M.~G\"{o}nen} (2019): \enquote{A Tutorial on
  Dirichlet Process Mixture Modeling,} \emph{Journal of Mathematical
  Psychology}, 91, 128--144.

\bibitem[\protect\citeauthoryear{Lin}{Lin}{2016}]{lin-dirichlet}
\textsc{Lin, J.} (2016): \enquote{On the Dirichlet Distribution,} \emph{Queens
  University, Kingston}.

\bibitem[\protect\citeauthoryear{Manning, Blumberg, and Moulton}{Manning
  et~al.}{1995}]{manning-et-al}
\textsc{Manning, W., L.~Blumberg, and L.~Moulton} (1995): \enquote{The Demand
  for Alcohol: The Differential Response to Price,} \emph{Journal of Health
  Economics}, 14, 123--148.

\bibitem[\protect\citeauthoryear{Matzkin}{Matzkin}{2003}]{matzkin-nonadd}
\textsc{Matzkin, R.} (2003): \enquote{Nonparametric Estimation of Nonadditive
  Random Functions,} \emph{Econometrica}, 71, 1339--1376.

\bibitem[\protect\citeauthoryear{McFadden}{McFadden}{1989}]{mcfadden-msm}
\textsc{McFadden, D.} (1989): \enquote{A Method of Simulated Moments for
  Estimation of Discrete Response Models Without Numerical Integration,}
  \emph{Econometrica}, 57, 995--1026.

\bibitem[\protect\citeauthoryear{Moschini}{Moschini}{1998}]{moschini}
\textsc{Moschini, G.} (1998): \enquote{The Semiflexible Almost Ideal Demand
  System,} \emph{European Economic Review}, 42, 349--364.

\bibitem[\protect\citeauthoryear{Navarro, Griffiths, Steyvers, and Lee}{Navarro
  et~al.}{2006}]{navarro}
\textsc{Navarro, D., T.~Griffiths, M.~Steyvers, and M.~Lee} (2006):
  \enquote{Modeling Individual Differences Using Dirichlet Processes,}
  \emph{Journal of Mathematical Psychology}, 50, 101--122.

\bibitem[\protect\citeauthoryear{Ng}{Ng}{2017}]{ng}
\textsc{Ng, S.} (2017): \enquote{Opportunities and Challenges: Lessons from
  Analyzing Terabytes of Scanner Data,} in \emph{Advances in Economics and
  Econometrics}, Cambridge, vol.~2.

\bibitem[\protect\citeauthoryear{Pakes and Pollard}{Pakes and
  Pollard}{1989}]{pakes-optimization}
\textsc{Pakes, A. and D.~Pollard} (1989): \enquote{Simulation and the
  Asymptotics of Optimization Estimators,} \emph{Econometrica}, 57, 1027--1057.

\bibitem[\protect\citeauthoryear{Polyanskiy and Wu}{Polyanskiy and
  Wu}{2017}]{informationtheory}
\textsc{Polyanskiy, Y. and Y.~Wu} (2017): \enquote{Lecture Notes on Information
  Theory,} \emph{Yale University}.

\bibitem[\protect\citeauthoryear{Richard}{Richard}{1975}]{richard-1975}
\textsc{Richard, S.} (1975): \enquote{Multivariate Risk Aversion, Utility
  Independence and Separable Utility Functions,} \emph{Management Science}, 22,
  12--21.

\bibitem[\protect\citeauthoryear{Rockafellar}{Rockafellar}{1970}]{rock}
\textsc{Rockafellar, R.} (1970): \emph{Convex Analysis}, Princeton: Princeton
  University Press.

\bibitem[\protect\citeauthoryear{Rolin}{Rolin}{1992}]{rolin}
\textsc{Rolin, J.} (1992): \enquote{Some Useful Properties of the Dirichlet
  Process,} \emph{CORE Discussion Paper, Universit\'{e} Catholique de Louvain}.

\bibitem[\protect\citeauthoryear{Roy}{Roy}{1952}]{roy-safety}
\textsc{Roy, A.} (1952): \enquote{Safety First and the Holding of Assets,}
  \emph{Econometrica}, 20, 431--449.

\bibitem[\protect\citeauthoryear{Samuelson}{Samuelson}{1948}]{samuelson1948}
\textsc{Samuelson, P.} (1948): \enquote{Consumption Theory in Terms of Revealed
  Preference,} \emph{Economica}, 15, 243--253.

\bibitem[\protect\citeauthoryear{Samuelson}{Samuelson}{1950}]{samuelson}
---\hspace{-.1pt}---\hspace{-.1pt}--- (1950): \enquote{The Problem of
  Integrability in Utility Theory,} \emph{Economica}, 17, 355--385.

\bibitem[\protect\citeauthoryear{Sethuraman}{Sethuraman}{1994}]{sethuraman}
\textsc{Sethuraman, J.} (1994): \enquote{A Constructive Definition of Dirichlet
  Priors,} \emph{Statistica Sinica}, 4, 639--650.

\bibitem[\protect\citeauthoryear{Theil and Neudecker}{Theil and
  Neudecker}{1958}]{theil-neudecker}
\textsc{Theil, H. and H.~Neudecker} (1958): \enquote{Substitution,
  Complementarity, and the Residual Variation around Engel Curves,} \emph{The
  Review of Economic Studies}, 25, 114--123.

\bibitem[\protect\citeauthoryear{{Van Rosen}}{{Van
  Rosen}}{2018}]{bilinear-vanrosen}
\textsc{{Van Rosen}, D.} (2018): \emph{Bilinear Regression Analysis: An
  Introduction}, Springer.

\bibitem[\protect\citeauthoryear{Varian}{Varian}{1990}]{varian-1990}
\textsc{Varian, H.} (1990): \enquote{Goodness-of-Fit in Optimizing Models,}
  \emph{The Journal of Econometrics}, 46, 125--140.

\end{thebibliography}

\appendix

\section{The Dirichlet Process}
\label{app:dir}

In this appendix, we briefly review the definition and properties of the Dirichlet proc- ess, and then describe how to simulate from the Dirichlet process (see \citeauthor{ferguson}, \citeyear{ferguson}, \citeauthor{rolin}, \citeyear{rolin}, \citeauthor{sethuraman}, \citeyear{sethuraman}, \citeauthor{lin-dirichlet}, \citeyear{lin-dirichlet}, and \citeauthor{li-2019}, \citeyear{li-2019}).

\subsection{Definition and Properties of the Dirichlet Process}

For exposition, let us describe the Dirichlet distribution, then the Dirichlet process:
\bi
\item \textbf{Dirichlet Distribution:}

Let $D_J(\alpha)$ denote the $J$-dimensional Dirichlet distribution with density:
\beq
		f_{\alpha}(q)=\frac{\Gamma\big(\sum_{j=1}^J\alpha_j\big)\prod_{j=1}^J q_j^{\alpha_j}}{\prod_{j=1}^J\Gamma(\alpha_j)},
\eeq
for every $q\in [0,1]^J$ such that $\sum_{j=1}^Jq_j=1$, where $\alpha\in \mathbb{R}_{++}^J$ denotes a $J$-dimensional vector of positive parameters. If a random vector $(Q_1,\dots,Q_J)$ has a Dirichlet distribution $D_J(\alpha)$, then:
\beq
\label{mean:var:dir}
	\mathbb{E}[Q_j]=\bar{\alpha}_j \quad \text{and} \quad V(Q_j)=\frac{\bar{\alpha}_j(1-\bar{\alpha}_j)}{1+\sum_{j=1}^J\alpha_j},	
\eeq
where $\bar{\alpha}_j=\alpha_j/\sum_{j=1}^J\alpha_j$.

\item \textbf{Dirichlet Process:}

In the SARA and SSF models, there are two taste parameters, $A_1$ and $A_2$. The probability distribution $\pi$ of $(A_1,A_2)$ is defined on $\mathbb{R}_+^2$. Therefore, in this section, we describe the Dirichlet process in this special case. Let $\mathscr{B}_0$ denote the Borel sets associated with $\mathbb{R}_+^2$, $\mathscr{F}$ denote the set of probability measures defined on $(\mathbb{R}_+^2,\mathscr{B}_0)$, and $\mathscr{B}_1$ denote the $\sigma$-algebra consisting of the Borel sets associated with the topology of weak convergence on $\mathscr{F}$. Let $\mu$ denote a (deterministic) probability measure defined on $(\mathbb{R}_+^2,\mathscr{B}_0)$, and let $c$ denote a strictly positive scalar. A process $G$ with values in $\mathscr{F}$ is a Dirichlet process with functional parameter $\mu$ and scaling parameter $c$ if, for every finite and measurable partition $\{C_1,\dots,C_J\}$ of $\mathbb{R}_+^2$, the random vector $\big[G(C_1),\dots,G(C_J)\big]'$ has a $J$-dimensional Dirichlet distribution given $\alpha=\big[c\mu(C_1),\dots,c\mu(C_J)\big]'$. There exists a Dirichlet process for every probability measure $\mu$ defined on $(\mathbb{R}_+^2,\mathscr{B}_0)$ and scaling parameter $c$. The distribution of the Dirichlet process is a probability measure defined on $(\mathscr{F},\mathscr{B}_1)$, whose realizations are almost surely \emph{discrete} probability measures defined on $(\mathbb{R}_{+}^2,\mathscr{B}_0)$, assigning probability one to the set of all discrete probability measures defined on $(\mathbb{R}_{+}^2,\mathscr{B}_0)$. The support of the distribution of the Dirichlet process is a set of distributions with support contained in the support of $c\mu$ \citep{ferguson}. The functional parameter $\mu$ (sometimes called the base distribution) can be thought of as the mean of the Dirichlet process---indeed, for any measurable set $C$ in $\mathbb{R}_+^2$, the mean of the Dirichlet distribution in \eqref{mean:var:dir} yields $\mathbb{E}[G(C)]=\mu(C)$. Therefore, in our framework, $\mu$ represents the expected uncertainty on taste parameters. Intuitively, the scaling parameter $c$ describes the ``strength'' of discretization: When $c$ is large, the realizations of the Dirichlet process are concentrated around $\mu$; loosely, as $c$ tends to infinity, the realizations become ``more continuous.''
\ee

\subsection{Simulating a Dirichlet Process}

A Dirichlet process is easy to simulate given $\mu$ and $c$. There are a number of ways to simulate a realization---this section outlines the \emph{stick-breaking method}, appropriate for drawing under the independence of $A_1$ and $A_2$, based on the construction of the Dirichlet process in \citet{sethuraman}.

Let $B(\alpha_1,\alpha_2)$ denote the beta distribution with continuous density:
\beq
	f(q)=\frac{\Gamma(\alpha_1+\alpha_2)q_1^{\alpha_1}q_2^{\alpha_2}}{\Gamma(\alpha_1)\Gamma(\alpha_2)},
\eeq
on the simplex $\{(q_1,q_2)\geq 0:q_1+q_2=1\}$, in which $\Gamma$ denotes the gamma function, and $\alpha_1,\alpha_2>0$ are positive scalar parameters. Under the independence of $A_1$ and $A_2$, it is sufficient to be able to make a draw from a Dirichlet process whose realizations are distributions on $[0,\infty)$. Let $\mu^*$ and $c^*$ denote the mean and scaling parameter of this Dirichlet process. We can simulate from this process by using the \mbox{following steps:}
\begin{steps}
	\item For large $L\geq 1$, independently simulate $W_1,\dots,W_L\sim B(1,c^*)$.
	\item Compute $W_1^*=W_1$, and:
	\beq
		W_{\ell}^*=W_{\ell}\prod_{j=1}^{\ell-1}(1-W_j), \; \; \forall \ell=2,\dots,L.
	\eeq
	\item Independently simulate $V_1,\dots,V_J\sim \mu^*$.
	\item Define:
	\beq
		G(C)=\sum_{\ell=1}^LW_{\ell}^*\delta_{V_{\ell}}(C), \; \; \forall C\subseteq\mathbb{R}_{+}^2,
	\eeq
	where $\delta_v$ denotes a point mass at $v\in\mathbb{R}_{+}^2$.
\end{steps}
Theoretically, if we could simulate an infinite number of draws, then this procedure would produce a realization of the Dirichlet process associated with functional parameter $\mu^*$ and scaling parameter $c^*$. Since $L$ is finite, the resulting probability measure $G$ is a \emph{truncated} approximation of a realization of such a process. Figure \ref{fig:dirsim} displays a simulated realization from the Dirichlet process given log-normal $\mu^*$ with mean 0 and standard deviation 1 (where these parameters are interpreted on the log-scale) and scaling parameter $c=100$. This realization was simulated using the stick-breaking method given $L=100$.

\begin{figure}
	\centering
	\includegraphics[scale=0.5]{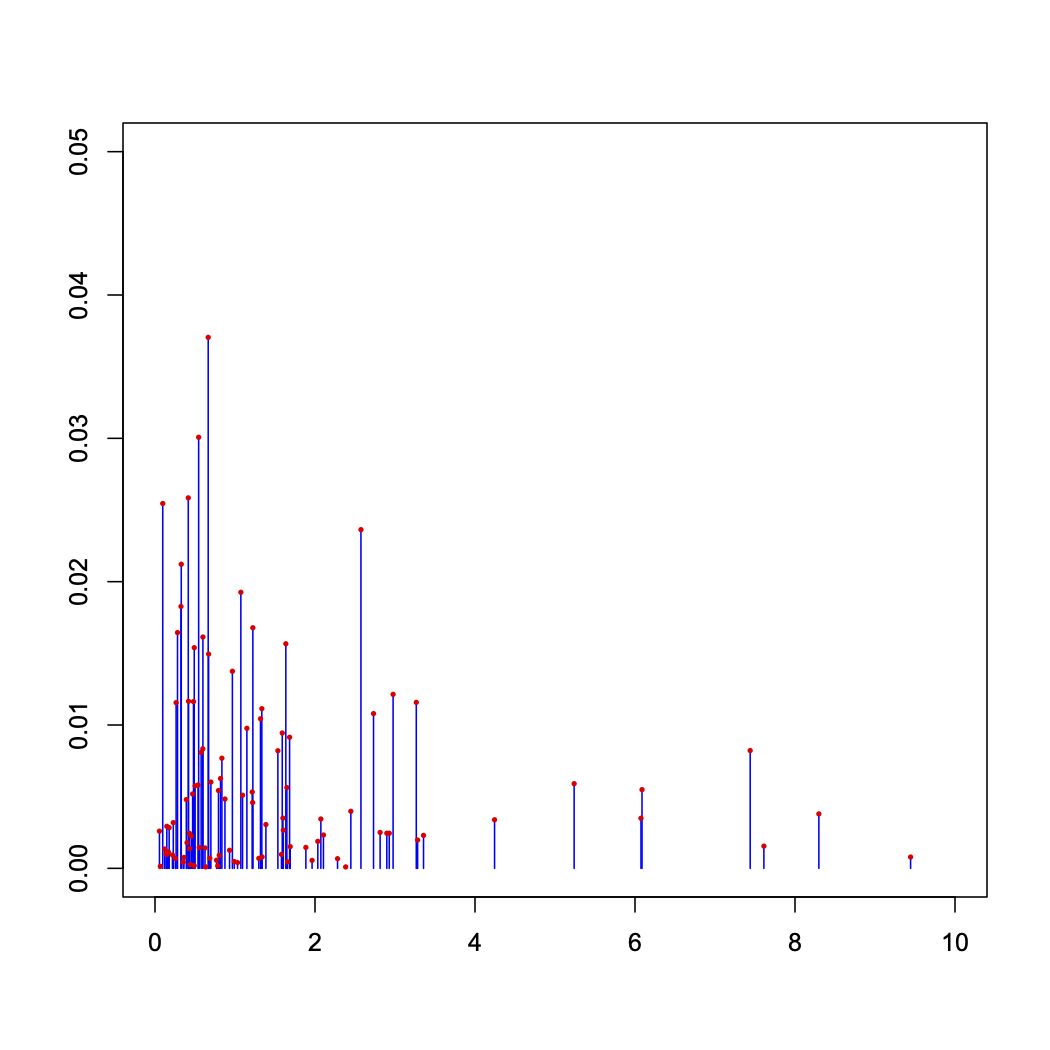}
	\caption{A simulated realization from the Dirichlet process given log-normal $\mu^*$ with mean 0 and standard deviation 1 (where these parameters are interpreted on the log-scale) and scaling parameter $c=100$. This realization was simulated using the stick-breaking method given $L=100$.}
	\label{fig:dirsim}
\end{figure}

\section{Integrability}
\label{app:int}

In each model, the Hessian of the utility function $U(x;\pi)$ is negative definite, implying that the demand function $X(z;\pi)$ is invertible in the second regime (see Proposition 2 in \citeauthor{dobronyi}, \citeyear{dobronyi}). Invertibility can also be analyzed \mbox{using the indi-} fference curves of $U(x;\pi)$, or the Slutsky coefficient $\Delta_x(z)$. As shown \mbox{below, this anal-} ysis can, for instance, yield new properties of the moment generating function (MGF).

\subsection{SARA Model}

Suppose $A_1$ and $A_2$ are independent. Let $\Psi_j$ denote the Laplace transform of $A_j$, for $j=1,2$. With this notation, we can write:
\beq
	\log U(x;\pi)=\log\Psi_1(x_1)+\log\Psi_2(x_2).
\eeq
The indifference curve $g_\pi(\cdot,u)$ associated with $U(x;\pi)$ is obtained by solving:
\beq
	\log\Psi_1(x_1)+\log\Psi_2(x_2)=\log u,
\eeq
for $x_2$. This procedure leads to:
\beq
	g_{\pi}(x_1,u)=\left(\log\Psi_2\right)^{-1}\left(\log u - \log \Psi_1(x_1)\right).
\eeq
In general, demand is invertible if the indifference curves are strictly convex such that:
\beq
	\frac{\partial^2 g_{\pi}(x_1,u)}{\partial x_1^2}>0,
\eeq
for every $x_1>0$, and every attainable $u<0$. When preferences are SARA, we obtain:
\begin{align}
\begin{split}
	\frac{d^2}{dv^2} &\left[\left(\log\Psi_2\right)^{-1}\right]\left(\log u - \log \Psi_1(x_1)\right)\left(\frac{d\log\Psi_1(x_1)}{dx_1}\right)^2 \\
	- \frac{d}{dv} &\left[\left(\log\Psi_2\right)^{-1}\right]\left(\log u - \log \Psi_1(x_1)\right)\frac{d^2\log\Psi_1(x_1)}{dx_1^2}>0,
\end{split}
\end{align}
for every $x_1>0$, and every attainable $u<0$. These inequalities, involving two MGFs, are always satisfied. Consequently, we have derived a new property of the MGF, as described in the introduction of this appendix.

\subsection{SSF Model}

If preferences are SSF, it is rather challenging to derive a closed-form expression for the indifference curve. We can, instead, write the integrability condition using the condition on the bordered Hessian in Lemma 1 in \citeauthor{dobronyi} (\citeyear{dobronyi}), but, for both brevity and exposition, let us simply restrict our attention to the general specification of utility in the example in Section \ref{example:2} and check that integrability holds for any Laplace transform. Because the strict convexity of the indifference curves is equivalent to the strict negativity of the Slutsky coefficient $\Delta_x(z)$, it is sufficient to check whether $\Delta_x(z)$ is strictly negative. We obtain:
\begin{align}
\begin{split}
	\Delta_x(z)&=\frac{\partial X_1(z;\pi)}{\partial p} + X_1(z;\pi) \frac{\partial X_1(z;\pi)}{\partial y} = -\frac{1}{\lambda p^3}\frac{d}{dv}\left(\frac{d\log\Psi}{dv}\right)^{-1}\left(-\frac{1}{p}\right).
\end{split}
\end{align}
It is sufficient to show that:
\beq
\label{positiveforslutsky}
	\frac{d}{dv}\left(\frac{d\log\Psi}{dv}\right)^{-1}\left(-\frac{1}{p}\right),
\eeq
is strictly positive. To do this, consider the following derivatives:
\beq
\begin{gathered}
	\frac{d\log\Psi(v)}{dv}=-\frac{\mathbb{E}[A_1\exp(-A_1v)]}{\mathbb{E}[\exp(-A_1v)]}, \\ \text{and} \; \; \frac{d^2\log\Psi(v)}{dv^2}=\frac{\mathbb{E}[A_1^2\exp(-A_1v)]}{\mathbb{E}[\exp(-A_1v)]}-\left(\frac{\mathbb{E}[A_1\exp(-A_1v)]}{\mathbb{E}[\exp(-A_1v)]}\right)^2=V_{\tilde{\pi}}(A_1)>0,
\end{gathered}
\eeq
where the variance is with respect to the transformed density:
\beq
	\frac{\exp(-A_1v)}{\mathbb{E}[\exp(-A_1v)]}\pi(v).
\eeq
Therefore, $\frac{d\log\Psi}{dv}$ is increasing, and so is its inverse $\left(\frac{d\log\Psi}{dv}\right)^{-1}$. Thus, $\Delta_x(z)$ is negative.

\section{Numerical Optimization}
\label{app:filter}

The optimization problem for filtering can be written as:
\beq
\label{numericaloptimization}
\begin{gathered}
	\min_{\pi_1,\pi_2} \;  (\pi_1-\hat{\pi}_{1})'(\pi_1-\hat{\pi}_{1})+(\pi_2-\hat{\pi}_{2})'(\pi_2-\hat{\pi}_{2}) \\
	\text{s.t.} \; \; \text{MRS restrictions \eqref{mrsrestrictions2}, $e'\pi_1=1$, and $e'\pi_2=1$,}
\end{gathered}
\eeq
where $\pi_1$ and $\pi_2$ are written on a sufficiently large discrete grid for $A_1$ and $A_2$, and $e=(1,\dots,1)'$. This optimization problem can be difficult due to the dimension of the problem. The objective function is minimized with respect to the total number $2J$ of grid points in $\pi_1$ and $\pi_2$, which is intentionally chosen to be very large (at least several hundred), and the number of constraints is $N_m$, where $N_m$ denotes the number of observations with positive consumption $x_{it}$ in segment $\Lambda_m$, which is typically around 1,000. Note, $2J$ has to be larger than $N_m$ for identification. Therefore, it is important to find a tractable algorithm for such a problem.

We can use the fact that the MRS restrictions are bilinear in $\pi_1$ and $\pi_2$. Indeed, these constraints can be written as:
\beq
\label{bilinear:ab}
	A_1(\pi_2)\pi_1=b_1(\pi_2) \; \; \text{or} \; \; A_2(\pi_1)\pi_2=b_2(\pi_1).
\eeq
To illustrate, consider the SARA model, and let $a_{1j}$ and $a_{2j}$, $j=1,\dots,J$, denote the locations of the points in the grids for $A_1$ and $A_2$, respectively. Moreover, let $\pi_1=(\pi_{1j})$ and $\pi_2=(\pi_{2j})$ denote the elementary probabilities on $(a_{1j})$ and $(a_{2j})$, respectively. Under the independence of $A_1$ and $A_2$, the MRS restrictions \mbox{have the form:}
\[
\begin{gathered}
	\sum_{j=1}^J[\pi_{1j}a_{1j}\exp(-a_{1j}x_{i1t})]\sum_{j=1}^J[\pi_{2j}\exp(-a_{2j}x_{i2t})] \\ -p_{it}\sum_{j=1}^J[\pi_{1j}\exp(-a_{1j}x_{i1t})]\sum_{j=1}^J[\pi_{2j}a_{2j}\exp(-a_{2j}x_{i2t})]=0,
\end{gathered}
\]
for every $i\in\Lambda_m$ and every $t$ with $x_{it}>0$. The closed-form expressions for $A_1(\pi_2)$, $b_1(\pi_2)$, $A_2(\pi_1)$, and $b_2(\pi_1)$ can be easily deduced. The unit mass restrictions can also be explicitly written as:
\[
	\sum_{j=1}^J\pi_{1j}=1 \; \; \text{and} \; \; \sum_{j=1}^J\pi_{2j}=1.
\]

The equivalent expressions in \eqref{bilinear:ab} can be used to solve the optimization problem in \eqref{numericaloptimization} by using a succession of optimization problems with smaller dimensions (see \citeauthor{bilinear}, \citeyear{bilinear}, and \citeauthor{bilinear-vanrosen}, \citeyear{bilinear-vanrosen}). Precisely, let $\pi_1(k)$ and $\pi_2(k)$ denote the solutions for $\pi_1$ and $\pi_2$ at the $k^{th}$ step of the optimization algorithm. Given $\pi_2(k)$, $\pi_1(k+1)$ is defined as the solution to:
\beq
\label{eq:smallopt1}
	\min_{\pi_1} \; (\pi_1-\hat{\pi}_1)'(\pi_1-\hat{\pi}_1) \; \; \text{s.t. $A_1[\pi_2(k)]\pi_1=b_1[\pi_2(k)]$ and $e'\pi_1=1$},
\eeq
and, similarly, $\pi_2(k+1)$ is defined as the solution to:
\beq
	\min_{\pi_2} \; (\pi_2-\hat{\pi}_2)'(\pi_2-\hat{\pi}_2) \; \; \text{s.t. $A_2[\pi_1(k+1)]\pi_2=b_2[\pi_1(k+1)]$ and $e'\pi_2=1$}.
\eeq
If this algorithm numerically converges, then the limit is the solution to the original optimization problem in \eqref{numericaloptimization}. Moreover, $\pi_1(k)$ and $\pi_2(k)$ have closed-form solutions:

\begin{proposition}
	The solution to \eqref{eq:smallopt1} is equal to:
	\beq
	\label{solution}
		\pi_1(k+1)=\hat{\pi}_1+A_1^*[\pi_2(k)]'\big\{A_1^*[\pi_2(k)]A_1^*[\pi_2(k)]'\big\}^{-1}\big\{b_1^*[\pi_2(k)]-A_1^*[\pi_2(k)]\hat{\pi}_1\big\},
	\eeq
	where $A_j^*$ and $b_j^*$ encompass the MRS constraint and the unit mass contraint together.
\end{proposition}
\begin{proof}
	The optimization problem in \eqref{numericaloptimization} is of the following type:
	\[
		\min_w \; (w-w_0)'(w-w_0) \; \; \text{s.t.} \; \; Aw=b.
	\]
	Let us introduce a Lagrange multiplier $\lambda$. The first-order conditions are, then:
	\beq
	\label{system}
		2(w-w_0)-A'\lambda=0 \; \; \text{and} \; \; Aw=b.
	\eeq
	The first condition can be written as:
	\beq
	\label{w}
		w=w_0+\frac{1}{2}A'\lambda.
	\eeq
	By plugging this expression for $w$ into the second condition, we obtain:
	\beq
	\label{lambda2}
		\frac{\lambda}{2}=(AA')^{-1}(b-Aw_0).
	\eeq
	Together, \eqref{w} and \eqref{lambda2} imply:
	\[
		w=w_0+A'(AA')^{-1}(b-Aw_0).
	\]
	This expression is exactly the form of the solution in the statement of this proposition.
\end{proof}

\begin{remark}
	Instead of minimizing the $\ell_2$-distance between $\pi_j$ and $\hat{\pi}_j$, we could use an information criterion, as in \citet{kitamura-stutzer}. However, we would no longer obtain a closed-form solution for $\pi_1(k)$ and $\pi_2(k)$, and we would have to solve a non-linear system in $\lambda$ with dimension $N_m$.
\end{remark}

\begin{remark}
	The inversion of $AA'$ is numerically feasible, but can be made more robust numerically by including a regularization. In particular, it can be replaced with the inversion of $AA'+\ve I$, where $\ve>0$ is a small regularization parameter. This regularization by shrinkage (see, for example, \citeauthor{ledoit}, \citeyear{ledoit}) is preferable to the machine learning practice which replaces $AA'$ with the diagonal matrix made up of the diagonal elements of $AA'$. In practice, it can also be easier to solve the system in \eqref{system}, instead of using \eqref{solution}.
\end{remark}

\begin{remark}
	The optimization problem in \eqref{numericaloptimization} has not explicitly accounted for the positivity of $\pi_1$ and $\pi_2$. We can incorporate positivity by adjusting after each step of the algorithm.
\end{remark}

\section{The Nielsen Database}
\label{app:d}

In this appendix, we provide more information about the Nielsen Homescan Consumer Panel (NHCP). First, we describe the individual records, then the representativeness of our restricted sample.

\subsection{Individual Records}

As mentioned in the text, all purchases are continuously recorded by each consumer. The left panel in Figure \ref{fig:expdaily} displays the daily (total and alcohol-specific) expenditures of a given consumer in October of 2016. During this month, this consumer purchased 166 units of 97 distinct goods (prior to aggregation). The right panel in Figure \ref{fig:expdaily} displays the daily number of units purchased by this consumer. These \mbox{purchases were all} made at three distinct retailers. 

\begin{figure}
	\centering
	\begin{subfigure}[b]{0.48\linewidth}
      \centering
      \includegraphics[scale=0.4]{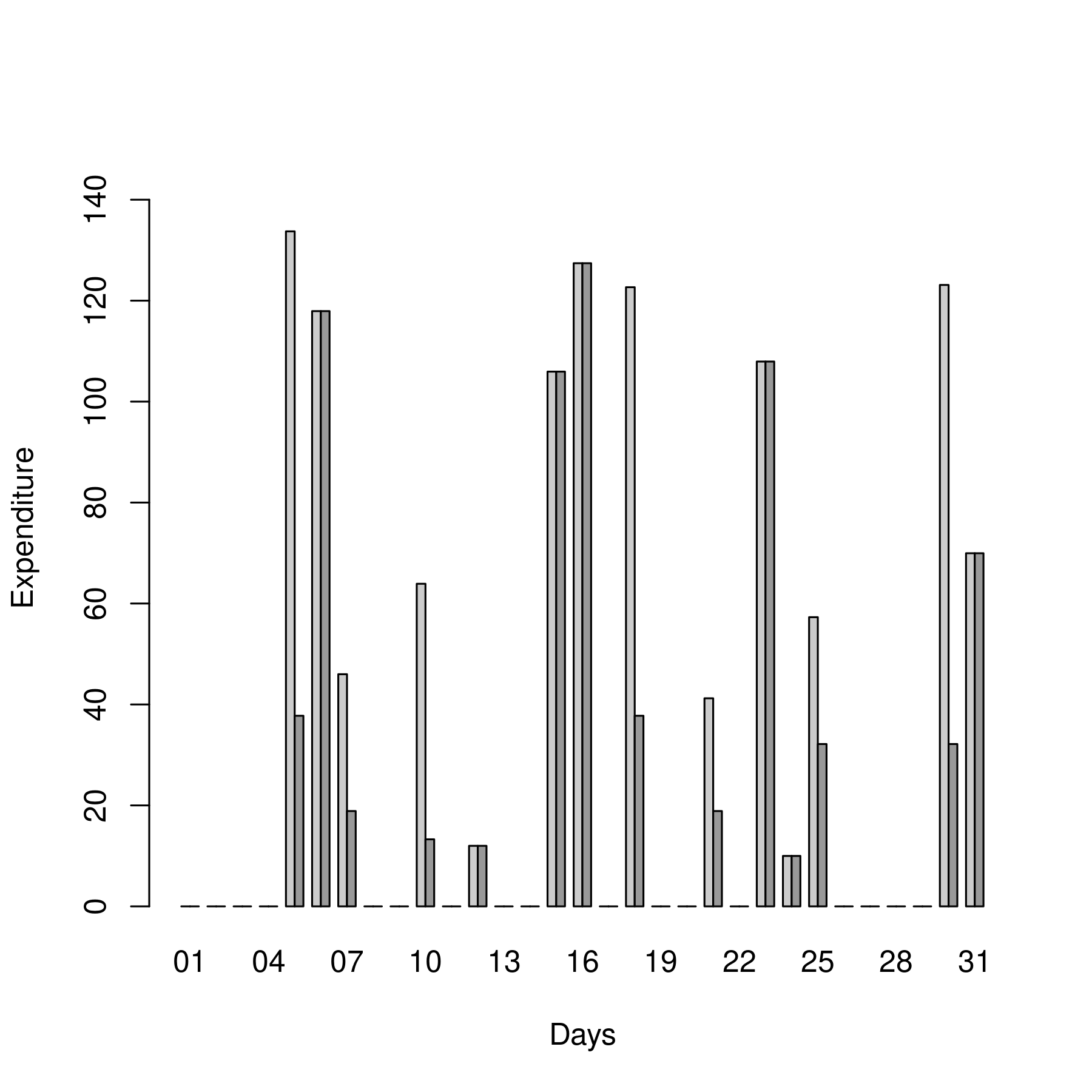}
      \end{subfigure}
      \begin{subfigure}[b]{0.48\linewidth}
      \centering
      \includegraphics[scale=0.4]{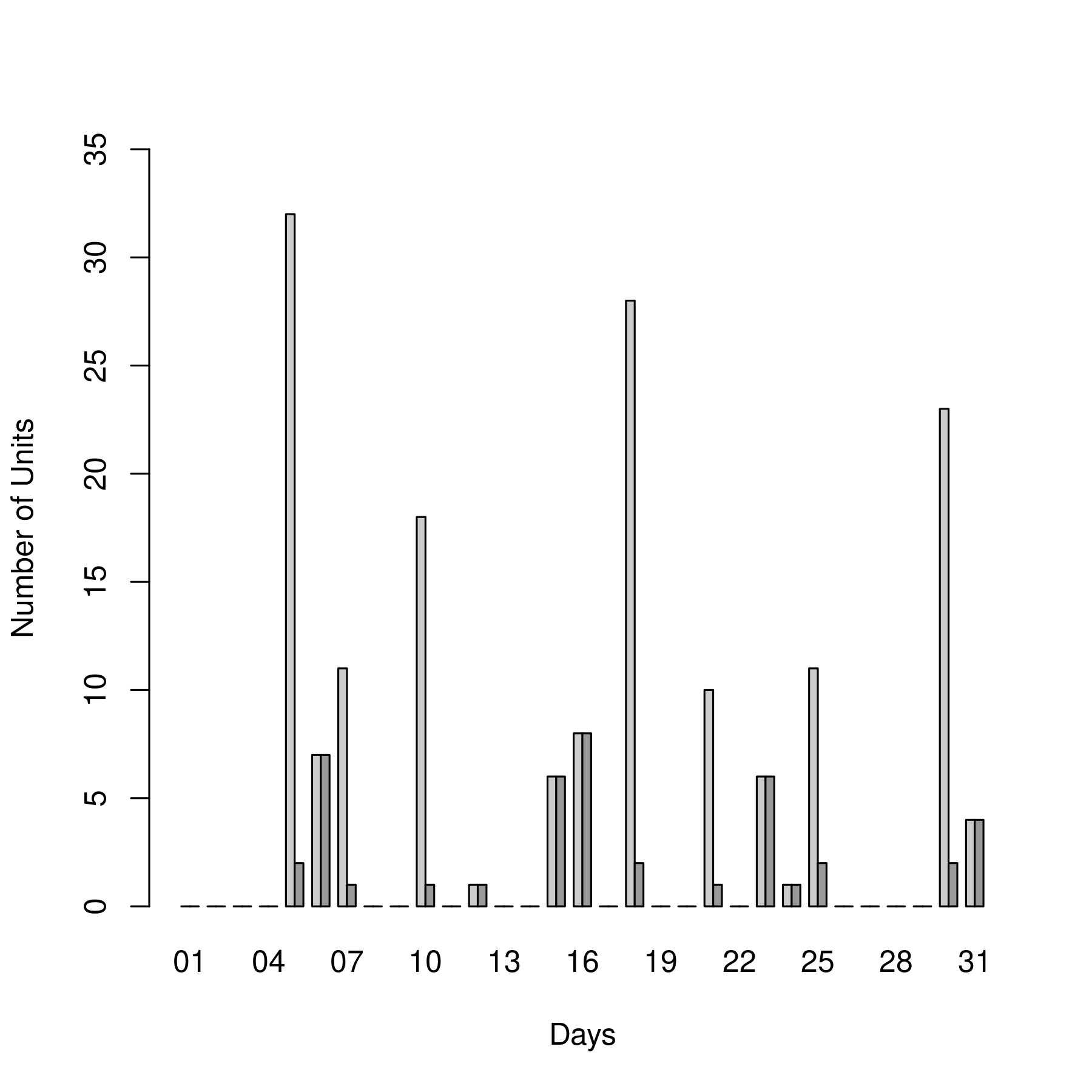}
      \end{subfigure}
	\caption{\textbf{Daily Purchases.} On the left, we illustrate daily expenditure for a single consumer in October of 2016. On the right, we illustrate the number of units purchased by this consumer. Light shaded bars represent all purchases and dark shaded bars represent alcohol-related purchases.}
	\label{fig:expdaily}
\end{figure}

\subsection{Representativeness of Sample}

Let us now report the demographics of the households in our data and compare these demographics with the Current Population Survey (CPS). See \citet{guha-ng} and \citet{dobronyi} for additional summary statistics.

Table \ref{sum:1} gives the distribution of household size in our sample and the CPS. These distributions are similar. Our sample has a slightly smaller proportion \mbox{of households} with a single member, and a slightly larger proportion of households with two members. This difference can be explained by single-member households simply \mbox{buying less} alcohol.

\begin{table}
      \centering
      \caption{Household size in our sample and in the 2017 Annual Social and Economic Supplement (ASEC) of the CPS. CPS numbers are in thousands.}
      \begin{tabular}{crrrr}
            \midrule \midrule
                 & \multicolumn{2}{c}{Sample} & \multicolumn{2}{c}{CPS} \\ \cmidrule{2-5}
            Size & Number & Proportion & Number & Proportion \\ \midrule
            1  & 5,862  & 0.2090 & 35,388 & 0.2812 \\
            2  & 12,768 & 0.4554 & 42,785 & 0.3400 \\
            3  & 4,121  & 0.1469 & 19,423 & 0.1543 \\
            4  & 3,395  & 0.1210 & 16,267 & 0.1292 \\
            5  & 1,288  & 0.0459 & 7,548  & 0.0599 \\
            6  & 422    & 0.0150 & 2,813  & 0.0223 \\
            7+  & 180   & 0.0064 & 1,596  & 0.0126 \\ \midrule
            Total & 28,036 & 1.0000 & 125,819 & 1.0000 \\ \midrule \midrule
      \end{tabular}
      \label{sum:1}
\end{table}

Table \ref{sum:2} describes the distribution of household income in our sample and the CPS. Once again, these two distributions are quite similar, but our sample has a higher proportion of households earning between \$70,000 and \$99,999.

\begin{table}
      \centering
      \caption{Annual household income in our sample and in the 2017 Annual Social and Economic Supplement (ASEC) of the CPS. CPS numbers are in thousands.}
      \begin{tabular}{crrrr}
            \midrule \midrule
                 & \multicolumn{2}{c}{Sample} & \multicolumn{2}{c}{CPS} \\ \cmidrule{2-5}
            Income & Number & Proportion & Number & Proportion \\ \midrule
            Under \$5,000         & 265   & 0.0094 & 4,138 & 0.0327 \\
            \$5,000 to \$9,999    & 274    & 0.0097 & 3,878 & 0.0307 \\
            \$10,000 to \$14,999  & 638   & 0.0227 & 6,122 & 0.0485 \\
            \$15,000 to \$19,999  & 694   & 0.0247 & 5,838 & 0.0462 \\
            \$20,000 to \$24,999  & 1,147   & 0.0409 & 6,245 & 0.0494 \\ 
            \$25,000 to \$29,999  & 1,282   & 0.0457 & 5,939 & 0.0470 \\
            \$30,000 to \$34,999  & 1,480  & 0.0527 & 5,919 & 0.0468 \\
            \$35,000 to \$39,999  & 1,432   & 0.0510 & 5,727 & 0.0453 \\
            \$40,000 to \$44,999  & 1,449   & 0.0516 & 5,487 & 0.0434 \\
            \$45,000 to \$49,999  & 1,637   & 0.0583 & 5,089 & 0.0403 \\
            \$50,000 to \$59,999  & 2,878   & 0.1026 & 9,417 & 0.0746 \\
            \$60,000 to \$69,999  & 2,380   & 0.0848 & 8,213 & 0.0650 \\
            \$70,000 to \$99,999  & 6,459  & 0.2303 & 19,249 & 0.1524 \\
            \$100,000+            & 6,021  & 0.2147 & 34,963 & 0.2769 \\ \midrule
            Total & 28,036 & 1.0000 & 126,224 & 1.0000 \\ \midrule \midrule
      \end{tabular}
      \label{sum:2}
\end{table}

Tables \ref{sum:3:eldest} gives the distribution of the age of the eldest head of the household in our sample and the age of the householder in the CPS. There is no direct comparison for these statistics, as the eldest head may differ from the householder. This aspect of the data can explain why our sample seems to be older than the \mbox{general population.}

\begin{table}
      \centering
      \caption{Age of eldest household head in our sample and the householder in the 2017 Annual Social and Economic Supplement (ASEC) of the CPS. CPS numbers are in thousands.}
      \begin{tabular}{crrrr}
            \midrule \midrule
                 & \multicolumn{2}{c}{Sample} & \multicolumn{2}{c}{CPS} \\ \cmidrule{2-5}
            Age & Number & Proportion & Number & Proportion \\ \midrule
            Under 20  &  4  & 0.0001 & 753    & 0.0059 \\
            20 to 24  &  54  & 0.0019 & 5,608  & 0.0445 \\
            25 to 29  &  476  & 0.0169 & 9,453  & 0.0751 \\
            30 to 34  &  1,201  & 0.0428 & 10,594 & 0.0842 \\
            35 to 39  &  1,817  & 0.0648 & 10,651 & 0.0846 \\ 
            40 to 44  &  1,893  & 0.0675 & 10,571 & 0.0840 \\
            45 to 49  &  2,398  & 0.0855 & 11,115 & 0.0883 \\
            50 to 54  &  3,058  & 0.1090 & 12,180 & 0.0968 \\
            55 to 64  &  7,869  & 0.2806 & 23,896 & 0.1899 \\
            65 to 74  &  6,507  & 0.2320 & 17,551 & 0.1394 \\
            75+       &  2,759  & 0.0984 & 13,448 & 0.1068 \\ \midrule
            Total     & 28,036 & 1.0000 & 125,819 & 1.0000 \\ \midrule \midrule
      \end{tabular}
      \label{sum:3:eldest}
\end{table}

There may also exist another source of non-representativeness: A consumer might behave differently because she is being observed. For example, she might increase her expenditure to give the impression that she she is richer. This type of behaviour can be observed when the period of observation is short, but is not usually sustainable in the long term. This effect should be negligible over the four months considered in the illustration in Section \ref{sec:illustration}.

\section{SARA Model with Taste Dependence}
\label{app:dependence}

Consider the SARA model with taste dependence described in Section \ref{sec:conclusion}. \mbox{In this mod-} el, dependence is introduced using a ``common'' stochastic taste parameter such that:
\beq
	U(x;\pi)=-\mathbb{E}_{\pi}\big[\exp(-(A_c+A_1)x_1-(A_c+A_2)x_2)\big],
\eeq
where $A_c$, $A_1$, $A_2$ are independent non-negative taste parameters with distributions $\pi_c$, $\pi_1$, and $\pi_2$. This utility function can be written in terms of the Laplace transforms of these taste parameters:
\beq
	U(x;\pi)=-\Psi_c(x_1+x_2)\Psi_1(x_1)\Psi_2(x_2),
\eeq
where $\Psi_c$, $\Psi_1$, and $\Psi_2$ are the Laplace transforms of $A_c$, $A_1$, and $A_2$. Therefore, its marginal rate of substitution has the form:
\beq
\label{eq:mrs:common}
	\text{MRS}(x;\pi)=\frac{d\log \Psi_c(x_1+x_2)/dx+d\log \Psi_1(x_1)/dx}{d\log \Psi_c(x_1+x_2)/dx+d\log \Psi_2(x_2)/dx}.
\eeq
Let us now discuss the possibility to identify the distributions of $A_c$, $A_1$, and $A_2$ from the knowledge of the utility function (not observable), and then from the knowledge of the marginal rate of substitution (which can be obtained by inverting the demand).

\subsection{Identification from the Utility Function}

We first ask whether the knowledge of the utility function is equivalent to the knowledge of the distributions $\pi_c$, $\pi_1$, and $\pi_2$. If the utility function $U(x;\pi)$ is \mbox{known, then:}
\beq
	\log (-U(x;\pi))=\log \Psi_c(x_1+x_2)+\log\Psi_1(x_1)+\log\Psi_2(x_2),
\eeq
is known. By taking the cross-derivative of this expression with respect to $x_1$ and $x_2$, we can also obtain knowledge of:
\beq
	\frac{d^2\log\Psi_c(x_1+x_2)}{dx^2},
\eeq
for any $x_1,x_2>0$. Consequently, $\log \Psi_c(x)$ is known up to an affine function $\alpha x + c$. Moreover, since $\log \Psi_c(0)=0$, we can identify $\log \Psi_c(x)$ up to a multiplicative factor $\alpha$. Equivalently, $A_c$ can be replaced with $A_c-\alpha$, and $A_j^*$ can be replaced with $A_j^*+\alpha$, for $j=1,2$, without changing the utility function. This reasoning leads to the follow- ing result:

\begin{proposition}
\label{app:prop:supp}
	If preferences are SARA with a common stochastic taste parameter, $A_c$, $A_1$, and $A_2$ are independent, and the distributions $\pi_c$, $\pi_1$, and $\pi_2$ have support $(0,\infty)$, then these distributions are identified from the observation of utility $U(x;\pi)$.
\end{proposition}
\begin{proof}
	If $\pi_c$, $\pi_1$, and $\pi_2$ have support $(0,\infty)$, then $\alpha=0$. The identification follows.
\end{proof}

Proposition \ref{app:prop:supp} shows that a condition on the supports of the taste distributions is needed for identification.

\subsection{Identification Issue}

Let us now consider the possibility of two distinct sets of taste distributions resulting in the same preference ordering (equivalently, the same marginal rate of substitution).

\begin{proposition}
	If preferences are SARA with a common stochastic taste parameter, and $A_c$, $A_1$, and $A_2$ are independent, then $(\Psi_c,\Psi_1,\Psi_2)$ and $(\Psi_c^{\nu},\Psi_1^{\nu},\Psi_2^{\nu})$ lead to the same preference ordering, for all positive scalars $\nu>0$.
\end{proposition}
\begin{proof}
	The proof is a direct consequence of the expression for $\text{MRS}(x;\pi)$ in \eqref{eq:mrs:common}.
\end{proof}

This type of identification issue has already been encountered in the SARA model with independent taste parameters, as described in Section \ref{within:SARA}. It is not surprising that we have a similar result in this model.

\subsection{Recursive Case}

We are left with the question: Is the identification issue described above the only type of issue that we encounter in this model? First, let us consider the case in which $A_1=0$. In this case, the total taste parameter for good 1 is $A_c$, and the total taste parameter for good 2 is $A_c+A_2$. Therefore, the consumer's risk aversion for \mbox{drinks wi-} th high ABV is systematically larger than her risk aversion for drinks with low ABV.

\begin{proposition}
	If preferences are SARA with a common stochastic taste parameter, $A_1=0$, $A_c$ and $A_2$ are independent, and the distributions $\pi_c$ and $\pi_2$ have support $(0,\infty)$, then these distributions are identified up to a power transform of their Laplace transforms.
\end{proposition}
\begin{proof}
	If $A_1=0$, then the knowledge of the MRS implies the knowledge of:
	\beq
		\frac{d\log\Psi_c(x_1+x_2)/dx}{d\log \Psi_2(x_2)/dx}.
	\eeq
	By considering $x_2=0$, we have to solve the equation:
	\beq
		\frac{d\log\Psi_c(x_1)}{dx}=\frac{d\log\Psi_2(0)/dx}{d\log\Psi_2^*(0)/dx}\cdot \frac{d\log\Psi_c^*(x_1)}{dx}
	\eeq
	Therefore, there is a positive scalar $\nu$ such that:
	\beq
		\frac{d\log\Psi_c(x_1)}{dx}=\nu \frac{d\log\Psi_c^*(x_1)}{dx},
	\eeq
	and the result follows by integration, using $\Psi_c(0)=\Psi_c^*(0)=1$.

\end{proof}

\subsection{General Case}

Let us now consider the general case. The proof of the following result uses the limiting behaviour of the Laplace transform of a positive random variable as $x$ \mbox{tends to infinity.}

\begin{proposition}
	If preferences are SARA with a common stochastic taste parameter, $A_c$, $A_1$, and $A_2$ are independent, and the distributions $\pi_c$, $\pi_1$, and $\pi_2$ have support $(0,\infty)$, then $(\Psi_c,\Psi_1,\Psi_2)$ and $(\Psi_c^*,\Psi_1^*,\Psi_2^*)$ lead to the same preference ordering if, and only if, for some $\nu>0$, we have:
	\beq
		\Psi_c=(\Psi_c^*)^{\nu}, \; \; \Psi_1=(\Psi_1^*)^{\nu}, \; \; \text{and} \; \; \Psi_2=(\Psi_2^*)^{\nu}.
	\eeq
\end{proposition}
\begin{proof}
Consider the following two steps to identification:
\bi
	\item For any positive variable $A$, the logarithmic derivative of its Laplace transform $\Psi(x)$ equals:
	\beq
		\frac{d\log\Psi(x)}{dx}=-\frac{\mathbb{E}[A\exp(-Ax)]}{\mathbb{E}[\exp(-Ax)]}=\mathbb{E}_{Q_x}[A],
	\eeq
	where $Q_x$ is the deformed probability distribution with density:
	\beq
		\frac{\exp(-Ax)}{\mathbb{E}[\exp(-Ax)]},
	\eeq
	with respect to the distribution of $A$. This deformed distribution tends to the positive point mass at $0$, and $d\log\Psi(x)/dx$ tends to 0, when $x$ tends to infinity, under the assumption that $A$ has full support.\footnote{Without this assumption, we get $\lim_{x\rightarrow\infty}\frac{d\log\Psi(x)}{dx}=\text{ess inf}\, A$, where $\text{ess inf}$ denotes the essential infimum of the distribution of $A$ (see Theorem 13.2.5 and Remark 13.3 in \citeauthor{informationtheory}, \citeyear{informationtheory}).} This result can be applied to $A_c$, $A_1$, and $A_2$.

	\item We can use the MRS in \eqref{eq:mrs:common} to identify:
	\beq
		\frac{d\log\Psi_1(x_1)/dx-d\log\Psi_2(x_2)/dx}{d\log\Psi_c(x_1+x_2)/dx}.
	\eeq
	Therefore, we can look for the solution(s) to the equality:
	\begin{align}
	\begin{split}
		&\left(\frac{d\log\Psi_1(x_1)}{dx}-\frac{d\log\Psi_2(x_2)}{dx}\right)\frac{d\log\Psi_c^*(x_1+x_2)}{dx} \\
		=&\left(\frac{d\log\Psi_1^*(x_1)}{dx}-\frac{d\log\Psi_2^*(x_2)}{dx}\right)\frac{d\log\Psi_c(x_1+x_2)}{dx}.
	\end{split}
	\end{align}
	Let us now assume that $x_2$ tends to infinity, and denote $z=x_1+x_2$. We obtain:
	\beq
		\frac{d\log\Psi_1(x_1)}{dx}\frac{d\log\Psi_c^*(z)}{dx}\simeq \frac{d\log\Psi_1^*(x_1)}{dx}\frac{d\log\Psi_c(z)}{dx},
	\eeq
	for any $x_1$ and large $z$. Consequently:
	\beq
		\frac{d\log\Psi_1(x_1)/dx}{d\log\Psi_1^*(x_1)/dx}\sim \frac{d\log\Psi_c(z)/dx}{d\log\Psi_c^*(z)/dx}.
	\eeq
	Since the left-hand side of this equivalence is fixed for $z$ tending to infinity, this limit exists:
	\beq
		\lim_{z\rightarrow\infty}\frac{d\log\Psi_c(z)/dx}{d\log\Psi_c^*(z)/dx}=\frac{d\log\Psi_1(x_1)/dx}{d\log\Psi_1^*(x_1)/dx}.
	\eeq
	Moreover, by definition, this limit is independent of $x_1$. We deduce that:
	\beq
		\frac{d\log\Psi_1(x_1)}{dx}=\nu\frac{d\log\Psi_1^*(x_1)}{dx},
	\eeq
	for some positive scalar $\nu$.
	\item The same reasoning can be used as $x_1$ tends to infinity. This procedure yields:
	\beq
		\nu=\lim_{z\rightarrow\infty}\frac{d\log\Psi_c(z)/dx}{d\log\Psi_c^*(z)/dx}=\frac{d\log\Psi_2(x_2)/dx}{d\log\Psi_2^*(x_2)/dx}.
	\eeq
	Therefore, we obtain:
	\beq
		\frac{d\log\Psi_c(x)}{dx}=\nu\frac{d\log\Psi_c^*(x)}{dx} \; \; \text{and} \; \; \frac{d\log\Psi_j(x)}{dx}=\nu\frac{d\log\Psi_j^*(x)}{dx},
	\eeq
	for $j=1,2$, and the result in this proposition follows by integration.
\ee
\end{proof}

\end{document}